\PassOptionsToPackage{dvipsnames,svgnames,x11names}{xcolor}
\documentclass[11pt,a4paper]{article}

\usepackage[T1]{fontenc}
\usepackage[utf8]{inputenc}
\usepackage[english]{babel}
\usepackage{iftex}

\usepackage{geometry}
\usepackage{lmodern}
\IfFileExists{microtype.sty}{
  \usepackage{microtype}
  \UseMicrotypeSet[protrusion]{basicmath}
}{}
\IfFileExists{upquote.sty}{\usepackage{upquote}}{}
\ifLuaTeX
  \usepackage{selnolig}
\fi

\usepackage{setspace}
\newcommand{\spacingset}[1]{\renewcommand{\baselinestretch}{#1}\small\normalsize}

\usepackage{amsmath,amssymb,amsthm,mathtools,bm}
\usepackage{dsfont}
\usepackage{bbm}
\usepackage{mathrsfs}

\newcommand{\indep}{\perp\!\!\!\perp}

\newcommand{\E}{\mathbb{E}}
\renewcommand{\P}{\mathbb{P}}
\newcommand{\V}{\mathbb{V}}
\newcommand{\R}{\mathbb{R}}

\newcommand{\N}{\mathbb{N}}

\def\1b{\mathbf{1}}
\newcommand{\lambdag}{\boldsymbol{\lambda}}

\newcommand{\cA}{\mathcal{A}}

\newcommand{\cE}{\mathcal{E}}
\newcommand{\cF}{\mathcal{F}}

\newcommand{\cH}{\mathcal{H}}

\newcommand{\cL}{\mathcal{L}}

\newcommand{\cN}{\mathcal{N}}
\newcommand{\cO}{\mathcal{O}}
\newcommand{\cP}{\mathcal{P}}
\newcommand{\cQ}{\mathcal{Q}}

\newcommand{\cS}{\mathcal{S}}

\newcommand{\bx}{\mathbf{x}}

\theoremstyle{definition}
\newtheorem{definition}{Definition}

\newtheorem{remark}{Remark}

\theoremstyle{plain}
\newtheorem{theorem}{Theorem}
\newtheorem{lemma}{Lemma}
\newtheorem{proposition}{Proposition}
\newtheorem{corollary}{Corollary}

\usepackage{graphicx}
\usepackage{float}
\usepackage{placeins}
\usepackage{needspace}
\usepackage{booktabs,array,longtable,multirow,dcolumn}
\usepackage{threeparttable}
\usepackage{xcolor}
\usepackage{colortbl}
\usepackage{calc}
\usepackage{etoolbox}
\usepackage{xpatch}

\BeforeBeginEnvironment{theorem}{\Needspace{5\baselineskip}}
\BeforeBeginEnvironment{lemma}{\Needspace{5\baselineskip}}
\BeforeBeginEnvironment{proposition}{\Needspace{5\baselineskip}}
\BeforeBeginEnvironment{corollary}{\Needspace{5\baselineskip}}
\BeforeBeginEnvironment{definition}{\Needspace{5\baselineskip}}
\BeforeBeginEnvironment{assumption}{\Needspace{5\baselineskip}}
\BeforeBeginEnvironment{example}{\Needspace{5\baselineskip}}
\BeforeBeginEnvironment{remark}{\Needspace{5\baselineskip}}
\makeatletter
\AtEndEnvironment{theorem}{\@endparpenalty=10000}
\AtEndEnvironment{lemma}{\@endparpenalty=10000}
\AtEndEnvironment{proposition}{\@endparpenalty=10000}
\AtEndEnvironment{corollary}{\@endparpenalty=10000}
\AtBeginEnvironment{corollary}{\postdisplaypenalty=10000}
\xpatchcmd{\proof}{\topsep6\p@\@plus6\p@\relax}{\topsep0pt\partopsep0pt\@beginparpenalty=10000\relax}{}{\PackageError{arxiv}{Could not adjust proof spacing}{}}
\def\maxwidth{\ifdim\Gin@nat@width>\linewidth\linewidth\else\Gin@nat@width\fi}
\def\maxheight{\ifdim\Gin@nat@height>\textheight\textheight\else\Gin@nat@height\fi}
\def\fps@figure{htbp}
\patchcmd\longtable{\par}{\if@noskipsec\mbox{}\fi\par}{}{}
\makeatother
\setkeys{Gin}{width=\maxwidth,height=\maxheight,keepaspectratio}

\usepackage{caption}
\usepackage{subcaption}

\IfFileExists{footnotehyper.sty}{
  \usepackage{footnotehyper}
}{
  \usepackage{footnote}
}
\makesavenoteenv{longtable}

\usepackage{algorithm}
\AfterEndEnvironment{algorithm}{\vspace{-\parskip}}
\usepackage{algpseudocode}
\usepackage{listings}

\usepackage{tikz}
\usetikzlibrary{automata,positioning}

\makeatletter
\floatstyle{ruled}
\@ifundefined{c@chapter}
  {\newfloat{codelisting}{h}{lop}}
  {\newfloat{codelisting}{h}{lop}[chapter]}
\floatname{codelisting}{Listing}

\makeatother

\usepackage{enumitem}
\setlist[enumerate]{topsep=0pt,partopsep=0pt,parsep=0pt,itemsep=3pt plus 1pt minus 1pt,beginpenalty=10000}
\setlist[enumerate,2]{topsep=3pt}
\usepackage{titlesec}

\usepackage{fancyhdr}
\newcommand{\shorttitle}{Agnostic Model-Assisted Estimation with Machine Learning}
\usepackage{natbib}
\bibpunct{(}{)}{;}{a}{,}{,}

\makeatletter
\newcommand{\referencepart}[1]{\def\referencegroup{#1}\def\@extra@b@citeb{.#1}\if@filesw\immediate\write\@auxout{\string\gdef\string\@extra@binfo{.#1}}\fi}
\newcommand{\recordreference}[1]{\edef\referencekey{\zap@space#1 \@empty}\global\@namedef{reference@\referencegroup @\referencekey}{}}
\newcommand{\recordreferences}[1]{\@for\referencekey:=#1\do{\expandafter\recordreference\expandafter{\referencekey}}}
\apptocmd{\NAT@sort@cites}{\expandafter\recordreferences\expandafter{\NAT@cite@list}}{}{\PackageError{arxiv}{Could not track citations}{}}
\pretocmd{\nocite}{\recordreferences{#1}}{}{\PackageError{arxiv}{Could not track uncited references}{}}
\let\originalharvarditem\harvarditem
\newcommand{\selectedharvarditem}[4][]{\selectedreference{#1}{#2}{#3}{#4}}
\long\def\selectedreference#1#2#3#4#5\par{\ifcsname reference@\referencegroup @*\endcsname\originalharvarditem[#1]{#2}{#3}{#4}#5\par\else\ifcsname reference@\referencegroup @#4\endcsname\originalharvarditem[#1]{#2}{#3}{#4}#5\par\fi\fi}
\newcommand{\filterreferences}{\let\harvarditem\selectedharvarditem}
\newcommand{\repeatbibliography}{\@input@{\jobname.bbl}}
\makeatother

\IfFileExists{xurl.sty}{\usepackage{xurl}}{}
\usepackage{hyperref}
\usepackage{bookmark}
\hypersetup{
  pdftitle={Title},
  pdfauthor={Author 1; Author 2},
  pdfkeywords={3 to 6 keywords, that do not appear in the title},
  colorlinks=true,
  linkcolor=blue,
  filecolor=Maroon,
  citecolor=red!70!black,urlcolor=blue!80!black,
  urlcolor=Blue,
  pdfcreator={LaTeX}
}

\newcommand{\anon}{1}

\newcommand{\Black}{\color{black}}

\begin{document}

\begingroup
\referencepart{main}

\if1\anon
{
  \title{\bfseries Agnostic Model-Assisted Estimation with\\[8pt]
    Machine Learning for Survey Data}

  \author{
    Ziming {\scshape An}
    \hspace{.2cm}\\
    Department of Mathematics and Statistics, University of Ottawa\\
    and\\
    Mehdi {\scshape Dagdoug}\\
    Department of Mathematics and Statistics, McGill University\\
    and\\
    David {\scshape Haziza}\\
    Department of Mathematics and Statistics, University of Ottawa\\
    and\\
    Yves {\scshape Till\'e}\\
    Institute of Statistics, University of Neuch\^atel
  }

  \date{}
  \maketitle
}
\fi

\if0\anon
{
  \bigskip
  \bigskip
  \bigskip
  \begin{center}
    {\LARGE\bfseries Agnostic Model-Assisted Estimation with\\[8pt]
      Machine Learning for Survey Data}
  \end{center}
  \medskip
}
\fi

\bigskip
\begin{abstract}
Model-assisted estimation uses prediction rules to improve the efficiency of
estimators of finite population parameters while retaining design-based
inference. Although flexible prediction methods have been considered, existing
theoretical results are largely method-specific. We develop a learner-agnostic
framework that replaces separate analyses for individual learners with general
conditions on the sampling design and prediction error. We connect design-aware
and design-agnostic cross-fitting and characterize the sampling designs under
which they yield conditional independence across folds. Under suitable
conditions, conditional weighting gives exact design-unbiasedness. We establish
first-order equivalence to oracle estimators, leading to design consistency and
asymptotic normality, and clarify when conditional and original inclusion
probabilities yield the same first-order behavior. We propose consistent
variance estimators based on cross-fitted residuals and construct asymptotically
valid confidence intervals. Under additional model and regularity conditions,
we establish asymptotic optimality through attainment of the Godambe--Joshi
lower bound. Simulations show that cross-fitting substantially reduces
finite-sample bias and improves variance estimation and coverage with adaptive
learners.
\end{abstract}


\vfill
\newpage
\spacingset{1.4}

\section{Introduction}\label{sec:intro}

The main goal in survey sampling is to estimate finite population characteristics, such as totals or means, from a sample. In the classical design-based framework, inference relies on the sampling design, while the population values are treated as fixed. The Horvitz--Thompson estimator is design-unbiased under very general conditions but, apart from auxiliary information incorporated into the sampling design through the inclusion probabilities, it does not use auxiliary information at the estimation stage. It can therefore be inefficient when the survey variable is strongly related to available auxiliary variables \citep{sarndal1992,breidtopsomer2017}.

Model-assisted estimation uses a prediction model to exploit this relationship and improve efficiency, while inference remains based on the sampling design. Thus, the model is used to improve precision, not to justify the validity of the methodology. This approach originates in classical regression estimation; see, among others, \citet{cochran1942}, \citet{cassel1976}, \citet{robinsonsarndal1983}, and the monograph of \citet{sarndal1992}. A modern review is provided by \citet{breidtopsomer2017}.

Much of the early literature focused on parametric settings, particularly linear working models, leading to the generalized regression (GREG) estimator. Because linear models may be too restrictive, nonparametric and semiparametric model-assisted methods have been developed to capture more complex relationships. Examples include local polynomial regression \citep{breidt2000local}, B-splines \citep{goga2005reduction}, penalized splines \citep{breidtclaeskensopsomer2005}, additive models \citep{breidtetal2007}, neural networks \citep{montanariranalli2005}, tree-based methods \citep{totheltinge2011,mcconvilletoth2019}, and random forests \citep{dagdouggogahaziza2023}.

Classical linear model-assisted estimators require only the auxiliary values for sampled units and their population totals. In contrast, nonparametric and machine learning-based estimators generally require auxiliary information for every population unit. Although stronger, this requirement is increasingly realistic in applications involving administrative records, registers, or linked data. Such information makes it possible to use flexible prediction methods and potentially achieve substantial efficiency gains \citep{breidtopsomer2017}.

These methods also introduce new inferential challenges. Because the prediction function is estimated from the same sample used to construct the final estimator, the fitted predictions depend on the sampling indicators and cannot be treated as fixed in the theoretical analysis. Estimating the prediction function may also compromise variance estimation: flexible learners can track the sampled observations too closely, producing artificially small residuals and hence underestimated variances. These problems are particularly relevant for highly adaptive learners, which may converge slowly and generally do not admit simple linearizations \citep{opsomer2005selecting}.

Recent work on debiased and double machine learning shows that orthogonal estimating equations and cross-fitting can restore valid large-sample inference in the presence of complex nuisance estimation \citep{chernozhukov2018}. These ideas are increasingly relevant in survey sampling. In particular, \citet{seaman2025debiased} considers data integration involving probability and non-probability samples, while \citet{dagdoug2026machine} study item nonresponse through imputation.

This paper develops a general design-based theory of model-assisted estimation and cross-fitting that is not tied to a particular prediction method. The main contribution is to replace learner-specific analyses with common conditions on the sampling design and prediction error. This allows the same estimation and inference results to be used across learners, without repeating the design-based analysis for each method or prescribing a convergence rate for the prediction rule. A central difficulty is that, under a design-based framework and dependent sampling, splitting the population does not generally make the training and evaluation samples independent. We give a necessary and sufficient characterization of the sampling designs that yield conditional independence across folds under uniform population splitting, and establish the connection with design-aware splitting. These results identify when cross-fitting provides the independence needed for design-based inference. Under suitable conditions, we establish exact design-unbiasedness with conditional inclusion probabilities, first-order equivalence to oracle estimators, design consistency, and asymptotic normality. We also determine when retaining the original inclusion probabilities preserves the first-order behavior of conditional weighting. Finally, we establish consistent variance estimation using cross-fitted residuals and give conditions for asymptotic attainment of the Godambe--Joshi lower bound. Together, these results provide a common foundation for using flexible prediction methods in survey estimation while retaining design-based guarantees.

The remainder of the paper is organized as follows. Section~\ref{sec:setup} introduces the setup and motivates the problem. Section~\ref{sec:crossfit} presents conditional cross-fitting schemes and establishes their properties. Section~\ref{sec:cdtMA} develops a cross-fitted model-assisted estimator based on conditional inclusion probabilities, whereas Section~\ref{Sec:uncond} studies its counterpart based on unconditional inclusion probabilities. Section~\ref{sec:6} compares the two estimators, and Section~\ref{Sec:Simulation} reports the simulation results. Section~\ref{sec:final} concludes with open questions and extensions to finite-population parameters defined through estimating equations. The appendices contain the proofs and supplementary results.

\section{Basic setup and motivation} \label{sec:setup}
Consider a finite population $U=\{1,\dots,N\}$ of size $N$. We are interested in estimating the population mean $\mu=N^{-1}\sum_{k \in U} y_k,$ where $y_k$ is the value of the study variable $Y$ for unit $k \in U$. A sample $S \subset U$ of size $n$ is selected according to a sampling design $p(\cdot)$. Let $\{I_k\}_{k \in U}$ denote the sample inclusion indicators, where $I_k=1$ if $k \in S$ and $I_k=0$ otherwise. The first-order and second-order inclusion probabilities are given by $\pi_k=\mathbb{P}(I_k=1)$ and $\pi_{kl}=\mathbb{P}(I_k=1,I_l=1)$. Provided that $\pi_k>0$ for all $k \in U$, the Horvitz--Thompson estimator
$$
\widehat{\mu}_{\pi}=\frac{1}{N}\sum_{k \in S}\frac{y_k}{\pi_k}
$$
is design-unbiased for $\mu$, that is, $\E_p[\widehat{\mu}_{\pi}]=\mu$, where $\mathbb{E}_p[\cdot]$ denotes the expectation operator with respect to the sampling design. We assume that, for each $k\in U$, the vector $\mathbf{x}_k=(x_{k1},\ldots,x_{kp})^\top$ of $p$ auxiliary variables is known. To exploit this information, consider the working relationship
\begin{equation}\label{eq:mod}
  y_k=m(\mathbf{x}_k)+\epsilon_k,\qquad k\in U,
\end{equation}
where $m:\mathbb{R}^p\to\mathbb{R}$ is an unknown function and
$\mathbb{E}(\epsilon_k\mid\mathbf{x}_k)=0$.

\subsection{Revisiting model-based estimation with machine learning}

A long-standing approach to using auxiliary information in finite population inference is model-based, or prediction-based, estimation. This perspective goes back at least to \citet{brewer1963ratio} and \citet{royall1970finite}; comprehensive accounts are given by \citet{valliant2000finite} and \citet{chambers2012introduction}. The finite population values are viewed as a realization from a superpopulation model, and the unobserved values for nonsampled units are predicted from the observed sample.

Let $\widehat m(\cdot)$ be a prediction rule fitted using $\{(\mathbf{x}_k,y_k):k\in S\}$. The resulting model-based estimator is
$$
\widehat{\mu}_{mb}(\widehat m) = \frac{1}{N} \left\{ \sum_{k\in S}y_k+ \sum_{k\in U\setminus S}\widehat
m(\mathbf{x}_k) \right\}.
$$
In the parametric case, when the assumed model is appropriate, this estimator can be highly efficient. Its validity, however, depends on the model: misspecification or informative sampling may lead to biased predictions and invalid measures of uncertainty. This issue was central to the classical debate between design-based and model-based inference; see, for example, \citet{hansen1983evaluation}.

Model-based estimators have been studied extensively under parametric models, particularly linear regression models. Nonparametric contributions include regression, kernel, and spline-based methods; see \citet{dorfman1992nonparametric}, \citet{chambers1993bias}, \citet{dorfman1993estimators}, \citet{kuk1993kernel}, \citet{zheng2003penalized}, and \citet{rueda2009predictive}. Their theoretical properties are generally established using assumptions and arguments tailored to the particular prediction method.

A fundamental difficulty with machine learning is that the model-based estimator is not Neyman-orthogonal with respect to the prediction rule. Indeed, for any reference rule $m$,
$$
\widehat{\mu}_{mb}(\widehat m)-\widehat{\mu}_{mb}(m) = \frac{1}{N}\sum_{k\in U\setminus S} \left\{\widehat
m(\mathbf{x}_k)-m(\mathbf{x}_k)\right\}.
$$
Thus, prediction error enters the estimator at first order. Equivalently, the design expectation of the estimating equation defining $\widehat{\mu}_{mb}$ has derivative $N^{-1}\sum_{k\in U}(1-\pi_k)h(\mathbf{x}_k)$ in a direction $h$, which is generally nonzero. Without additional structure, the convergence rate of the estimator may therefore inherit the potentially slower rate of the learner, making root-$n$ inference difficult. An elementary example illustrating this phenomenon is given in Section~\ref{app:mb-example} of the Supplementary Material. See also \citet{chernozhukov2018} for a formal discussion of Neyman orthogonality.

This observation motivates the model-assisted approach developed below. Unlike the direct model-based estimator, the model-assisted estimator is Neyman-orthogonal, allowing flexible and slowly converging learners to be used without requiring parametric rates. However, using the same sample to fit the prediction rule and construct the estimator creates an additional dependence that must be accounted for. We address this issue using cross-fitting.

\subsection{Model-assisted estimation: Advantages and obstacles to a general theory}

Let $\widetilde{m}(\cdot)$ denote a deterministic target associated with the chosen learning strategy. It is treated as fixed under the sampling design and need not be the true regression function or the rule obtained by fitting the learner to the full finite population. Rather, it represents a stable target toward which the fitted predictions converge in a finite-population average sense, as formalized later through an $L_2$ stability condition. In classical settings such as linear regression, it may coincide with the usual population-level fit.

If $\widetilde{m}(\mathbf{x}_k)$ were known for every $k\in U$, the oracle model-assisted estimator, also known as the difference estimator, would be
\begin{equation}\label{oracle_pi}
\widehat{\mu}_{\mathrm{ma}}(\widetilde{m},\pi)
=
\frac{1}{N}
\left\{
\sum_{k\in U}\widetilde{m}(\mathbf{x}_k)
+
\sum_{k\in S}
\frac{y_k-\widetilde{m}(\mathbf{x}_k)}{\pi_k}
\right\}.
\end{equation}
Because $\widetilde{m}(\cdot)$ is fixed under the sampling design, $\E_p\{\widehat{\mu}_{\mathrm{ma}}(\widetilde{m},\pi)\}=\mu$, regardless of its predictive accuracy. Thus, the target need not correctly describe the conditional mean of $Y$. Although infeasible, \eqref{oracle_pi} provides a useful benchmark for studying the effect of estimating the prediction rule. In the linear case, a natural target is $\widetilde{m}(\mathbf{x}_k)=\mathbf{x}_k^\top\mathbf{B}$, where $\mathbf{B}$ is the finite population least-squares coefficient, and \eqref{oracle_pi} reduces to the classical difference estimator studied by \citet{cassel1976}.

In practice, $\widetilde{m}(\cdot)$ is replaced by a prediction rule $\widehat{m}(\cdot)$ fitted using $\{(\mathbf{x}_k,y_k):k\in S\}$. The feasible model-assisted estimator is
\begin{equation}\label{maestimator}
\widehat{\mu}_{\mathrm{ma}}(\widehat{m},\pi)
=
\frac{1}{N}
\left\{
\sum_{k\in U}\widehat{m}(\mathbf{x}_k)
+
\sum_{k\in S}
\frac{y_k-\widehat{m}(\mathbf{x}_k)}{\pi_k}
\right\}.
\end{equation}
Under a linear working model, $\widehat{m}(\mathbf{x}_k)=\mathbf{x}_k^\top\widehat{\boldsymbol{\beta}}$, with $\widehat{\boldsymbol{\beta}}$ estimated by weighted least squares, leading to the generalized regression (GREG) estimator. Unlike more flexible estimators, the GREG estimator requires only the population totals of the auxiliary variables.

In a seminal paper, \citet{breidt2000local} established the $\sqrt{n}$-consistency and asymptotic normality of a local polynomial estimator, proposed a consistent variance estimator, and showed that the estimator attains the Godambe--Joshi lower bound. As with the other model-assisted methods cited in the introduction, these results do not require the prediction rule to converge at a parametric rate. This contrasts with direct model-based estimation and reflects the Neyman orthogonality of the model-assisted estimator.

A remaining difficulty is that the same sample is used to fit $\widehat{m}(\cdot)$ and to construct \eqref{maestimator}. Consequently, the fitted predictions depend on the sampling indicators that also enter the estimator. Writing them as $\widehat{m}(\mathbf{x}_k;\mathbf{I})$, where $\mathbf{I}=(I_1,\ldots,I_N)^\top$, gives
$$
\widehat{\mu}_{\mathrm{ma}}(\widehat{m},\pi)-\mu = \frac{1}{N}\sum_{k\in U} \left(\frac{I_k}{\pi_k}-1\right)
\left\{y_k-\widehat{m}(\mathbf{x}_k;\mathbf{I})\right\}.
$$
If the prediction rule were fixed under the sampling design, the expectation of the right-hand side would be zero. With sample-fitted predictions, however,
\begin{equation}\label{bias}
\E_p \left\{
\widehat{\mu}_{\mathrm{ma}}(\widehat{m},\pi)-\mu
\right\}
=
-\frac{1}{N}\sum_{k\in U}
\operatorname{Cov}_p \left\{
\frac{I_k}{\pi_k},
\widehat{m}(\mathbf{x}_k;\mathbf{I})
\right\}.
\end{equation}
Thus, the design bias is driven by the dependence between the fitted predictions and the sampling indicators. These covariance terms may be non-negligible in finite samples and need not vanish asymptotically without suitable restrictions on the learner. A simple counterexample showing that the feasible estimator may fail to be design-consistent is given in Section~\ref{app:ma-counterexample} of the Supplementary Material.

The same dependence affects variance estimation. When the prediction rule is fitted and evaluated on the same observations, the sample residuals, $(y_k-\widehat{m}(\mathbf{x}_k))_{k\in S}$, may be artificially small, particularly for highly adaptive learners. Plug-in variance estimators that treat $\widehat{m}(\cdot)$ as fixed may then underestimate the sampling variance. A general theory for model-assisted estimation with arbitrary learners must therefore account for the dependence created by data reuse.

\subsection{Asymptotic framework}
We consider the asymptotic framework of \citet{isaki1982survey}. Let $\{U_v\}_{v \in \mathbb{N}}$ be an increasing sequence of populations, that is, $U_1 \subset U_2 \subset \cdots$, with corresponding sizes $\{N_v\}_{v \in \mathbb{N}}$. In each population $U_v$, a sample $S_v$ is drawn according to a sampling design in order to estimate the population mean $\mu$. In the asymptotic results, $n_v:=\E_p(|S_v|)=\sum_{k\in U_v}\pi_k$ denotes the expected sample size; it equals the sample size for fixed-size designs. All subsequent steps (e.g., fitting a prediction rule, prediction, and variance estimation) are carried out at each $v \in \mathbb{N}$. The number of covariates $p$ is assumed to remain fixed as $v$ increases. For simplicity, we omit the index $v$ whenever there is no ambiguity. Below, we state a set of assumptions related to the sampling design.

\begin{enumerate}[
  label=({D\arabic*}),
  leftmargin=2.4em
]
\setlength{\abovedisplayskip}{8pt}
\setlength{\belowdisplayskip}{8pt}
\setlength{\abovedisplayshortskip}{6pt}
\setlength{\belowdisplayshortskip}{6pt}

\item\label{D1} The expected sampling fraction satisfies
$$
\lim_{v\to\infty}\frac{n_v}{N_v}=\pi_\star>0.
$$
\item\label{D2} There exist positive constants $\lambda > 0$ and $\lambda^{\star} > 0$ such that, for all $v \in \mathbb{N}$,
$$
\min_{k \in U_v} \pi_k \geq \lambda, \qquad \text{and} \qquad \min_{k,l \in U_v} \pi_{kl} \geq
\lambda^{\star}.
$$
\item\label{D3} Let $\Delta_{kl}:=\pi_{kl}-\pi_k\pi_l$. There exists a positive constant $C>0$ such that $$ \max_{k \neq l \in U_v} |\Delta_{kl}|\leq \dfrac{C}{n_v}.$$
\item\label{D4} Let $D_{t,N_v}$ denote the set of distinct $t$-tuples from $U_v$. Then
$$
\lim_{v \to \infty} \max_{(i,j,k,l) \in D_{4,N_v}} \left| \mathbb{E}_p\big[(I_i I_j - \pi_i \pi_j)(I_k I_l -
\pi_k \pi_l)\big] \right| = 0.
$$

\end{enumerate}

These regularity conditions are standard in the survey sampling literature. Assumption~\ref{D1} states that the expected sample size $n_v$ grows at the same rate as the population size $N_v$. Assumption~\ref{D2} ensures that the first- and second-order inclusion probabilities are bounded away from zero. Assumption \ref{D3} requires that the sampling covariances decrease fast enough to zero. Finally, Assumption~\ref{D4} is commonly used to establish the consistency of the Horvitz--Thompson variance estimator; see \cite{boistard2017functional} for a discussion.

\section{Cross-fitting in survey sampling}\label{sec:crossfit}
As discussed above, the standard model-assisted estimator suffers from the dependence created by using the same observations to fit the prediction rule and construct the estimator. Cross-fitting addresses this issue through honest predictions, obtained from data not reused at the evaluation stage.

Implementing cross-fitting under design-based inference requires care because the sampling indicators $\{I_k\}_{k\in U}$ are generally dependent. For certain sampling designs, \citet{lu2025conditional} propose a construction based on conditional rather than unconditional independence. We present this \emph{design-aware} approach in Section~\ref{Sec:Lu_cf}. Section~\ref{Sec:Proposed_cf} introduces a \emph{design-agnostic} alternative based on uniformly randomized folds. We then relate the two approaches and discuss their implementation under more complex designs.

\subsection{The design-aware cross-fitting method}\label{Sec:Lu_cf}
We first describe the approach of \citet{lu2025conditional}, independently of any specific estimator. The goal is to partition the finite population into folds satisfying conditional independence properties that separate the data used to train a prediction rule from those on which it is evaluated.

Let $\mathcal{P}=\{U_1,\ldots,U_M\}$ be a random partition of $U$ into $M$ disjoint folds, so that $U=\bigcup_{j=1}^M U_j$ and $U_j\cap U_q=\emptyset$ for $j\neq q$. Let $N_j=|U_j|$ and, for a given sample $S$, define $S_j:=S\cap U_j$. Finally, let $\delta_{k,j}=1$ if $k\in U_j$ and $\delta_{k,j}=0$ otherwise,
and let $\boldsymbol{\delta}=\{\delta_{k,j}:k\in U, j=1,\ldots,M\}$ denote the collection of partition indicators. The construction of \citet{lu2025conditional} aims to make the sampling indicators conditionally independent across folds, given the partition. More precisely,
\begin{equation}\label{Lu_cond}
\text{(Ind-$\delta$)}:\qquad
\{I_k\}_{k\in U_j}
\indep
\{I_k\}_{k\in U\backslash U_j}
\mid\boldsymbol{\delta},
\qquad j=1,\ldots,M.
\end{equation}
Thus, conditional on the partition, the sampling indicators in one fold provide no additional information about those in the remaining folds. Examples of the fold constructions proposed by \citet{lu2025conditional} for Bernoulli sampling, simple random sampling without replacement (SRSWOR), and stratified SRSWOR are given in Section~\ref{supp:lu-schemes} of the Supplementary Material.

\begin{definition} \label{def1}
     An $M$-fold cross-fitting scheme is a pair $(\boldsymbol{\delta}, \cA)$ where $\boldsymbol{\delta} $ is a random partition of $U$ into folds $U_1, ..., U_M$, and $\cA$ is a random element satisfying $\sigma(\boldsymbol{\delta}) \subseteq \sigma (\cA).$
    Let $$\pi_{k|\cA} = \P(k \in S  \mid  \cA), \qquad \pi_{kl|\cA} = \P(k,l \in S  \mid  \cA), \qquad \Delta_{kl|\cA} =\pi_{kl|\cA}-\pi_{k|\cA}\pi_{l|\cA}$$ be the first- and second-order inclusion probabilities and design covariances of elements $k,l \in U$. 
\end{definition}

In Definition \ref{def1}, the random element $\cA$ specifies what is conditioned on to obtain independence across folds. In the examples of \citet{lu2025conditional} presented in 
Section~\ref{supp:lu-schemes} of the Supplementary Material, $\cA = \boldsymbol{\delta}$; however, as shown in the next section, $\cA$ need not coincide with
$\boldsymbol{\delta}$: the required conditional independence may hold
after conditioning on a different random element.
 
\begin{definition}\label{def:honest}
    Let $(\boldsymbol{\delta}, \cA)$ be a cross-fitting scheme. We say that $(\boldsymbol{\delta}, \cA)$ is \emph{honest} if 
    \begin{equation}
        \text{(Ind-$\mathcal{A}$)}: \qquad	\{I_k\}_{k\in U_j}\indep \{I_k\}_{k\in U \backslash U_j}\mid  \mathcal{A}, \qquad j=1,\ldots,M.
    \end{equation}
\end{definition}

\subsection{A design-agnostic cross-fitting procedure}\label{Sec:Proposed_cf}
We now describe a simple cross-fitting procedure, implemented at the estimation stage, after the sample has been selected. The key idea is to uniformly at random partition the population $U$, and therefore the observed sample $S$, into folds, and to construct prediction rules using observations outside each fold. Let $\mathcal{P} = \{U_1, ..., U_M\}$ be a random partition of the population with fixed respective sizes $N_1, N_2, ..., N_M$, characterized by $\boldsymbol{\delta}$. For each fold $j=1,\dots,M$, define $S_j=S\cap U_j$, and let
$n_j=\sum_{k \in U_j} I_k$
denote the number of sampled units in $U_j$. We write $\boldsymbol{n}=(n_1,\dots,n_M)$ for the vector of fold-specific sample sizes. This cross-fitting scheme is summarized in the following algorithm. 

\begin{algorithm}[H]
\setstretch{1}
\setlength{\parskip}{4pt}
\setlength{\abovedisplayskip}{6pt}
\setlength{\belowdisplayskip}{6pt}
\setlength{\abovedisplayshortskip}{3pt}
\setlength{\belowdisplayshortskip}{3pt}
\caption{Uniform cross-fitting scheme.}
\label{alg:design-agnostic}
\vspace{1mm}
\textbf{Input:}
Population $U$, sample $S$, and positive integer fold sizes
$N_1,\ldots,N_M$.

\textbf{Do:}
\begin{enumerate}
    \item Independently of $S$, uniformly partition $U$ into
    $U_1,\ldots,U_M$, where $|U_j|=N_j.$

    \item For $j=1,\ldots,M$, set
    $
        S_j=S\cap U_j
        \qquad\text{and}\qquad
        n_j=|S_j|.
   $
\end{enumerate}

\textbf{Output:}
The population partition $\cP=(U_1,\ldots,U_M).$
\end{algorithm}

Unlike the construction in Section \ref{Sec:Lu_cf}, this procedure does not, in general, satisfy the conditional independence property \text{(Ind-$\delta$)} in \eqref{Lu_cond}. However, under an important class of sampling designs, another form of conditional independence holds. Specifically, conditional on $(\boldsymbol{n},\boldsymbol{\delta})$, we have
\begin{equation}\label{conditionindepn}
\text{(Ind-$\{\boldsymbol{n}, \boldsymbol{\delta}\}$)}: \qquad	\{I_k\}_{k\in U_j}\indep \{I_k\}_{k\in U \backslash U_j}\mid \left(\boldsymbol{n},\boldsymbol{\delta}\right), \qquad j=1,\ldots,M.
\end{equation}

We now ask under which sampling designs the uniform cross-fitting scheme satisfies condition \textnormal{(Ind-$\{\boldsymbol{n},\boldsymbol{\delta}\}$)}. To answer this question, we recall two definitions from \citet{tille2006sampling}. Let $\cP(U)$ denote the power set of $U$, and let $\cQ:=\{s\in\cP(U):p(s)>0\}$ be the support of the sampling design $p$. The support $\cQ$ is \emph{symmetric} if, for every $s\in\cQ$, all subsets $r\subseteq U$ satisfying $|r|=|s|$ also belong to $\cQ$. A sampling design is \emph{exponential} on $\cQ$ if, for some $\boldsymbol{\lambda}\in\mathbb{R}^N$,
$$
p(s;\cQ,\boldsymbol{\lambda}) = \frac{ \exp \left(\sum_{k\in s}\lambda_k\right) }{ \displaystyle
\sum_{r\in\cQ} \exp \left(\sum_{l\in r}\lambda_l\right) }, \qquad s\in\cQ.
$$
An overview of exponential sampling designs is provided in Appendix \ref{exponetialdesign}. A more comprehensive discussion is presented in \citet{tille2006sampling}. Below, we provide additional results regarding these designs, which will be useful for the rest of the article. \Black

\begin{theorem} \label{theo:condIndExp}
    Let $p$ be a sampling design on a symmetric support $\cQ$. Let $M\geq 2$ and let $N_1, ..., N_M \geq 2$ satisfy $\sum_{j=1}^M N_j=N$. Define $$\mathcal{D}_N := \bigg\{ \left(U_1, ..., U_M\right)  : \biguplus_{j=1}^M U_j = U,   |U_j | = N_j, \text{ for } j =1,...,M\bigg\}.$$

    Let $\boldsymbol{\delta}$ be a random partition, independent of $S$, whose distribution has full support on $\mathcal{D}_N$, that is, $\P(\boldsymbol{\delta} = \boldsymbol{d})>0$ for every $\boldsymbol{d} \in \mathcal{D}_N$. The following statements are equivalent:
    \begin{enumerate}
        \item[(i)] Property (Ind-$\{\boldsymbol{n}, \boldsymbol{\delta}\}$) in \eqref{conditionindepn} holds.
        \item[(ii)] For every sample size $n$ such that $\P(|S| = n)>0$, the conditional design given $|S| = n$ is an exponential design on $ \cS_n := \{s \in \cP(U) : |s| = n\}.$
    \end{enumerate}
\end{theorem}
\begin{proof}
    See Appendix \ref{Prooftheo:condIndExp}.
\end{proof}
Theorem~\ref{theo:condIndExp} gives an exact characterization: under the uniform cross-fitting scheme, condition \textnormal{(Ind-$\{\boldsymbol n,\boldsymbol\delta\}$)} holds if and only if, conditional on each attainable sample size $n$, the sampling design is exponential on $\cS_n$. This class includes SRSWOR, Bernoulli sampling, Poisson sampling, and conditional Poisson sampling. Thus, for these designs, a uniformly generated partition provides the required conditional independence.

\subsection{Relationship between design-aware and design-agnostic cross-fitting schemes}

Although it is sufficient to obtain conditional independence in an interesting class of sampling designs, the design-agnostic construction has the drawback of not controlling the fold sample sizes. For example, it is possible to obtain a partition $\mathcal{P}$ such that $n_j =0 $ for some fold $U_j$. On the other hand, the design-aware cross-fitting philosophy, such as the construction proposed for SRSWOR in \cite{lu2025conditional}, allows one to choose how many sampled elements should be in each fold. We restate this procedure for a general fixed-size design.

\begin{algorithm}[H]
\setstretch{1}
\setlength{\parskip}{4pt}
\setlength{\abovedisplayskip}{6pt}
\setlength{\belowdisplayskip}{6pt}
\setlength{\abovedisplayshortskip}{3pt}
\setlength{\belowdisplayshortskip}{3pt}
\caption{Design-aware cross-fitting scheme.}
\label{alg:design-aware}
\vspace{1mm}
\textbf{Input:}
Population $U$, sample $S$, positive integer fold sizes
$N_1,\ldots,N_M$, and prescribed integer sample fold sizes
$n_1,\ldots,n_M$.

\textbf{Do:}

\begin{enumerate}

    \item Partition $S$ uniformly at random into
    $S_1,\ldots,S_M$, where $|S_j|=n_j$.

    \item Independently of the first partition, partition $U\setminus S$ uniformly at random into
    $R_1,\ldots,R_M$, where $|R_j|=N_j-n_j.$

    \item For $j=1,\ldots,M$, set $U_j=S_j\cup R_j.$
    
\end{enumerate}

\textbf{Output:}
The population partition $\cP=(U_1,\ldots,U_M).$

\end{algorithm}

Algorithm \ref{alg:design-aware} happens to be exactly the cross-fitting scheme proposed by \cite{lu2025conditional} for the particular case of completely randomized trials (corresponding to SRSWOR). It happens that, in fact, this procedure produces valid folds satisfying conditional independence in a wider variety of designs, as the next theorem shows.

\begin{theorem}\label{Theo:eqAwareAgno}
Let $\boldsymbol{\delta}_{\rm unif}$ and $\boldsymbol{\delta}_{\rm aware}$ denote the partitions obtained from Algorithm \ref{alg:design-agnostic} and Algorithm \ref{alg:design-aware}, respectively, using the same fold sizes $N_1, \ldots, N_M$. Fix target fold counts $\boldsymbol{n}=(n_1,\ldots,n_M)$, and suppose that Algorithm \ref{alg:design-aware} uses these counts. For $s\in\cP(U)$ and a partition $\boldsymbol{d}$, define $\boldsymbol{n}(s, \boldsymbol{d}) := [|s \cap U_1(\boldsymbol{d})|, \ldots,|s \cap U_M(\boldsymbol{d})|]^\top.$\\
Then, for every $s\in\cP(U)$ and every partition $\boldsymbol{d}$ with fold sizes $N_1,\ldots,N_M$ for which $\P( \boldsymbol{n}(S, \boldsymbol{d}) = \boldsymbol{n})>0$,
$$ \P \left( S = s  \big\rvert  \boldsymbol{\delta}_{\rm aware} = \boldsymbol{d}\right) = \P \left( S=s  \big\rvert  \boldsymbol{\delta}_{\rm unif} = \boldsymbol{d}, \boldsymbol{n}(S, \boldsymbol{d}) = \boldsymbol{n}\right).$$
Consequently, Algorithm \ref{alg:design-aware} satisfies \textnormal{(Ind-$\delta$)} for every choice of $\boldsymbol{n}$ if and only if Algorithm \ref{alg:design-agnostic} satisfies \textnormal{(Ind-$\{\boldsymbol{n},\boldsymbol{\delta}\}$)}, and both schemes yield the same first- and second-order conditional inclusion probabilities.
\end{theorem}
\begin{proof}
    See Appendix \ref{ProofTheo:eqAwareAgno}.
\end{proof}

 For fixed-size sampling designs, Theorems~\ref{theo:condIndExp} and \ref{Theo:eqAwareAgno} show that whenever the uniform scheme satisfies \textnormal{(Ind-$\{\boldsymbol n,\boldsymbol\delta\}$)}, Algorithm~\ref{alg:design-aware} satisfies \textnormal{(Ind-$\delta$)} while allowing the sample fold sizes to be fixed in advance. The construction does not, however, cover conditional independence arising through a different conditioning variable $\cA$.

\subsection{Assumptions on the cross-fitting scheme}

In order to analyze the asymptotic behavior of model-assisted estimators based on cross-fitting, we will require some regularity conditions on the asymptotic behavior of the sequence of cross-fitting schemes as $v \to \infty$. Throughout, the number of folds $M$ is fixed as $v\to\infty$. All probabilities and expectations involving the cross-fitting scheme are taken with respect to the joint distribution of the sample and the random partition.

\begin{enumerate}[
  label=({CF\arabic*}),
  leftmargin=2.4em
]
\setlength{\abovedisplayskip}{8pt}
\setlength{\belowdisplayskip}{8pt}
\setlength{\abovedisplayshortskip}{6pt}
\setlength{\belowdisplayshortskip}{6pt}

\item\label{CF1} There exist constants $0 < c_1 \leq c_2 < 1$ such that, for $j=1,\ldots,M$,
$$
\frac{N_{j,v}}{N_v} \xrightarrow[v \to \infty]{\P} \alpha_j \in [c_1, c_2].
$$

\item Consider the following three versions.

\begin{enumerate}[
  label=({CF\arabic{enumi},\alph*}),
  leftmargin=3.5em,
  labelsep=0.5em,
  align=left
]
\item\label{CF2a} There exist positive constants $\lambda_{cf}>0$ and $\lambda_{cf}^\star>0$ such that
$$
\lim_{v \to \infty}  \P\left( \min_{k\in U_v} \pi_{k|\cA}\geq \lambda_{cf}\right) = 1, \qquad \lim_{v \to
\infty} \P \left( \min_{k\neq l\in U_v}\pi_{kl|\cA}\geq \lambda^\star_{cf}\right)=1.
$$
\item\label{CF2b} First-order conditional inclusion probabilities satisfy 
$$\P \left( \min_{k\in U_v} \pi_{k\rvert \cA}>0\right)  = 1 - o\left( n_v^{-1}\right).$$
\item\label{CF2c} There exists a constant $\lambda_c>0$ such that, for all $v \in \mathbb{N}$, almost surely,
$$
\min_{k \in U_v} \pi_{k|\cA} \geq \lambda_c.
$$
\end{enumerate}

\item\label{CF3} There exists a positive constant $C_1>0$ such that, for all $v \in \N$,
$$
\E_p\left[\max_{k \in U_v}(\pi_{k|\cA}-\pi_k)^2\right]\leq \dfrac{C_1}{n_v}.
$$
    
\item\label{CF4} Define $\Delta_{kl|\cA} = \pi_{kl|\cA} - \pi_{k|\cA} \pi_{l|\cA}$ for $k, l \in U_v$. There exists a constant $C_{\Delta 2}>0$ such that, for all $v \in \N$,
$$
n_v \max_{k \neq l \in U_v} \rvert \Delta_{kl|\cA}\rvert \leq C_{\Delta 2},
$$
almost surely.
\end{enumerate}

Assumption~\ref{CF1} requires that the fold proportions converge to positive limits below one. This must be ensured when choosing the fold sizes in our constructions. Conditions \ref{CF2a} and \ref{CF3} require the conditional inclusion probabilities to inherit the good behavior of the unconditional ones. Assumption~\ref{CF4} guarantees that the conditional sampling covariances decrease to zero at a rate of at least $\mathcal{O}(n_v^{-1})$, and is used in the analysis of the consistency of the point estimator. This is the conditional analog of what is classically used; see, for example, \cite{breidt2000local}.

Some additional assumptions will also be needed for variance estimation. These are the conditional analogs of classical assumptions in the literature; see, e.g., \cite{breidt2000local}.

\begin{enumerate}[
  label=({CF\arabic*}),
  start=5,
  leftmargin=2.8em
]
\setlength{\abovedisplayskip}{8pt}
\setlength{\belowdisplayskip}{8pt}
\setlength{\abovedisplayshortskip}{6pt}
\setlength{\belowdisplayshortskip}{6pt}
\item\label{CF5} Assume that, almost surely,
    $$
\max_{(i,j,k,l) \in D_{4,N_v}} \big\rvert \E_p \left[
\left(I_iI_j-\pi_{ij|\cA}\right)\left(I_kI_l-\pi_{kl|\cA}\right)\big\rvert \cA\right] \big\rvert=o(1).
$$
\item\label{CF6} Let $\pi_+ := \min\limits_{k\in U_v: \pi_{k\rvert \cA} >0} \pi_{k\rvert \cA}$. Set $\pi_+=1$ if all the conditional inclusion probabilities are zero. Assume that there exists $\delta>0$ such that $$\sup_{v \in \N} \E_p \left[ \dfrac{1}{\pi_+^{2+\delta}}\right] <\infty.$$
\end{enumerate}

The following proposition verifies all these assumptions under SRSWOR.

\begin{proposition}\label{propAssumptions}
Suppose that $p$ is SRSWOR satisfying \ref{D1}. Use Algorithm~\ref{alg:design-aware} with deterministic fold sizes and sample counts such that $n_{j,v}\geq1$ for every $j,v$, and
$$
\frac{N_{j,v}}{N_v}\xrightarrow[v \to \infty]{}\alpha_j\in(0,1), \qquad \max_{1\leq j\leq
M}\left|\frac{n_{j,v}}{N_{j,v}}-\frac{n_v}{N_v}\right| =\mathcal{O}(n_v^{-1/2}).
$$
Then the scheme with $\cA=(\boldsymbol{n},\boldsymbol{\delta})$ is honest and satisfies \ref{CF1}--\ref{CF6}, including \ref{CF2a}--\ref{CF2c}.
\end{proposition}
\begin{proof}
See Appendix~\ref{proofpropAssumptions}.
\end{proof}
The allocation condition holds, for example, when $n_{j,v}=n_vN_{j,v}/N_v+\mathcal{O}(1)$. Similar results are expected for Poisson, conditional Poisson, and stratified sampling under suitable conditions on the design and the cross-fitting scheme.

\section{Conditional cross-fitted model-assisted estimation}
\label{sec:cdtMA}

\subsection{Point estimation}

We now define model-assisted estimators based on the cross-fitting procedures introduced above. For each fold $j=1,\ldots,M$, a prediction rule $\widehat{m}_j$ is fitted using the observations $\{(\mathbf{x}_k,y_k):k\in S\backslash S_j\}$ outside the fold, and a model-assisted estimator is computed using the units in $U_j$. Any tuning or variable selection based on the survey variable must also use only the training observations in $S\backslash S_j$. The fold-specific estimators are then aggregated. The prediction rule is assumed to return finite values, including when the training sample is empty. Any randomness used to fit the prediction rules is treated as fixed.

We first introduce the \emph{conditional model-assisted estimator}, which uses the conditional inclusion probabilities $\pi_{k|\cA}:=\P(I_k=1\mid\cA)$ for $k\in U$. For each $j$ with $N_j>0$, define
$$
\widehat{\mu}_j(\widehat{m}_j,\pi_{|\cA}) = \frac{1}{N_j} \left\{ \sum_{k\in U_j}\widehat{m}_j(\mathbf{x}_k) +
\sum_{k\in S_j} \frac{y_k-\widehat{m}_j(\mathbf{x}_k)}{\pi_{k|\cA}} \right\}.
$$
If $N_j=0$, set $\widehat{\mu}_j(\widehat{m}_j,\pi_{|\cA})=0$. The corresponding cross-fitted model-assisted estimator is
\begin{equation}\label{mod_ass_Lu}
\widehat{\mu}_{cf}(\widehat{m},\pi_{|\cA})
=
\sum_{j=1}^M
\frac{N_j}{N}
\widehat{\mu}_j(\widehat{m}_j,\pi_{|\cA}).
\end{equation}

\begin{remark}
The estimator in \eqref{mod_ass_Lu} is well defined with probability one because
$\P(\pi_{k|\cA}>0\text{ for every }k\in S)=1$. This does not imply that $\pi_{k|\cA}>0$ for every $k\in U$. For example, under SRSWOR with the uniform cross-fitting scheme of Algorithm~\ref{alg:design-agnostic}, $\pi_{k|\cA}=n_j/N_j$ for $k\in U_j$. A fold may contain no sampled unit, in which case $n_j=0$ and $\pi_{k|\cA}=0$ throughout that fold. Nevertheless, no division by zero occurs in \eqref{mod_ass_Lu}, because $S_j=\emptyset$ whenever $n_j=0$.
\end{remark}

Although \eqref{mod_ass_Lu} is defined as a weighted combination of fold-specific estimators, it can also be written in the usual model-assisted form. Let $\widehat{m}^{(-k)}$ denote the prediction rule fitted without using the fold containing unit $k$. Then
\begin{equation}\label{SimplifiedVer}
\widehat{\mu}_{cf}(\widehat{m},\pi_{|\cA})
=
\frac{1}{N}
\left\{
\sum_{k\in U}\widehat{m}^{(-k)}(\mathbf{x}_k)
+
\sum_{k\in S}
\frac{y_k-\widehat{m}^{(-k)}(\mathbf{x}_k)}
{\pi_{k|\cA}}
\right\}.
\end{equation}

\begin{proposition}\label{prop:unb}
Assume that the cross-fitting scheme is honest and that $\pi_{k|\cA}>0$ for every $k\in U$ almost surely. Then, $\E_p \left\{
\widehat{\mu}_{cf}(\widehat{m},\pi_{|\cA})
\right\}=\mu.$
\end{proposition}

\begin{proof}
See Appendix~\ref{Proofprop:unb}.
\end{proof}

Proposition~\ref{prop:unb} establishes exact design unbiasedness under \textnormal{(Ind-$\cA$)} and positivity of the conditional inclusion probabilities, regardless of the machine learning method used to construct the predictions. Thus, flexible prediction methods can be combined with exact finite-sample design validity. Traditional model-assisted estimators, including the GREG estimator, generally have a small-sample bias because the prediction rule is estimated from the same sample. Under the stated conditions, sample splitting removes this source of bias.

For ease of exposition, we will generally assume below that the first-order conditional inclusion probabilities are positive. Some of our theoretical results also cover cases where conditional inclusion probabilities may be equal to zero, although this requires additional technical arguments. In practice, cross-fitting schemes ensuring positivity should be preferred.

We next study the asymptotic properties of the conditional model-assisted estimator when the fitted prediction rule $\widehat{m}$ converges to a deterministic target $\widetilde{m}$. This target often corresponds to the prediction rule fitted at the population level, although this interpretation is not required. We impose the following conditions.
\begin{enumerate}[
  label=({M\arabic*}),
  leftmargin=2.5em,
  labelsep=0.5em,
  align=left
]
\setlength{\abovedisplayskip}{8pt}
\setlength{\belowdisplayskip}{8pt}
\setlength{\abovedisplayshortskip}{6pt}
\setlength{\belowdisplayshortskip}{6pt}
\item\label{M1} The fitted prediction rule $\widehat{m}$ converges to its oracle target $\widetilde{m}$ in the following finite-population $L_2$ sense:
$$
\lim_{v\rightarrow\infty} \frac{1}{N_v} \sum_{k\in U_v} \E_p\left[ \left\{ \widehat{m}^{(-k)}(\bx_k) -
\widetilde{m}(\bx_k) \right\}^2 \right] = 0.
$$

\item Let $\widetilde{e}_k:=y_k-\widetilde{m}(\bx_k)$ for $k\in U_v$, and consider the following two moment conditions:
\begin{enumerate}[
  label=({M\arabic{enumi},\alph*}),
  leftmargin=3.5em,
  labelsep=0.5em,
  align=left
]
\item\label{M2a}
The oracle predictor $\widetilde{m}$ satisfies
$$
\limsup_{v\rightarrow\infty} \frac{1}{N_v} \sum_{k\in U_v} \widetilde{e}_k^{2} < \infty.
$$

\item\label{M2b}
The oracle predictor $\widetilde{m}$ satisfies
$$
\limsup_{v\rightarrow\infty} \frac{1}{N_v} \sum_{k\in U_v} \widetilde{e}_k^{4} < \infty.
$$
\end{enumerate}
\end{enumerate}

Write $\widehat{\mu}_{\mathrm{ma}}(\widetilde m,\pi_{|\cA})$ for the estimator in \eqref{SimplifiedVer} with every fitted prediction replaced by $\widetilde m$. Theorem~\ref{th:equivalent1} establishes its first-order equivalence with the fitted estimator.

\begin{theorem}\label{th:equivalent1}
Assume \textnormal{(Ind-$\cA$)}, \ref{D1}, \ref{CF2a}, \ref{CF4}, and \ref{M1}. Then,
$$
\sqrt{n_v} \left\{ \widehat{\mu}_{cf}(\widehat{m},\pi_{|\cA}) -
\widehat{\mu}_{\mathrm{ma}}(\widetilde{m},\pi_{|\cA}) \right\} = o_\P(1).
$$
\end{theorem}

\begin{proof}
See Appendix~\ref{Proofth:equivalent1}.
\end{proof}

The next result shows that the oracle estimator is $\sqrt{n_v}$-consistent for $\mu$. Together, the two results imply the $\sqrt{n_v}$-consistency of $\widehat{\mu}_{cf}(\widehat{m},\pi_{|\cA})$.

\begin{theorem}\label{theo:consCond}
Assume \ref{D1}, \ref{CF2a}, \ref{CF4}, and \ref{M2a}. Then,
$$
\sqrt{n_v} \left\{ \widehat{\mu}_{\mathrm{ma}}(\widetilde{m},\pi_{|\cA}) - \mu \right\} = \mathcal{O}_\P(1).
$$
\end{theorem}

\begin{proof}
See Appendix~\ref{Prooftheo:consCond}.
\end{proof}
\subsection{Variance estimation}

Consider the oracle estimator $\widehat{\mu}_{\mathrm{ma}}(\widetilde{m},\pi_{|\cA})$, and set $\pi_{kk|\cA}=\pi_{k|\cA}$. If $\pi_{k|\cA}>0$ for every $k\in U_v$ almost surely, then the estimator is conditionally design-unbiased given $\cA$. The law of total variance gives
$$
\begin{aligned}
\V_p \left\{
\widehat{\mu}_{\mathrm{ma}}(\widetilde{m},\pi_{|\cA})
\right\}
=
\E_p \left[
\V_p \left\{
\widehat{\mu}_{\mathrm{ma}}(\widetilde{m},\pi_{|\cA})
\mid\cA
\right\}
\right] 
=
\E_p \left[
\frac{1}{N_v^2}
\sum_{k\in U_v}\sum_{l\in U_v}
\Delta_{kl|\cA}
\frac{\widetilde e_k}{\pi_{k|\cA}}
\frac{\widetilde e_l}{\pi_{l|\cA}}
\right],
\end{aligned}
$$
where $\widetilde e_k=y_k-\widetilde m(\bx_k)$. If, in addition, $\pi_{kl|\cA}>0$ for every $k,l\in U_v$ almost surely, a design-unbiased estimator of this variance would be
$$
\frac{1}{N_v^2} \sum_{k\in S_v}\sum_{l\in S_v} \frac{\Delta_{kl|\cA}}{\pi_{kl|\cA}} \frac{\widetilde
e_k}{\pi_{k|\cA}} \frac{\widetilde e_l}{\pi_{l|\cA}}.
$$
Because the oracle residuals are unknown, we replace them with the cross-fitted residuals $\widehat e_{(-k)}:=y_k-\widehat m^{(-k)}(\bx_k)$. These residuals are computed using prediction rules fitted without the fold containing unit $k$. Unlike ordinary in-sample residuals, they are therefore not artificially reduced by overfitting. This leads to
$$
\widehat V_{cf}(\widehat m,\pi_{|\cA}) := \frac{1}{N_v^2} \sum_{k\in S_v}\sum_{l\in S_v}
\frac{\Delta_{kl|\cA}}{\pi_{kl|\cA}} \frac{\widehat e_{(-k)}}{\pi_{k|\cA}} \frac{\widehat
e_{(-l)}}{\pi_{l|\cA}}.
$$
No zero denominator occurs in this sum almost surely, since $I_kI_l=1$ implies $\pi_{kl|\cA}>0$, and hence $\pi_{k|\cA}>0$ and $\pi_{l|\cA}>0$.
This estimator is generally not design-unbiased for either the variance of the feasible estimator or that of the oracle estimator. The next theorem allows the conditional inclusion probabilities to be zero on events whose probabilities decrease sufficiently fast, and establishes consistency for
$\V_p\{\widehat{\mu}_{\mathrm{ma}}(\widetilde m,\pi_{|\cA})\}$.
\begin{theorem}\label{theo:consistencyCondVar}
Assume \ref{CF2a}, \ref{CF2b}, \ref{CF4}, \ref{CF5}, \ref{CF6}, \ref{M1}, and \ref{M2b}. Suppose, moreover, that there exists a constant $c>0$ such that
$$
n_v \V_p \left\{ \widehat{\mu}_{\mathrm{ma}}(\widetilde m,\pi_{|\cA}) \mid\cA \right\} \xrightarrow[v \to
\infty]{\P} c.
$$
Then,
$$
n_v \left[ \widehat V_{cf}(\widehat m,\pi_{|\cA}) - \V_p \left\{ \widehat{\mu}_{\mathrm{ma}}(\widetilde
m,\pi_{|\cA}) \right\} \right] = o_\P(1).
$$
\end{theorem}
\begin{proof}
See Appendix~\ref{Prooftheo:consistencyCondVar}.
\end{proof}
The preceding results can now be combined to establish the asymptotic distribution of the feasible estimator. If the oracle estimator is asymptotically normal, as expected under standard conditions, then the feasible estimator is asymptotically normal as well. Studentized statistics and confidence intervals below are defined arbitrarily when the estimated variance is nonpositive.
\begin{corollary}\label{coro1}
Suppose that the conditions of Theorems~\ref{th:equivalent1}, \ref{theo:consCond}, and \ref{theo:consistencyCondVar} hold. If
$$
\frac{ \widehat{\mu}_{\mathrm{ma}}(\widetilde m,\pi_{|\cA})-\mu }{ \sqrt{ \V_p \left\{
\widehat{\mu}_{\mathrm{ma}}(\widetilde m,\pi_{|\cA}) \right\} } } \xrightarrow[v \to \infty]{\cL} \cN(0,1),
$$
then
$$
\frac{ \widehat{\mu}_{cf}(\widehat m,\pi_{|\cA})-\mu }{ \sqrt{ \widehat V_{cf}(\widehat m,\pi_{|\cA}) } }
\xrightarrow[v \to \infty]{\cL} \cN(0,1).
$$
\end{corollary}

\begin{proof}
See Appendix~\ref{proofcoro1}.
\end{proof}
It follows from Corollary~\ref{coro1} that an asymptotically valid $100(1-\alpha)\%$ confidence interval for $\mu$ is
$$
\mathrm{CI}_{1-\alpha} \left( \widehat{\mu}_{cf}(\widehat m,\pi_{|\cA}), \widehat V_{cf}(\widehat
m,\pi_{|\cA}) \right) := \widehat{\mu}_{cf}(\widehat m,\pi_{|\cA}) \pm z_{1-\alpha/2} \sqrt{ \widehat
V_{cf}(\widehat m,\pi_{|\cA}) }.
$$

\section{Unconditional cross-fitted model-assisted estimation}
\label{Sec:uncond}

\subsection{Point estimation}

The preceding section provides a complete inferential framework based on the conditional inclusion probabilities. Under SRSWOR and conditional Poisson sampling, which are exponential sampling designs with symmetric support, the first- and second-order conditional inclusion probabilities can be obtained. Their calculation and implementation nevertheless require additional work, particularly for variance estimation. It is therefore natural to ask whether the original, unconditional inclusion probabilities $\pi_k=\P(I_k=1)$ can be used instead. Retaining the original inclusion probabilities avoids computing $\pi_{k|\cA}$ and $\pi_{kl|\cA}$ altogether and leads to the estimator
\begin{equation}\label{eq:biasUncond}
\widehat{\mu}_{cf}(\widehat{m},\pi)
=
\frac{1}{N}
\left\{
\sum_{k\in U}\widehat{m}^{(-k)}(\mathbf{x}_k)
+
\sum_{k\in S}
\frac{y_k-\widehat{m}^{(-k)}(\mathbf{x}_k)}{\pi_k}
\right\}.
\end{equation}
Using unconditional inclusion probabilities generally sacrifices exact design unbiasedness. Unbiasedness is retained if $\pi_{k|\cA}=\pi_k$ for every $k\in U_v$ almost surely, in which case the conditional and unconditional estimators coincide. This occurs, for example, under Poisson sampling with a uniform partition and $\cA=\boldsymbol{\delta}$. More generally, under \textnormal{(Ind-$\cA$)}, the design bias of the unconditional estimator is
$$
\operatorname{Bias}_p \left\{ \widehat{\mu}_{cf}(\widehat m,\pi) \right\} = -\frac{1}{N_v} \sum_{k\in U_v}
\frac{1}{\pi_k} \operatorname{Cov}_p \left\{ \widehat m^{(-k)}(\bx_k), \pi_{k|\cA} \right\}.
$$
By comparison, \eqref{bias} shows that the bias of the non-cross-fitted estimator depends on $\operatorname{Cov}_p\{\widehat m(\bx_k),I_k\}$. The distinction is important: $\widehat m(\bx_k)$ is fitted using the same sample and therefore depends directly on $I_k$, whereas $\widehat m^{(-k)}(\bx_k)$ is conditionally independent of $I_k$ given $\cA$. Its remaining dependence on $I_k$ is mediated by $\cA$ and captured by its covariance with $\pi_{k|\cA}$. Thus, although cross-fitting with unconditional inclusion probabilities does not eliminate the bias exactly, it can be expected to make it small in many applications.

To bound the bias, we first write
$$
\operatorname{Cov}_p \left\{\widehat{m}^{(-k)}(\bx_k),\pi_{k|\cA}\right\} =
\operatorname{Cov}_p \left\{\widehat{m}^{(-k)}(\bx_k)-\widetilde{m}(\bx_k),\pi_{k|\cA}-\pi_k\right\},
$$
and then apply the Cauchy--Schwarz inequality twice under \ref{D2} to obtain
\begin{align*}
\left\rvert\operatorname{Bias}_p \left\{\widehat{\mu}_{cf}(\widehat{m},\pi)\right\}\right\rvert
&\leq \frac{1}{\lambda}
\left\{\max_{k\in U_v}\E_p\left[(\pi_{k|\cA}-\pi_k)^2\right]\right\}^{1/2}\\
&\quad\times\left\{\frac{1}{N_v}\sum_{k\in U_v}\E_p\left[\left\{\widehat{m}^{(-k)}(\bx_k)-\widetilde{m}(\bx_k)\right\}^2\right]\right\}^{1/2}.
\end{align*}
The first term on the right-hand side is $\cO(n_v^{-1/2})$ by \ref{CF3} and, under \ref{M1}, the second term is $o(1)$; this shows directly that the bias is negligible at the first-order scale. 

We next show that the unconditional cross-fitted estimator is first-order equivalent to its oracle,
$\widehat{\mu}_{\mathrm{ma}}(\widetilde m,\pi)$. Thus, estimating the prediction rule has no first-order effect. Unlike the corresponding result for the conditionally weighted estimator, this result requires the proximity condition \ref{CF3}.

\begin{theorem}\label{th:equivalent2}
Assume \textnormal{(Ind-$\cA$)}, \ref{D1}, \ref{D2}, \ref{CF3}, \ref{CF4}, and \ref{M1}. Then,
$$
\sqrt{n_v} \left\{ \widehat{\mu}_{cf}(\widehat m,\pi) - \widehat{\mu}_{\mathrm{ma}}(\widetilde m,\pi) \right\}
= o_\P(1).
$$
\end{theorem}
\begin{proof}
See Appendix~\ref{Proofth:equivalent2}.
\end{proof}
\begin{remark}
The proof reveals a trade-off between the proximity of the conditional and unconditional inclusion probabilities and the convergence of the prediction rule. More precisely, the part involving these probability differences is controlled by
$$
\left\{ \sqrt{n_v} \max_{k\in U_v} \left|\pi_{k|\cA}-\pi_k\right| \right\} \left\{ \frac{1}{N_v} \sum_{k\in
U_v} \left[ \widehat m^{(-k)}(\bx_k)-\widetilde m(\bx_k) \right]^2 \right\}^{1/2}.
$$
Consequently, the proximity condition \ref{CF3} may be weakened when the prediction rule converges sufficiently fast.
\end{remark}
Because $\widetilde m$ is deterministic, $\widehat{\mu}_{\mathrm{ma}}(\widetilde m,\pi)$ is a classical difference estimator. Its first-order properties therefore follow from standard design-based arguments: it is $\sqrt{n_v}$-consistent and, under many sampling designs, asymptotically normal.
\begin{theorem}\label{theo:consUnCond}
Assume \ref{D1}--\ref{D3} and \ref{M2a}. Then,
$$
\widehat{\mu}_{\mathrm{ma}}(\widetilde m,\pi)-\mu = \mathcal{O}_\P \left(\frac{1}{\sqrt{N_v}}\right).
$$
\end{theorem}
\begin{proof}
See Appendix~\ref{Prooftheo:consUnCond}.
\end{proof}
Theorems~\ref{th:equivalent2} and \ref{theo:consUnCond} together establish the $\sqrt{n_v}$-consistency of $\widehat{\mu}_{cf}(\widehat m,\pi)$ without imposing a specific convergence rate on $\widehat m$ toward $\widetilde m$. When the oracle central limit theorem holds, the predictive quality of $\widetilde m$ affects the first-order distribution only through the asymptotic variance. Optimality results are given and established in Appendix \ref{App:GJ}.

\subsection{Variance estimation }

Section~\ref{sec:cdtMA} introduced the variance estimator
$\widehat V_{cf}(\widehat m,\pi_{|\cA})$ for the conditional estimator
$\widehat{\mu}_{cf}(\widehat m,\pi_{|\cA})$. The conditional and unconditional point estimators are first-order equivalent to different oracle difference estimators and may therefore have different asymptotic variances. We now construct a variance estimator for $\widehat{\mu}_{cf}(\widehat m,\pi)$. Conditions under which the two point estimators are asymptotically equivalent are given in Section~\ref{sec:6}. Set $\pi_{kk}=\pi_k$, so that $\Delta_{kk}=\pi_k(1-\pi_k)$. Whenever $\pi_{kl}>0$ for every $k,l\in U_v$, the natural unconditional analogue of $\widehat V_{cf}(\widehat m,\pi_{|\cA})$ is
$$
\widehat V_{cf}(\widehat m,\pi) := \frac{1}{N_v^2} \sum_{k\in S_v}\sum_{l\in S_v} \frac{\Delta_{kl}}{\pi_{kl}}
\frac{\widehat e_{(-k)}}{\pi_k} \frac{\widehat e_{(-l)}}{\pi_l}.
$$
This estimator is closely related to that proposed by \citet{dagdouggogahaziza2023}. Under leave-one-out cross-fitting, it reduces to the estimator of \citet{opsomer2005selecting}, who established its properties for local polynomial regression. The leave-one-out version was also recently studied by \citet{bouhadra2026high} for high-dimensional linear models. For general cross-fitting, however, \citet{dagdouggogahaziza2023} investigated only its empirical performance and did not establish consistency. The next result fills this gap for generic prediction rules and an arbitrary number of folds.
\begin{theorem}\label{theoVar}
Assume \ref{D1}--\ref{D4}, \ref{M1}, and \ref{M2b}. Then,
$$
n_v\left\rvert\widehat{V}_{cf}(\widehat{m},\pi)-\V_p \left\{\widehat{\mu}_{\mathrm{ma}}(\widetilde{m},\pi)\right\}\right\rvert=o_\P(1).
$$
\end{theorem}
\begin{proof}
    See Appendix \ref{ProoftheoVar}.
\end{proof}

The following corollary establishes valid inference based on
$\{\widehat{\mu}_{cf}(\widehat m,\pi),
\widehat V_{cf}(\widehat m,\pi)\}$
without requiring first- or second-order conditional inclusion probabilities.

\begin{corollary}\label{coro2}
Suppose that the conditions of Theorems~\ref{th:equivalent2} and \ref{theoVar} hold and that
\begin{equation}\label{eq:varOrder}
\liminf_{v\to \infty} n_v \V_p \left\{\widehat{\mu}_{\mathrm{ma}}(\widetilde{m},\pi)\right\}>0.
\end{equation}
If
$$
\frac{\widehat{\mu}_{\mathrm{ma}}(\widetilde{m},\pi)-\mu}{\sqrt{\V_p \left\{\widehat{\mu}_{\mathrm{ma}}(\widetilde{m},\pi)\right\}}}
\xrightarrow[v \to \infty]{\cL}\cN(0,1),
$$
then
$$
\frac{\widehat{\mu}_{cf}(\widehat{m},\pi)-\mu}{\sqrt{\widehat{V}_{cf}(\widehat{m},\pi)}} \xrightarrow[v \to
\infty]{\cL}\cN(0,1).
$$
\end{corollary}
\begin{proof}
See Appendix~\ref{proofcoro2}.
\end{proof}

\section{Conditional versus unconditional cross-fitting}
\label{sec:6}

Sections~\ref{sec:cdtMA} and \ref{Sec:uncond} introduced the conditional estimator $\widehat{\mu}_{cf}(\widehat m,\pi_{|\cA})$ in \eqref{SimplifiedVer} and the unconditional estimator $\widehat{\mu}_{cf}(\widehat m,\pi)$ in \eqref{eq:biasUncond}. They use the same out-of-fold predictions and differ only in the inclusion probabilities used in the residual correction. Consequently, they are first-order equivalent to different oracle estimators. We now compare their asymptotic properties.

Let
$$
B_{\cA} := \E_p \left[ \widehat{\mu}_{\mathrm{ma}}(\widetilde m,\pi)-\mu \mid\cA \right]
$$
denote the conditional bias of the unconditional oracle estimator. Since this estimator is unconditionally design-unbiased, $\E_p(B_{\cA})=0$. The magnitude of its fluctuations determines whether the conditional and unconditional estimators are first-order equivalent.

\begin{theorem}\label{theo:eqq}
Suppose that \ref{M2a} and the conditions of Theorems~\ref{th:equivalent1} and \ref{th:equivalent2} hold. Then,
$$
\sqrt{n_v} \left\{ \widehat{\mu}_{cf}(\widehat m,\pi) - \widehat{\mu}_{cf}(\widehat m,\pi_{|\cA}) \right\} =
\sqrt{n_v}B_{\cA}+o_\P(1).
$$
Consequently, the two estimators are first-order asymptotically equivalent if and only if $\sqrt{n_v}B_{\cA}=o_\P(1)$.
\end{theorem}

\begin{proof}
See Appendix~\ref{Prooftheo:eqq}.
\end{proof}

Condition \ref{CF3} controls the discrepancy between the conditional and unconditional inclusion probabilities and is used in the expansion in Theorem~\ref{theo:eqq}. It does not, however, imply that $\sqrt{n_v}B_{\cA}=o_\P(1)$. Supplementary Section~\ref{supp:nonequivalence} gives an SRSWOR example for which
$\sqrt{n_v}B_{\cA}\xrightarrow[v \to \infty]{\cL}\cN(0,1-\pi_\star)$, showing that the two estimators need not be first-order equivalent.

When equivalence fails, the following result provides a general comparison of their first-order variances.

\begin{theorem}\label{Thm:reminder}
Assume \ref{D1}--\ref{D2}, \ref{CF2c}, \ref{CF3}, \ref{CF4}, and \ref{M2a}. Then,
$$
\begin{aligned}
n_v\Bigl[
\V_p \left\{
\widehat{\mu}_{\mathrm{ma}}(\widetilde m,\pi)
\right\}
-
\V_p \left\{
\widehat{\mu}_{\mathrm{ma}}(\widetilde m,\pi_{|\cA})
\right\}
\Bigr]
&=
n_v\E_p \left[
\left\{
\widehat{\mu}_{\mathrm{ma}}(\widetilde m,\pi_{|\cA})
-
\widehat{\mu}_{\mathrm{ma}}(\widetilde m,\pi)
\right\}^2
\right]
+o(1),
\end{aligned}
$$
where the leading term is $\mathcal O(1)$.
\end{theorem}

\begin{proof}
See Appendix~\ref{ProofThm:reminder}.
\end{proof}

Since the leading term is nonnegative, Theorem~\ref{Thm:reminder} implies
$$
\liminf_{v\to\infty} n_v \left[ \V_p \left\{ \widehat{\mu}_{\mathrm{ma}}(\widetilde m,\pi) \right\} -
\V_p \left\{ \widehat{\mu}_{\mathrm{ma}}(\widetilde m,\pi_{|\cA}) \right\} \right] \geq 0.
$$
Thus, the conditional oracle is at least as efficient as the unconditional oracle to first order. Under the preceding oracle-equivalence conditions and suitable uniform integrability conditions, this comparison extends to the feasible estimators, and their first-order variances coincide when they are first-order equivalent. Supplementary Section~\ref{supp:balanced-residuals} gives a condition for equivalence under SRSWOR. This first-order comparison does not give a second-order ordering. Appendix~\ref{App:GJ} gives conditions under which the proposed estimators asymptotically attain the Godambe--Joshi lower bound.

\section{Simulation study}\label{Sec:Simulation}

We evaluate the finite-sample performance of
$\widehat{\mu}_{ma}(\widehat{m},\pi)$,
$\widehat{\mu}_{cf}(\widehat{m},\pi^{(\delta)})$,
$\widehat{\mu}_{cf}(\widehat{m},\pi_{|\cA})$, and
$\widehat{\mu}_{cf}(\widehat{m},\pi)$, together with their variance estimators. Point estimators are evaluated in terms of relative bias and efficiency, and variance estimators in terms of relative bias and coverage of Wald confidence intervals.

We generate finite populations of sizes $N\in\{1000,2500,5000\}$. For each unit, the $p=20$ components of $\mathbf{x}_k=(x_{k1},\ldots,x_{kp})^\top$ are generated independently from a normal distribution with mean 5 and variance 2. The survey variable is generated under either the \textit{linear model} $y_k=4x_{k1}+3x_{k2}+3x_{k3}+2x_{k4}+5x_{k5}+7x_{k6}+\varepsilon_k$, with $\varepsilon_k\sim\mathcal N(0,225)$, or the \textit{nonlinear model} $y_k=10+2I(x_{k1}>5)+(x_{k2}+x_{k3}-10)^2+\varepsilon_k$, with $\varepsilon_k\sim\mathcal N(0,9)$. Each population is held fixed throughout the Monte Carlo experiment.

From each population, we draw $R=20{,}000$ independent SRSWOR samples of size $n=0.2N$. Four prediction methods are considered: linear regression; random forests with $B=200$ trees and minimum terminal node size $n_0=5$; $K$-nearest neighbors with $K=5$ and Euclidean distance; and regression trees with minimum terminal node size $n_0=5$.

For each sample, we compute four model-assisted estimators. First, $\widehat{\mu}_{ma}(\widehat{m},\pi)$ in \eqref{maestimator} is the standard estimator without cross-fitting, based on the original inclusion probabilities. Second, $\widehat{\mu}_{cf}(\widehat{m},\pi^{(\delta)})$ in \eqref{mod_ass_Lu} uses the design-aware procedure of Section~\ref{Sec:Lu_cf} and the conditional inclusion probabilities $\pi^{(\delta)}$. We use $M=2$ population folds with $(N_1,N_2)=N(0.5,0.5)$ and consider both $(n_1,n_2)=n(0.5,0.5)$ and $(n_1,n_2)=n(0.8,0.2)$. Third, $\widehat{\mu}_{cf}(\widehat{m},\pi_{|\cA})$ in \eqref{SimplifiedVer} uses the design-agnostic procedure of Section~\ref{Sec:Proposed_cf} (Algorithm \ref{alg:design-agnostic}), with $M=2$ folds and the conditional inclusion probabilities $\pi_{|\cA}$ with $\cA = (\boldsymbol{n}, \boldsymbol{\delta})$. Finally, $\widehat{\mu}_{cf}(\widehat{m},\pi)$ in \eqref{eq:biasUncond} uses the same design-agnostic procedure but retains the original inclusion probabilities $\pi$.

For a generic estimator $\widehat{\mu}_m$ among the four above, let $\widehat{\mu}_m^{(r)}$ denote its value at replication $r$. We consider the Monte Carlo percent relative bias $\mathrm{RB}$, relative efficiency $\mathrm{RE}$ using the Horvitz--Thompson estimator as benchmark, percent relative bias of the variance estimator $\mathrm{RB\text{-}Var}$, and coverage probability $\mathrm{CP}$:
\begin{align*}
\mathrm{RB}(\widehat{\mu}_m)
&=
100\%\times\frac{1}{R}
\sum_{r=1}^R
\frac{\widehat{\mu}_m^{(r)}-\mu}{\mu},\\
\mathrm{RE}(\widehat{\mu}_m)
&=
100\%\times
\frac{\mathbb{E}_{MC} \left[
(\widehat{\mu}_m-\mu)^2
\right]}
{\mathbb{E}_{MC} \left[
(\widehat{\mu}_{\pi}-\mu)^2
\right]},\\
\mathrm{RB\text{-}Var}(\widehat{\mu}_m)
&=
100\%\times
\frac{
\mathbb{E}_{MC} \left[
\widehat{\V}(\widehat{\mu}_m)
\right]
-
\V_{MC}(\widehat{\mu}_m)
}
{\V_{MC}(\widehat{\mu}_m)},\\
\mathrm{CP}(\widehat{\mu}_m)
&=
100\%\times
\mathbb{E}_{MC} \left[
I \left\{
\widehat{\mu}_m
\pm
1.96\sqrt{\widehat{\V}(\widehat{\mu}_m)}
\ni\mu
\right\}
\right].
\end{align*}
Table~\ref{Table 1} shows that all estimators are essentially unbiased under the linear relationship. Under the nonlinear relationship, the standard model-assisted estimator has small but persistent biases with flexible learners, whereas the cross-fitted estimators remain nearly unbiased.

\begin{table}[p]
\centering
\caption{Relative bias and relative efficiency of the estimators under the linear and nonlinear relationships.}
\label{Table 1}

\begingroup
\small
\renewcommand{\arraystretch}{0.65}
\setlength{\tabcolsep}{4.5pt}

\resizebox{\textwidth}{!}{\begin{tabular}{@{}lllcccccc@{}}
\toprule
\multirow{2}{*}{Model}
& \multirow{2}{*}{Method}
& \multirow{2}{*}{}
& \multicolumn{3}{c}{Linear relationship}
& \multicolumn{3}{c}{Nonlinear relationship} \\
\cmidrule(lr){4-6}\cmidrule(lr){7-9}
& & $n$ & 200 & 500 & 1000 & 200 & 500 & 1000 \\
\midrule

\multirow{10}{*}{\shortstack{Linear\\ Regression}}
& \multirow{2}{*}{$\widehat{\mu}_{ma}(\widehat{m},\pi)$}
& RB & 0.0 & 0.0 & 0.0 & -0.4 & -0.2 & -0.1 \\
& & RE & 39.3 & 34.2 & 34.2 & 112.2 & 105.2 & 102.5 \\
\cmidrule(lr){2-9}
& \multirow{2}{*}{$\widehat{\mu}_{cf}(\widehat{m},\pi^{(\delta)})$ (equal)}
& RB & 0.0 & 0.0 & 0.0 & 0.0 & 0.0 & 0.0 \\
& & RE & 44.7 & 35.9 & 34.8 & 129.9 & 110.1 & 105.6 \\
\cmidrule(lr){2-9}
& \multirow{2}{*}{$\widehat{\mu}_{cf}(\widehat{m},\pi^{(\delta)})$ (unequal)}
& RB & 0.0 & 0.0 & 0.0 & 0.0 & 0.0 & 0.0 \\
& & RE & 78.1 & 60.8 & 59.6 & 224.6 & 188.5 & 179.9 \\
\cmidrule(lr){2-9}
& \multirow{2}{*}{$\widehat{\mu}_{cf}(\widehat{m},\pi)$}
& RB & 0.0 & 0.0 & 0.0 & -0.1 & 0.0 & 0.0 \\
& & RE & 44.4 & 35.1 & 34.8 & 133.4 & 110.5 & 105.9 \\
\cmidrule(lr){2-9}
& \multirow{2}{*}{$\widehat{\mu}_{cf}(\widehat{m},\pi_{|\cA})$}
& RB & 0.0 & 0.0 & 0.0 & -0.1 & 0.0 & 0.0 \\
& & RE & 44.1 & 35.2 & 34.9 & 133.2 & 110.6 & 106.0 \\
\midrule

\multirow{10}{*}{\shortstack{Random\\ Forests}}
& \multirow{2}{*}{$\widehat{\mu}_{ma}(\widehat{m},\pi)$}
& RB & 0.0 & 0.0 & 0.0 & -1.1 & -0.9 & -0.7 \\
& & RE & 60.8 & 49.6 & 47.7 & 45.4 & 34.1 & 31.2 \\
\cmidrule(lr){2-9}
& \multirow{2}{*}{$\widehat{\mu}_{cf}(\widehat{m},\pi^{(\delta)})$ (equal)}
& RB & 0.0 & 0.0 & 0.0 & 0.0 & 0.0 & 0.0 \\
& & RE & 68.3 & 54.6 & 51.4 & 59.7 & 36.7 & 25.2 \\
\cmidrule(lr){2-9}
& \multirow{2}{*}{$\widehat{\mu}_{cf}(\widehat{m},\pi^{(\delta)})$ (unequal)}
& RB & 0.0 & 0.0 & 0.0 & 0.0 & 0.0 & 0.0 \\
& & RE & 111.3 & 91.5 & 86.2 & 87.4 & 54.1 & 39.2 \\
\cmidrule(lr){2-9}
& \multirow{2}{*}{$\widehat{\mu}_{cf}(\widehat{m},\pi)$}
& RB & 0.0 & 0.0 & 0.0 & 0.0 & 0.0 & 0.0 \\
& & RE & 68.8 & 54.3 & 51.6 & 59.5 & 37.0 & 25.0 \\
\cmidrule(lr){2-9}
& \multirow{2}{*}{$\widehat{\mu}_{cf}(\widehat{m},\pi_{|\cA})$}
& RB & 0.0 & 0.0 & 0.0 & 0.0 & 0.0 & 0.0 \\
& & RE & 68.9 & 54.3 & 51.7 & 59.4 & 37.0 & 25.0 \\
\midrule

\multirow{10}{*}{KNN}
& \multirow{2}{*}{$\widehat{\mu}_{ma}(\widehat{m},\pi)$}
& RB & 0.0 & 0.0 & 0.0 & -1.7 & -2.1 & -1.7 \\
& & RE & 72.0 & 64.5 & 62.5 & 89.7 & 134.6 & 161.4 \\
\cmidrule(lr){2-9}
& \multirow{2}{*}{$\widehat{\mu}_{cf}(\widehat{m},\pi^{(\delta)})$ (equal)}
& RB & 0.0 & 0.0 & 0.0 & 0.0 & 0.0 & 0.0 \\
& & RE & 77.4 & 69.1 & 65.5 & 88.9 & 84.6 & 80.7 \\
\cmidrule(lr){2-9}
& \multirow{2}{*}{$\widehat{\mu}_{cf}(\widehat{m},\pi^{(\delta)})$ (unequal)}
& RB & 0.0 & 0.0 & 0.0 & 0.0 & 0.0 & 0.0 \\
& & RE & 130.4 & 115.7 & 111.5 & 148.9 & 140.4 & 135.8 \\
\cmidrule(lr){2-9}
& \multirow{2}{*}{$\widehat{\mu}_{cf}(\widehat{m},\pi)$}
& RB & 0.0 & 0.0 & 0.0 & -0.1 & 0.0 & 0.0 \\
& & RE & 78.1 & 68.3 & 65.7 & 89.4 & 84.4 & 80.6 \\
\cmidrule(lr){2-9}
& \multirow{2}{*}{$\widehat{\mu}_{cf}(\widehat{m},\pi_{|\cA})$}
& RB & 0.0 & 0.0 & 0.0 & -0.1 & 0.0 & 0.0 \\
& & RE & 77.6 & 68.3 & 65.6 & 89.5 & 84.4 & 80.8 \\
\midrule

\multirow{10}{*}{\shortstack{Regression\\ Trees}}
& \multirow{2}{*}{$\widehat{\mu}_{ma}(\widehat{m},\pi)$}
& RB & 0.0 & 0.0 & 0.0 & -1.4 & -1.3 & -1.0 \\
& & RE & 92.3 & 78.4 & 74.9 & 72.3 & 65.9 & 60.0 \\
\cmidrule(lr){2-9}
& \multirow{2}{*}{$\widehat{\mu}_{cf}(\widehat{m},\pi^{(\delta)})$ (equal)}
& RB & 0.0 & 0.0 & 0.0 & 0.0 & 0.0 & 0.0 \\
& & RE & 103.3 & 87.0 & 82.5 & 86.8 & 57.0 & 42.5 \\
\cmidrule(lr){2-9}
& \multirow{2}{*}{$\widehat{\mu}_{cf}(\widehat{m},\pi^{(\delta)})$ (unequal)}
& RB & 0.0 & 0.0 & 0.0 & 0.0 & 0.0 & 0.0 \\
& & RE & 169.6 & 143.5 & 138.4 & 130.7 & 85.3 & 66.8 \\
\cmidrule(lr){2-9}
& \multirow{2}{*}{$\widehat{\mu}_{cf}(\widehat{m},\pi)$}
& RB & 0.0 & 0.0 & 0.0 & 0.0 & 0.0 & 0.0 \\
& & RE & 103.6 & 84.3 & 82.3 & 88.9 & 58.2 & 42.1 \\
\cmidrule(lr){2-9}
& \multirow{2}{*}{$\widehat{\mu}_{cf}(\widehat{m},\pi_{|\cA})$}
& RB & 0.0 & 0.0 & 0.0 & 0.0 & 0.0 & 0.0 \\
& & RE & 103.9 & 84.4 & 82.5 & 88.8 & 58.2 & 42.2 \\
\bottomrule
\end{tabular}}

\endgroup
\end{table}

When the linear model is correctly specified, cross-fitting with balanced folds entails a modest efficiency loss that decreases as the sample size increases. Under the nonlinear relationship, cross-fitting can improve efficiency with flexible learners in larger samples. For random forests with $n=1000$, the relative efficiency decreases from $31.2\%$ for the standard estimator to $25.2\%$ for design-aware cross-fitting with equal folds. With KNN, persistent bias makes the standard estimator less efficient as $n$ increases, whereas the cross-fitted estimators remain stable. Thus, even a small persistent bias can dominate the mean squared error in large samples.

\begin{table}[p]
\centering
\caption{Relative bias of the variance estimators and coverage probabilities under the linear and nonlinear relationships.}
\label{Table 2}

\begingroup
\small
\renewcommand{\arraystretch}{0.65}
\setlength{\tabcolsep}{4.5pt}

\resizebox{\textwidth}{!}{\begin{tabular}{@{}lllcccccc@{}}
\toprule
\multirow{2}{*}{Model}
& \multirow{2}{*}{Method}
& \multirow{2}{*}{}
& \multicolumn{3}{c}{Linear relationship}
& \multicolumn{3}{c}{Nonlinear relationship} \\
\cmidrule(lr){4-6}\cmidrule(lr){7-9}
& & $n$ & 200 & 500 & 1000 & 200 & 500 & 1000 \\
\midrule

\multirow{10}{*}{\shortstack{Linear\\ Regression}}
& \multirow{2}{*}{$\widehat{\mu}_{ma}(\widehat{m},\pi)$}
& RB-Var & -19.1 & -7.1 & -4.6 & -20.2 & -8.9 & -5.0 \\
& & CP & 92.1 & 94.1 & 94.4 & 90.8 & 93.4 & 94.1 \\
\cmidrule(lr){2-9}
& \multirow{2}{*}{$\widehat{\mu}_{cf}(\widehat{m},\pi^{(\delta)})$ (equal)}
& RB-Var & -0.2 & 0.5 & -0.5 & 0.2 & -0.3 & -1.0 \\
& & CP & 94.9 & 95.0 & 94.8 & 94.5 & 94.7 & 94.6 \\
\cmidrule(lr){2-9}
& \multirow{2}{*}{$\widehat{\mu}_{cf}(\widehat{m},\pi^{(\delta)})$ (unequal)}
& RB-Var & 0.2 & 0.7 & -1.1 & 0.4 & -0.5 & -1.6 \\
& & CP & 94.6 & 94.9 & 94.8 & 93.0 & 94.1 & 94.2 \\
\cmidrule(lr){2-9}
& \multirow{2}{*}{$\widehat{\mu}_{cf}(\widehat{m},\pi)$}
& RB-Var & 1.2 & 1.7 & 0.1 & 0.6 & -0.4 & -0.6 \\
& & CP & 95.0 & 95.4 & 94.9 & 94.3 & 94.8 & 94.7 \\
\cmidrule(lr){2-9}
& \multirow{2}{*}{$\widehat{\mu}_{cf}(\widehat{m},\pi_{|\cA})$}
& RB-Var & 0.2 & 1.4 & -0.3 & 0.0 & -0.7 & -0.9 \\
& & CP & 94.8 & 95.2 & 94.9 & 94.2 & 94.7 & 95.8 \\
\midrule

\multirow{10}{*}{\shortstack{Random\\ Forests}}
& \multirow{2}{*}{$\widehat{\mu}_{ma}(\widehat{m},\pi)$}
& RB-Var & -66.9 & -67.3 & -68.1 & -53.0 & -56.7 & -60.1 \\
& & CP & 74.1 & 73.4 & 72.9 & 75.9 & 68.4 & 59.7 \\
\cmidrule(lr){2-9}
& \multirow{2}{*}{$\widehat{\mu}_{cf}(\widehat{m},\pi^{(\delta)})$ (equal)}
& RB-Var & 0.2 & 0.7 & -0.8 & -1.4 & -2.7 & -3.1 \\
& & CP & 94.9 & 95.1 & 94.7 & 94.1 & 94.2 & 94.3 \\
\cmidrule(lr){2-9}
& \multirow{2}{*}{$\widehat{\mu}_{cf}(\widehat{m},\pi^{(\delta)})$ (unequal)}
& RB-Var & 0.8 & -0.5 & -1.6 & 0.3 & -0.7 & -2.4 \\
& & CP & 94.6 & 94.9 & 94.8 & 93.0 & 94.1 & 94.2 \\
\cmidrule(lr){2-9}
& \multirow{2}{*}{$\widehat{\mu}_{cf}(\widehat{m},\pi)$}
& RB-Var & 1.0 & 0.8 & -0.1 & -0.8 & -2.8 & -2.1 \\
& & CP & 95.0 & 95.1 & 94.9 & 94.1 & 94.2 & 94.5 \\
\cmidrule(lr){2-9}
& \multirow{2}{*}{$\widehat{\mu}_{cf}(\widehat{m},\pi_{|\cA})$}
& RB-Var & -0.2 & 0.5 & -0.5 & -1.7 & -3.1 & -2.3 \\
& & CP & 94.9 & 95.0 & 94.9 & 94.0 & 94.2 & 94.6 \\
\midrule

\multirow{10}{*}{KNN}
& \multirow{2}{*}{$\widehat{\mu}_{ma}(\widehat{m},\pi)$}
& RB-Var & -32.0 & -31.8 & -32.5 & -24.6 & -26.4 & -26.8 \\
& & CP & 89.3 & 89.5 & 89.1 & 85.3 & 74.7 & 68.1 \\
\cmidrule(lr){2-9}
& \multirow{2}{*}{$\widehat{\mu}_{cf}(\widehat{m},\pi^{(\delta)})$ (equal)}
& RB-Var & 0.2 & 0.3 & 0.6 & 1.2 & 0.1 & -0.2 \\
& & CP & 94.7 & 94.8 & 95.0 & 94.4 & 94.9 & 94.9 \\
\cmidrule(lr){2-9}
& \multirow{2}{*}{$\widehat{\mu}_{cf}(\widehat{m},\pi^{(\delta)})$ (unequal)}
& RB-Var & -0.7 & 0.3 & -0.9 & 0.5 & 0.3 & -1.1 \\
& & CP & 94.5 & 94.8 & 94.6 & 93.2 & 94.2 & 94.3 \\
\cmidrule(lr){2-9}
& \multirow{2}{*}{$\widehat{\mu}_{cf}(\widehat{m},\pi)$}
& RB-Var & 1.3 & 1.2 & 0.8 & 1.9 & 0.2 & 0.6 \\
& & CP & 95.0 & 95.0 & 95.1 & 94.7 & 94.8 & 94.9 \\
\cmidrule(lr){2-9}
& \multirow{2}{*}{$\widehat{\mu}_{cf}(\widehat{m},\pi_{|\cA})$}
& RB-Var & 0.5 & 0.8 & 0.6 & 1.5 & 0.1 & 0.4 \\
& & CP & 94.5 & 95.0 & 95.0 & 94.6 & 94.7 & 94.9 \\
\midrule

\multirow{10}{*}{\shortstack{Regression\\ Trees}}
& \multirow{2}{*}{$\widehat{\mu}_{ma}(\widehat{m},\pi)$}
& RB-Var & -80.0 & -80.6 & -80.9 & -59.0 & -65.6 & -69.8 \\
& & CP & 61.4 & 61.5 & 61.1 & 71.7 & 58.0 & 49.7 \\
\cmidrule(lr){2-9}
& \multirow{2}{*}{$\widehat{\mu}_{cf}(\widehat{m},\pi^{(\delta)})$ (equal)}
& RB-Var & -0.8 & -0.7 & -1.1 & -1.1 & -2.5 & -2.2 \\
& & CP & 94.8 & 94.9 & 94.8 & 94.2 & 94.2 & 94.6 \\
\cmidrule(lr){2-9}
& \multirow{2}{*}{$\widehat{\mu}_{cf}(\widehat{m},\pi^{(\delta)})$ (unequal)}
& RB-Var & 0.7 & 0.5 & -1.5 & 0.3 & 0.9 & -2.1 \\
& & CP & 94.5 & 94.8 & 94.7 & 94.1 & 94.7 & 94.5 \\
\cmidrule(lr){2-9}
& \multirow{2}{*}{$\widehat{\mu}_{cf}(\widehat{m},\pi)$}
& RB-Var & 0.2 & 2.2 & -0.2 & -3.4 & -2.7 & -0.8 \\
& & CP & 94.8 & 95.2 & 94.8 & 93.8 & 94.4 & 94.7 \\
\cmidrule(lr){2-9}
& \multirow{2}{*}{$\widehat{\mu}_{cf}(\widehat{m},\pi_{|\cA})$}
& RB-Var & -0.6 & 1.9 & -0.5 & -4.1 & -2.9 & -1.0 \\
& & CP & 94.8 & 95.2 & 94.8 & 93.7 & 94.4 & 94.8 \\
\bottomrule
\end{tabular}}

\endgroup
\end{table}

Balanced sample fold sizes are important for the design-aware procedure. For random forests under the linear relationship with $n=200$, moving from equal to unequal folds increases the relative efficiency from $68.3\%$ to $111.3\%$. Apart from this effect, the design-aware and design-agnostic procedures perform similarly, and the choice between $\pi$ and $\pi_{|\cA}$ is essentially immaterial.

Table~\ref{Table 2} shows that the standard model-assisted variance estimator can be severely biased downward. Under the linear relationship, this bias decreases with sample size for linear regression but remains large for random forests and regression trees. The resulting confidence intervals undercover substantially: under the nonlinear relationship with $n=1000$, coverage is $59.7\%$ for random forests and $49.7\%$ for regression trees.

Cross-fitting largely eliminates this downward variance bias and yields coverage generally around $94\%$--$95\%$. Unequal folds can slightly reduce coverage for the design-aware procedure, but the effect is much smaller than for point estimation. The design-aware and design-agnostic procedures again perform similarly.

\FloatBarrier

\section{Final remarks}\label{sec:final}

In this paper, we developed an agnostic framework for model-assisted estimation with modern machine learning methods. The proposed cross-fitting procedures accommodate flexible prediction rules while ensuring design consistency, asymptotic normality, and valid variance estimation under general conditions. The simulations show negligible bias, competitive or improved efficiency, and variance estimators with small bias, yielding Wald confidence intervals with coverage close to the nominal level, particularly for highly adaptive learners.

Although our development focuses on the population mean, the framework extends naturally to finite-population parameters defined through estimating equations. Let $\theta_N$ solve the census estimating equation
$$
\Psi(\theta)
=
\frac{1}{N}\sum_{k\in U}\psi(y_k,\theta)
=
0,
$$
where $\psi:\mathbb{R}\times\Theta\rightarrow\mathbb{R}$ is known. For example, $\psi(y,\theta)=y-\theta$ yields the population mean, whereas $\psi(y,\theta)=\mathbf{1}(y\leq t)-\theta$ yields the finite-population distribution function at $t$.

Let $\widehat{g}^{(-k)}(\mathbf{x}_k,\theta)$ be a cross-fitted predictor of $\psi(y_k,\theta)$ constructed without observations from the fold containing unit $k$. A cross-fitted model-assisted estimator $\widehat{\theta}_{cf}$ can then be defined as a solution to $\widehat{\Psi}_{cf}(\widehat{\theta}_{cf})=0$, where
$$
\widehat{\Psi}_{cf}(\theta)
=
\frac{1}{N}
\left\{
\sum_{k\in U}\widehat{g}^{(-k)}(\mathbf{x}_k,\theta)
+
\sum_{k\in S}
\frac{
\psi(y_k,\theta)
-
\widehat{g}^{(-k)}(\mathbf{x}_k,\theta)
}{\pi_k}
\right\}.
$$
Model-assisted estimation of finite-population distribution functions using auxiliary information was studied by \citet{rao1990} under simple regression models, including ratio and difference-type estimators. The formulation above accommodates general prediction methods and addresses data reuse through cross-fitting. Encouraging preliminary simulations motivate a detailed study of these extensions in future work.

\renewcommand{\bibsection}{\section*{References}}

\begingroup
\setstretch{1.1}
\setlength{\bibsep}{5pt}
\setlength{\parskip}{0pt}
\filterreferences
\bibliography{cross-fitting}
\endgroup

\endgroup

\newpage
\begingroup
\referencepart{appendix}
\appendix
\addtocontents{toc}{\protect\setcounter{tocdepth}{2}}

\begingroup
\setstretch{1}
\setlength{\parskip}{0pt}
\renewcommand{\contentsname}{Contents of the Appendix}
\setcounter{tocdepth}{-1}
\tableofcontents
\endgroup
\bigskip

\newpage
\section{Illustrative examples for model-based and model-assisted estimation}
\label{app:illustrative-examples}
\subsection{A simple illustration for the model-based estimator}
\label{app:mb-example}

This section illustrates how the convergence rate of a model-based estimator may be driven by that of the fitted prediction rule. This issue is particularly relevant for nonparametric and machine learning methods, which typically converge more slowly than the parametric rate $n^{-1/2}$. For example, when the regression function is Lipschitz in $p$ dimensions, a typical minimax rate in root mean squared prediction error is $n^{-1/(p+2)}$, up to constants and logarithmic factors. This gives $n^{-1/3}$ for $p=1$ and $n^{-1/4}$ for $p=2$.

Consider the prediction rule
$$
\widehat{m}(\mathbf{x}) = m(\mathbf{x})+\frac{\xi}{n^\alpha}, \qquad 0<\alpha<\frac{1}{2}, \qquad
\xi\sim\mathcal{N}(0,1).
$$
Although $\widehat m$ converges to $m$ at rate $n^{-\alpha}$, its contribution to the model-based estimator is
$$
\widehat{\mu}_{mb}(\widehat m) - \widehat{\mu}_{mb}(m) = \frac{1}{N} \sum_{k\in U\setminus S}
\frac{\xi}{n^\alpha} = \frac{N-n}{N}\frac{\xi}{n^\alpha}.
$$
Consequently,
$$
\sqrt n \left\{ \widehat{\mu}_{mb}(\widehat m) - \widehat{\mu}_{mb}(m) \right\} =
\frac{N-n}{N}n^{1/2-\alpha}\xi.
$$
If the nonsampling fraction $(N-n)/N$ is bounded away from zero, this quantity is not bounded in probability when $\alpha<1/2$. Thus, even though $\widehat m$ is consistent for $m$, the error due to estimating the prediction rule may dominate the usual $n^{-1/2}$ sampling variation. This illustrates the lack of Neyman orthogonality of the model-based estimator: prediction errors enter the estimator at first order.

\subsection{Failure of design consistency under data reuse}
\label{app:ma-counterexample}

We give a simple example showing that data reuse may prevent the feasible model-assisted estimator from being design-consistent. Consider a population with distinct auxiliary vectors and $y_k=1$ for every $k\in U$, so that $\mu=1$, and define the fitted prediction rule by
$$
\widehat{m}(\mathbf{x}_k;\mathbf{I})=I_k.
$$
Thus, the prediction is exact for every sampled unit but equals zero for every nonsampled unit. Substituting this rule into the model-assisted estimator gives
$$
\begin{aligned}
\widehat{\mu}_{\mathrm{ma}}(\widehat m,\pi)
&=
\frac{1}{N}
\left\{
\sum_{k\in U}I_k
+
\sum_{k\in S}\frac{1-I_k}{\pi_k}
\right\}  \\
&=
\frac{1}{N}\sum_{k\in U}I_k
=
\frac{n}{N}.
\end{aligned}
$$
Therefore, if $n/N\xrightarrow[v \to \infty]{} f<1$, then
$$
\widehat{\mu}_{\mathrm{ma}}(\widehat m,\pi) \xrightarrow[v \to \infty]{} f\neq\mu.
$$
The estimator is consequently not design-consistent unless the sampling fraction converges to one. The failure is caused by the perfect dependence between $\widehat{m}(\mathbf{x}_k;\mathbf{I})$ and $I_k$: the fitted rule interpolates the sampled observations but does not provide meaningful predictions outside the sample. This example illustrates why some control of the dependence induced by data reuse is necessary for a general theory.

\section{Additional material on cross-fitting} \label{AppB}
\subsection{Design-aware cross-fitting schemes}
\label{supp:lu-schemes}

\citet{lu2025conditional} propose several cross-fitting schemes satisfying the conditional independence property (Ind-$\delta$). We briefly describe their constructions for two folds, that is, $M=2$.

\paragraph*{Bernoulli sampling.}
For a Bernoulli randomized experiment, corresponding to Bernoulli sampling in our setting, generate $N$ independent Bernoulli random variables $(L_k)_{k\in U}$ with success probability $\pi$, independently of $S$. The two folds are defined by
$$
U_1=\{k\in U:L_k=1\}, \qquad U_2=\{k\in U:L_k=0\}.
$$

\paragraph*{Simple random sampling without replacement.}
For a completely randomized experiment, corresponding to simple random sampling without replacement (SRSWOR), first split the sample $S$ uniformly at random into two disjoint subsets $S_1$ and $S_2$ of prescribed sizes $n_1$ and $n_2$, where $n_1+n_2=n$. Independently of this split, split the nonsampled units $U\setminus S$ uniformly at random into two disjoint subsets $R_1$ and $R_2$ of respective sizes $N_1-n_1$ and $N_2-n_2$. The folds are then defined by
$$
U_1=S_1\cup R_1, \qquad U_2=S_2\cup R_2.
$$

\paragraph*{Stratified simple random sampling without replacement.}
For a stratified randomized experiment, corresponding to stratified SRSWOR, \citet{lu2025conditional} describe two constructions. The first applies the preceding SRSWOR construction independently within each stratum and combines the resulting subsets across strata. The second partitions the stratum labels $1,\ldots,H$ between the two folds, independently of $S$, assigning all units from a given stratum to the same fold.

Each construction produces a random partition for which the sampling indicators in one fold are conditionally independent of those in the other fold, given the partition. Thus, all three schemes satisfy condition (Ind-$\delta$).

\subsection{Extension of cross-fitting schemes for stratified sampling}

Assume that the population $U$ is partitioned into $H$ strata $U^{(1)}, ..., U^{(H)}$. Stratified sampling uses the design $p(s)=\prod_{h=1}^H p_h(s_h)$, where $s_h=s\cap U^{(h)}$ and sampling is independent across strata. One valid cross-fitting scheme consists of partitioning the stratum labels into $F_1, ..., F_M$, independently of $S$, and setting $U_j=\bigcup_{h\in F_j}U^{(h)}$. Since each stratum belongs to only one fold, independence across strata gives \textnormal{(Ind-$\delta$)}, regardless of the within-stratum designs. The drawback is that prediction accuracy may suffer: if the relationships between the covariates and the survey variable differ substantially across strata, the prediction rule must extrapolate to strata absent from its training sample.

To circumvent this drawback, we propose the following construction, which uses within-stratum factorization when available and the above observation otherwise.

\begin{algorithm}[H]
\setstretch{1}
\setlength{\parskip}{4pt}
\setlength{\abovedisplayskip}{6pt}
\setlength{\belowdisplayskip}{6pt}
\setlength{\abovedisplayshortskip}{3pt}
\setlength{\belowdisplayshortskip}{3pt}
\caption{Cross-fitting under stratified sampling.}
\label{alg:stratified-crossfitting}
\vspace{1mm}

\textbf{Input:}
Strata $U^{(1)},\ldots,U^{(H)}$, sample $S$, number of folds $M$, and a fixed partition of the strata labels
$$
\{1,\ldots,H\}     =     \cH_{\rm split}\biguplus\cH_{\rm full},
$$
where a valid within-stratum cross-fitting scheme is specified for each $h\in\cH_{\rm split}$, and each stratum $h\in\cH_{\rm full}$ is kept intact.

\textbf{Do:}

\begin{enumerate}

    \item For each $h\in\cH_{\rm split}$, apply the specified
    cross-fitting scheme using only the sample in $U^{(h)}$ and independent randomization across strata, producing
    $$
U^{(h)}=U_1^{(h)}\biguplus\cdots\biguplus U^{(h)}_M.
$$
    
        \item For $j=1,\ldots,M$, combine the corresponding
    stratum-level folds by setting
    $$
U_j=\bigcup_{h\in\cH_{\rm split}} U_j^{(h)}.
$$

    \item For each $h\in\cH_{\rm full}$, assign the entire stratum
    $U^{(h)}$ to one of the folds $U_1,\ldots,U_M$, independently of $S$ and of the preceding randomizations.

\end{enumerate}

\textbf{Output:}
The population partition
$$
\cP=(U_1,\ldots,U_M).
$$
\end{algorithm}

\begin{proposition}\label{prop:strat}
    Let $p$ be a stratified sampling design with independent sampling across strata, and let $\cP$ be constructed by Algorithm \ref{alg:stratified-crossfitting}. Let $\cA$ collect the within-stratum conditioning variables and the assignments of intact strata. Then, the resulting scheme satisfies \textnormal{(Ind-$\cA$)}.
\end{proposition}
\begin{proof}

The samples and the randomizations used within different strata are independent. Each within-stratum conditioning variable depends only on that stratum's sample and randomization. Conditioning on their collection therefore preserves independence across strata. The assignments of intact strata are independent of these variables and of the samples, so conditioning on them also preserves this independence.

Within each split stratum, the conditional sample law factorizes across folds by the assumed validity of its cross-fitting scheme. Each intact stratum contributes sampling indicators to only one fold, so its factors for the other folds equal one. For $s\in\cP(U)$, these factorizations give
\begin{align*}
\P(S=s\mid\cA)
&=\prod_{h=1}^H\P(S\cap U^{(h)}=s\cap U^{(h)}\mid\cA)\\
&=\prod_{h=1}^H\prod_{j=1}^M
\P(S\cap U_j\cap U^{(h)}=s\cap U_j\cap U^{(h)}\mid\cA)\\
&=\prod_{j=1}^M\P(S\cap U_j=s\cap U_j\mid\cA).
\end{align*}
The last equality groups the independent stratum samples by fold. This proves \textnormal{(Ind-$\cA$)}.
\end{proof}

Following \cite{sarndal1992}, for strata $h \in \cH_{\rm split}$, several possibilities exist in order to train a regression function estimator $\widehat{m}$: (i) stratum-specific training; (ii) pooled training. Stratum-specific training consists in fitting a separate model for each stratum $h$, that is, fitting a regression function estimator $\widehat{m}_{jh}$ on $(\bx_k, y_k)_{S^{(h)} \backslash S_j^{(h)}}$ for prediction in $U_j^{(h)}$. Pooled training, on the other hand, would consist of pooling all the data across strata together in folds and proceeding to train a common regression function estimator $\widehat{m}_{j}$ on $(\bx_k, y_k)_{S \backslash S_j}$, to predict in the whole fold $U_j$, which would typically contain different strata. The main advantage of the stratum-specific training is that it allows for relationships between covariates $X$ and the survey variable $Y$ to differ freely in each stratum. However, the associated drawback is that the resulting sample size used to fit the stratum-specific learner $\widehat{m}_{jh}$ would typically be much smaller than that which could be used with the pooled learner $\widehat{m}_{j}$. If the sample size within strata is large enough and strata are very heterogeneous, this may be an appropriate strategy. If not, fitting a complex learner like a deep neural network, for example, might be challenging. The pooled version, on the other hand, would not suffer from these drawbacks, at the price of perhaps capturing only global trends in the regression function. An alternative strategy, when the number of strata is small compared to the overall sample size, consists of including the strata memberships as additional covariates.
Nonetheless, independently of the learning strategy used, the cross-fitting scheme will remain valid. For strata $h \in \cH_{\rm full}$, however, these strategies are not possible, and only pooled training is available.

\section{Godambe-Joshi asymptotic optimality} \label{App:GJ}

To discuss optimality, suppose that the population data $(\bx_k,y_k)_{k\in U_v}$ are independent and identically distributed. We can then write
$$y_k=m(\bx_k)+\epsilon_k,\qquad k\in U_v,$$
where $m:\bx\mapsto\E[Y\mid X=\bx]$. Conditional on the population covariates $\boldsymbol{X}_U$, the errors are independent, with $\E_m[\epsilon_k]=0$ and $\E_m[\epsilon_k^2]=\sigma^2(\bx_k)$. The sampling design is non-informative, and the fold construction and $\cA$ use only the sample, the covariates, and independent randomization. Thus, given $\boldsymbol{X}_U$, $(\boldsymbol{I}_U,\cA)$ is independent of $\boldsymbol{y}_U$; see \cite{pfeffermann2009inference} for non-informative sampling. Throughout this appendix, expectations are conditional on the population covariates: $\E_m$ averages over the responses, and $\E_{mp}$ also averages over sampling and fold construction. We moreover assume that
$$\limsup_{v\to\infty}\dfrac{1}{N_v}\sum_{k\in U_v}\sigma^2(\bx_k)<\infty.$$

In this setup, \cite{godambe1965admissibility} showed that, for any design-unbiased $\widehat{\mu}$,
$$\E_{mp}\left[\left(\widehat{\mu}-\mu\right)^2\right]\geq\dfrac{1}{N_v^2}\sum_{k\in U_v}\dfrac{1-\pi_k}{\pi_k}\sigma^2(\bx_k).$$
Several model-assisted estimators reach this bound to first order under appropriate assumptions, meaning
$$n_v\E_{mp}\left[\left(\widehat{\mu}_{ma}-\mu\right)^2\right]=\dfrac{n_v}{N_v^2}\sum_{k\in U_v}\dfrac{1-\pi_k}{\pi_k}\sigma^2(\bx_k)+o(1):=\textrm{GJ}+o(1).$$
This holds for linear regression with a well-specified model \citep{robinsonsarndal1983}, local polynomials \citep{breidt2000local}, and $B$-splines \citep{goga2005reduction}, to mention a few. The same argument applies to cross-fitting. When $\widetilde m=m$, the oracle $\widehat{\mu}_{\mathrm{ma}}(m,\pi)$ reaches the bound exactly. It therefore suffices to strengthen the first-order equivalence with the feasible estimator to convergence in $L^2$.

For the conditional oracle $\widehat{\mu}_{\mathrm{ma}}(m,\pi_{|\cA})$, the joint mean-squared error has an additional nonnegative term. A sufficient condition for this term to vanish at first order is that $\E_p[\max_{k\in U_v}(\pi_{k|\cA}-\pi_k)^2]=o(1)$. To make the role of this proximity clear, we do not assume \ref{CF3} in the following theorem and instead state the conditions directly.

\begin{theorem}\label{theoGJ}
Assume (Ind-$\cA$), \ref{D1}, \ref{D2}, \ref{CF2c}, \ref{CF4}, and \ref{M1}, with \ref{M1} holding almost surely under the model. Suppose that $\widetilde m=m$ and that the prediction errors $\widehat{m}^{(-k)}(\bx_k)-m(\bx_k)$ are bounded by a fixed constant, uniformly in $k$ and $v$, almost surely under the joint model and sampling law. Then, the following statements hold.
\begin{enumerate}
    \item[(i)] Assume that
    $$\E_p\left[\max_{k\in U_v}\left(\pi_{k|\cA}-\pi_k\right)^2\right]=o(1).$$
    Then, $\widehat{\mu}_{cf}(\widehat m,\pi_{|\cA})$ asymptotically reaches the Godambe-Joshi lower bound.
    \item[(ii)] Assume that
    $$\lim_{v\to\infty}n_v\E_{mp}\left[\max_{k\in U_v}\left(\pi_{k|\cA}-\pi_k\right)^2\times\frac{1}{N_v}\sum_{k\in U_v}\left(\widehat{m}^{(-k)}(\bx_k)-m(\bx_k)\right)^2\right]=0.$$
    Then, $\widehat{\mu}_{cf}(\widehat m,\pi)$ asymptotically reaches the Godambe-Joshi lower bound.
\end{enumerate}
\end{theorem}
\begin{proof}
Let $\widehat{\mu}$ denote either feasible estimator and $\widetilde{\mu}$ its corresponding oracle, $\widehat{\mu}_{\mathrm{ma}}(m,\pi_{|\cA})$ or $\widehat{\mu}_{\mathrm{ma}}(m,\pi)$. We have
$$\E_{mp}\left[\left(\widehat{\mu}-\mu\right)^2\right]=\E_{mp}\left[\left(\widehat{\mu}-\widetilde{\mu}\right)^2\right]+\E_{mp}\left[\left(\widetilde{\mu}-\mu\right)^2\right]+R_v,$$
where, by Cauchy--Schwarz,
\begin{align*}
|R_v|
&=2\left|\E_{mp}\left[(\widehat{\mu}-\widetilde{\mu})(\widetilde{\mu}-\mu)\right]\right|\\
&\leq2\left\{\E_{mp}\left[\left(\widehat{\mu}-\widetilde{\mu}\right)^2\right]\right\}^{1/2}
\left\{\E_{mp}\left[\left(\widetilde{\mu}-\mu\right)^2\right]\right\}^{1/2}.
\end{align*}
Since $n_v/N_v\leq1$, \ref{D2} and the bound on the average model variance give $\textrm{GJ}=\mathcal{O}(1)$. Thus, it is enough to show that
$$n_v\E_{mp}\left[\left(\widetilde{\mu}-\mu\right)^2\right]=\textrm{GJ}+o(1),\qquad n_v\E_{mp}\left[\left(\widehat{\mu}-\widetilde{\mu}\right)^2\right]=o(1).$$
Indeed, the bound on $R_v$ then gives $n_v|R_v|=o(1)$.
To bound the difference between each feasible estimator and its oracle, we will use the mean-square convergence of the prediction errors. By \ref{M1}, the inner average below tends to zero for almost every model realization. It is bounded by the square of the fixed bound on the prediction errors. Thus, dominated convergence gives
\begin{align*}
&\frac{1}{N_v}\sum_{k\in U_v}\E_{mp}\left[\left(\widehat{m}^{(-k)}(\bx_k)-m(\bx_k)\right)^2\right]\\
&\quad=\E_m\left[\frac{1}{N_v}\sum_{k\in U_v}
\E_p\left[\left(\widehat{m}^{(-k)}(\bx_k)-m(\bx_k)\right)^2\right]\right]
\xrightarrow[v\to\infty]{}0.
\end{align*}

\emph{Proof of (i).} We first compute the mean squared error of the conditional oracle. Its estimation error is
$$\widehat{\mu}_{\mathrm{ma}}(m,\pi_{|\cA})-\mu
=\frac{1}{N_v}\sum_{k\in U_v}\left(\frac{I_k}{\pi_{k|\cA}}-1\right)\epsilon_k.$$
The model errors are independent and centered, and are independent of sampling and fold construction given the covariates. For the diagonal terms,
\begin{align*}
\E_p\left[\left(\frac{I_k}{\pi_{k|\cA}}-1\right)^2\Bigm|\cA\right]=\frac{1-\pi_{k|\cA}}{\pi_{k|\cA}}.
\end{align*}
Expanding the square and separating the sampling and model expectations, we obtain
\begin{align*}
n_v\E_{mp}\left[\left(\widehat{\mu}_{\mathrm{ma}}(m,\pi_{|\cA})-\mu\right)^2\right]
&=\frac{n_v}{N_v^2}\sum_{k\in U_v}\sum_{l\in U_v}
\E_p\left[
\left(\frac{I_k}{\pi_{k|\cA}}-1\right)
\left(\frac{I_l}{\pi_{l|\cA}}-1\right)
\right]\E_m[\epsilon_k\epsilon_l]\\
&=\dfrac{n_v}{N_v^2}\sum_{k\in U_v}\E_p\left[\dfrac{1-\pi_{k|\cA}}{\pi_{k|\cA}}\right]\sigma^2(\bx_k)\\
&=\textrm{GJ}+\dfrac{n_v}{N_v^2}\sum_{k\in U_v}\left(\E_p\left[\dfrac{1}{\pi_{k|\cA}}\right]-\dfrac{1}{\pi_k}\right)\sigma^2(\bx_k).
\end{align*}
To bound the additional term, use the identity
$$\frac{1}{\pi_{k|\cA}}-\frac{1}{\pi_k}
=-\frac{\pi_{k|\cA}-\pi_k}{\pi_k^2}
+\frac{(\pi_{k|\cA}-\pi_k)^2}{\pi_k^2\pi_{k|\cA}}$$
and $\E_p[\pi_{k|\cA}]=\pi_k$ to obtain
$$\E_p\left[\dfrac{1}{\pi_{k|\cA}}\right]-\dfrac{1}{\pi_k}
=\E_p\left[\dfrac{(\pi_{k|\cA}-\pi_k)^2}{\pi_k^2\pi_{k|\cA}}\right].$$
By \ref{D2} and \ref{CF2c}, the additional term is nonnegative and at most
$$C\E_p\left[\max_{k\in U_v}(\pi_{k|\cA}-\pi_k)^2\right]\dfrac{1}{N_v}\sum_{k\in U_v}\sigma^2(\bx_k)=o(1),$$
for a constant $C$ independent of $v$.

We next bound the difference between the feasible and oracle estimators. From their definitions,
$$\widehat{\mu}_{cf}(\widehat m,\pi_{|\cA})-\widehat{\mu}_{\mathrm{ma}}(m,\pi_{|\cA})
=-\frac{1}{N_v}\sum_{k\in U_v}(I_k-\pi_{k|\cA})\frac{\widehat{m}^{(-k)}(\bx_k)-m(\bx_k)}{\pi_{k|\cA}}.$$
Applying Lemma~\ref{lemma:secondMoment} with $d_k=\{\widehat{m}^{(-k)}(\bx_k)-m(\bx_k)\}/\pi_{k|\cA}$ gives
\begin{align*}
&n_v\E_p\left[\left\{\widehat{\mu}_{cf}(\widehat m,\pi_{|\cA})-\widehat{\mu}_{\mathrm{ma}}(m,\pi_{|\cA})\right\}^2\right]\\
&\qquad\leq\frac{C}{N_v}\sum_{k\in U_v}\E_p\left[\left\{\frac{\widehat{m}^{(-k)}(\bx_k)-m(\bx_k)}{\pi_{k|\cA}}\right\}^2\right].
\end{align*}
Averaging over the model and using \ref{CF2c}, we obtain
\begin{align*}
&n_v\E_{mp}\left[\left\{\widehat{\mu}_{cf}(\widehat m,\pi_{|\cA})-\widehat{\mu}_{\mathrm{ma}}(m,\pi_{|\cA})\right\}^2\right]\\
&\qquad\leq\frac{C}{N_v}\sum_{k\in U_v}\E_{mp}\left[\left(\widehat{m}^{(-k)}(\bx_k)-m(\bx_k)\right)^2\right]=o(1).
\end{align*}
This proves (i).

\emph{Proof of (ii).} For the unconditional oracle, the same model calculation gives
$$n_v\E_{mp}\left[\left(\widehat{\mu}_{\mathrm{ma}}(m,\pi)-\mu\right)^2\right]
=\dfrac{n_v}{N_v^2}\sum_{k\in U_v}\dfrac{1-\pi_k}{\pi_k}\sigma^2(\bx_k)=\textrm{GJ}.$$
For the difference between the feasible and oracle estimators, we have
\begin{align*}
&\widehat{\mu}_{cf}(\widehat m,\pi)-\widehat{\mu}_{\mathrm{ma}}(m,\pi)\\
&\quad=-\frac{1}{N_v}\sum_{k\in U_v}(I_k-\pi_{k|\cA})\frac{\widehat{m}^{(-k)}(\bx_k)-m(\bx_k)}{\pi_k}-\frac{1}{N_v}\sum_{k\in U_v}(\pi_{k|\cA}-\pi_k)\frac{\widehat{m}^{(-k)}(\bx_k)-m(\bx_k)}{\pi_k}.
\end{align*}
For the first sum, Lemma~\ref{lemma:secondMoment} and \ref{D2}, followed by averaging over the model, give
\begin{align*}
&n_v\E_{mp}\left[\left\{\frac{1}{N_v}\sum_{k\in U_v}(I_k-\pi_{k|\cA})\frac{\widehat{m}^{(-k)}(\bx_k)-m(\bx_k)}{\pi_k}\right\}^2\right]\\
&\qquad\leq\frac{C}{N_v}\sum_{k\in U_v}\E_{mp}\left[\left(\widehat{m}^{(-k)}(\bx_k)-m(\bx_k)\right)^2\right].
\end{align*}
For the second sum, Cauchy--Schwarz and \ref{D2} give
\begin{align*}
&n_v\left\{\frac{1}{N_v}\sum_{k\in U_v}(\pi_{k|\cA}-\pi_k)\frac{\widehat{m}^{(-k)}(\bx_k)-m(\bx_k)}{\pi_k}\right\}^2\\
&\qquad\leq n_v
\left\{\frac{1}{N_v}\sum_{k\in U_v}\left(\frac{\pi_{k|\cA}-\pi_k}{\pi_k}\right)^2\right\}
\left\{\frac{1}{N_v}\sum_{k\in U_v}\left(\widehat{m}^{(-k)}(\bx_k)-m(\bx_k)\right)^2\right\}\\
&\qquad\leq\frac{n_v}{\lambda^2}\max_{k\in U_v}(\pi_{k|\cA}-\pi_k)^2
\times\frac{1}{N_v}\sum_{k\in U_v}\left(\widehat{m}^{(-k)}(\bx_k)-m(\bx_k)\right)^2.
\end{align*}
Taking expectations in the second bound and combining the two bounds using $(a+b)^2\leq2a^2+2b^2$, we obtain
\begin{align*}
&n_v\E_{mp}\left[\left\{\widehat{\mu}_{cf}(\widehat m,\pi)-\widehat{\mu}_{\mathrm{ma}}(m,\pi)\right\}^2\right]\\
&\quad\leq\frac{C}{N_v}\sum_{k\in U_v}\E_{mp}\left[\left(\widehat{m}^{(-k)}(\bx_k)-m(\bx_k)\right)^2\right]\\
&\qquad+Cn_v\E_{mp}\left[\max_{k\in U_v}(\pi_{k|\cA}-\pi_k)^2\times\frac{1}{N_v}\sum_{k\in U_v}\left(\widehat{m}^{(-k)}(\bx_k)-m(\bx_k)\right)^2\right].
\end{align*}
The first term tends to zero by the joint mean-square convergence above, and the second by the assumption in (ii). This completes the proof.
\end{proof}

\section{Overview of exponential sampling designs}\label{exponetialdesign}

\subsection{Notation and definitions}

Let $U=\{1,\dots,N\}$ be a population from which a sample is to be selected.
A sample $s\subset U$ is simply a subset of $U$. A \emph{support}
$\mathcal{Q}$ is a set of samples. We will mainly work with two supports:
the set of all samples of $U$ (the power set of $U$),
$\mathcal{S} = \{s\subset U\},$
and the set of all samples of fixed size $n$,
$\mathcal{S}_n = \{s\subset U : |s| = n\}.$

A random sample is selected by means of a \emph{sampling design} $p(\cdot)$
defined on a support $\mathcal{Q}$, that is, a function satisfying
$$
p(s) > 0, \quad s \in \mathcal{Q}, \qquad \sum_{s\in\mathcal{Q}} p(s) = 1.
$$
The (first-order) \emph{inclusion probabilities} are obtained from the design by
$$
\pi_k = \sum_{\substack{s\in \mathcal{Q}\\ s \ni k}} p(s), \qquad k\in U,
$$
and the second-order (joint) inclusion probabilities by
$$
\pi_{k\ell} = \sum_{\substack{s\in \mathcal{Q}\\ s \supset \{k,\ell\}}} p(s), \qquad k,\ell\in U.
$$
We write $S$ for the random sample selected according to $p(\cdot)$, so that
$\Pr(S=s) = p(s)$ for $s\in\mathcal{Q}$.

First, we will define exponential schemes, and then we will list  a number of classical sampling schemes,
each of which belongs to the exponential family.

\subsection{Exponential designs}

\subsubsection{Definition}

\begin{definition}
Given a sampling design $p(\cdot)$ on $\mathcal{S}$, its \emph{support} is
the set of samples with nonzero probability,
$$
\mathcal{Q} = \{ s \subset U \mid p(s) > 0 \}.
$$
\end{definition}

\begin{definition}
A support $\mathcal{Q}$ is \emph{symmetric} if, for every $s\in\mathcal{Q}$,
every subset $r\subseteq U$ with $|r| = |s|$ also belongs to $\mathcal{Q}$.
\end{definition}

Both $\mathcal{S}$ and $\mathcal{S}_n$ are symmetric supports.

Among all designs on a support $\mathcal{Q}$ with prescribed inclusion
probabilities $\pi_k$, $k\in U$, the \emph{exponential design} is the one
that maximizes the entropy
$$
I(p) = -\sum_{s\in\mathcal{Q}} p(s)\log p(s).
$$
Maximizing $I(p)$ subject to $\sum_{s\in\mathcal{Q}} p(s) = 1$ and
$\sum_{s\in\mathcal{Q}, s\ni k} p(s) = \pi_k$, $k\in U$, by means of a
Lagrangian function with multipliers $\lambdag = (\lambda_1,\dots,\lambda_N)^\top$
yields a design of the form
\begin{equation}
p_{exp}(s;\mathcal{Q},\lambdag) = \frac{\prod_{k\in s}\omega_k}{\sum_{r\in\mathcal{Q}}\prod_{k\in r}\omega_k},
\qquad \omega_k = \exp(\lambda_k), k\in U. \label{entrmax}
\end{equation}
We use $p_{exp}(s;\mathcal{Q},\lambdag)$ as a shorthand for this exponential
design, built on support $\mathcal{Q}$ with parameter $\lambdag$ (equivalently,
with weights $\omega_k=\exp(\lambda_k)$).

The parameters $\omega_k$, $k\in U$, are determined so that the design
\eqref{entrmax} reproduces the target inclusion probabilities, i.e.\ by
solving the system of equations
\begin{equation}
\sum_{\substack{s\in \mathcal{Q}\\ s \ni k}} p_{exp}(s;\mathcal{Q},\lambdag)
= \sum_{\substack{s\in \mathcal{Q}\\ s \ni k}}\frac{ \prod_{\ell\in s} \omega_\ell }{\sum_{r\in\mathcal{Q}} \prod_{\ell\in r} \omega_\ell }
= \pi_k,\quad k\in U. \label{eq:omegasystem}
\end{equation}
Solving this system is not always straightforward: it has a closed-form
solution in some cases (Bernoulli sampling, Poisson sampling) but not in
others (conditional Poisson sampling).

Throughout, $\1b = (1,\dots,1)^\top \in \R^N$. We now show that several
classical designs are exponential designs, i.e.\ special cases of
\eqref{entrmax}.

\subsubsection{Bernoulli sampling}

Bernoulli sampling is the exponential design $p_{exp}(s,\mathcal{S},\theta\1b)$
with a common parameter $\theta\in\R$ (so $\lambda_k=\theta$ for every
$k\in U$). We obtain
\begin{align*}
p(s)
&= \frac{ \prod_{k\in s} \exp(\theta) }{ \sum_{r\in \mathcal{S}} \prod_{k\in r} \exp(\theta) }
= \frac{ \exp(\theta |s|) }{ \sum_{r\in \mathcal{S}} \exp(\theta |r|) }
= \frac{ \exp(\theta |s|) }{ \sum_{g=0}^N \binom{N}{g} \exp(\theta g) } \\
&= \frac{ \exp(\theta |s|) }{ \{1+\exp(\theta)\}^N }
= \left\{ \frac{\exp(\theta)}{1+\exp(\theta)} \right\}^{|s|} \left\{ \frac{1}{1+\exp(\theta)} \right\}^{N-|s|}
= \pi^{|s|}(1-\pi)^{N-|s|}, \quad s\subset U,
\end{align*}
where the third equality uses the binomial theorem, and where the
inclusion probabilities are $\pi_k = \exp(\theta)/\{1+\exp(\theta)\}$ for all
$k\in U$.

\subsubsection{Simple random sampling}

Simple random sampling without replacement, with fixed sample size $n$, is
the exponential design $p_{exp}(s,\mathcal{S}_n,\theta\1b)$ with common
parameter $\theta\in\R$. For $|s|=n$,
$$
p(s) = \frac{ \prod_{k\in s} \exp(\theta) }{ \sum_{r\in \mathcal{S}_n} \prod_{k\in r} \exp(\theta) } = \frac{
\exp(\theta n) }{ \sum_{r\in \mathcal{S}_n} \exp(\theta n) } = \binom{N}{n}^{-1}, \qquad s\in \mathcal{S}_n.
$$
Since every $s\in\mathcal{S}_n$ has the same size $n$, the factor
$\exp(\theta n)$ cancels between the numerator and the denominator: the
design does not depend on $\theta$ and reduces to the uniform design on
$\mathcal{S}_n$. The resulting inclusion probabilities are $\pi_k = n/N$,
$k\in U$.

\subsubsection{Poisson sampling}

Poisson sampling is the exponential design $p_{exp}(s,\mathcal{S},\lambdag)$,
now with a design-specific parameter $\lambda_k$ for each unit. For this
design, the general system \eqref{eq:omegasystem}, with
$\mathcal{Q}=\mathcal{S}$, has the closed-form solution
$$
\omega_k = \frac{\pi_k}{1-\pi_k}, \qquad k\in U
$$
\citep[see][]{tille2006sampling}. Substituting into \eqref{entrmax} and simplifying the
sum over $\mathcal{S}$ gives
\begin{equation}
p_{ME}(s;\mathcal{S}) = \frac{ \prod_{k\in s} \omega_k }{\sum_{r\in\mathcal{S}} \prod_{k\in r} \omega_k }
= \left\{ \prod_{k\in s} \pi_k \right\}\left\{ \prod_{k\in U\setminus s} (1-\pi_k) \right\},
\label{plandepoisson}
\end{equation}
that is, Poisson sampling selects each unit $k$ independently with
probability $\pi_k$.

\subsubsection{Conditional Poisson sampling}

Conditional Poisson sampling (CPS), also called the maximum entropy design
with fixed sample size, is the exponential design
$p_{exp}(s,\mathcal{S}_n,\lambdag)$. Here the system \eqref{eq:omegasystem},
restricted to $\mathcal{Q}=\mathcal{S}_n$, becomes
$$
\sum_{\substack{s\in \mathcal{S}_n\\ s \ni k}} p(s)
= \sum_{\substack{s\in \mathcal{S}_n\\ s \ni k}}\frac{ \prod_{\ell\in s} \omega_\ell }{\sum_{r\in\mathcal{S}_n} \prod_{\ell\in r} \omega_\ell }
= \pi_k,\quad k\in U,
$$
and the relationship between $\pi_k$ and $\omega_k$ is considerably more
complex than in the Poisson case: this system cannot be solved analytically,
and there is no closed form for $\omega_k$ as a function of $\pi_k$.

Nevertheless, several methods allow $\omega_k$ to be computed from $\pi_k$
without summing over all samples in $\mathcal{S}_n$; a detailed account is
given in \citet{tille2006sampling}. The design is therefore
\begin{equation}
p_{ME}(s;\mathcal{S}_n) = \frac{ \prod_{k\in s} \omega_k }{\sum_{r\in\mathcal{S}_n} \prod_{k\in r} \omega_k }, \label{cpsdesign}
\end{equation}
an expression that cannot be simplified further in general. Once the
$\omega_k$ have been determined, however, several algorithms are available
for drawing a sample from this design  \citep[see again][]{tille2006sampling}.

In practice, one first computes the $\omega_k$ from the target inclusion
probabilities, and then draws the sample, as in the following R code:
\begingroup
\small
\begin{verbatim}
install.packages("sampling")
pik <- c(0.07, 0.17, 0.41, 0.61, 0.83, 0.91)
n <- sum(pik)
pikt <- sampling::UPMEpiktildefrompik(pik)
w <- pikt / (1 - pikt)
sampling::UPMEsfromq(sampling::UPMEqfromw(w, n))
\end{verbatim}
\endgroup

\subsection{Conditional inference}

Conditional inference in survey sampling has been studied by
\citet{rao:85}, \citet{dev:92}, and \citet{til:98,til:99}. To perform
conditional inference, we first derive a \emph{conditional design}; from it
we obtain conditional inclusion probabilities, which are in turn used to
build an estimator.

Let $\mathcal{Q}_1$ and $\mathcal{Q}_2$ be two supports with
$\mathcal{Q}_2\subset\mathcal{Q}_1$, and let $p(s;\mathcal{Q}_1)$ be a
sampling design on $\mathcal{Q}_1$. The conditional design given
$\mathcal{Q}_2$ is
$$
p(s;\mathcal{Q}_1\mid \mathcal{Q}_2) = \frac{p(s;\mathcal{Q}_1)} {\sum_{r\in \mathcal{Q}_2}
p(r;\mathcal{Q}_1)}, \qquad s \in \mathcal{Q}_2.
$$
From this conditional design we obtain the conditional inclusion
probabilities $\pi_k(\mathcal{Q}_2)$ and $\pi_{k\ell}(\mathcal{Q}_2)$.

The support $\mathcal{Q}_2$ may itself depend on the selected sample $S$:
for instance, $\mathcal{Q}_2$ may be the set of samples whose size equals
the size of the observed sample $S$. Given the conditional inclusion
probabilities, we estimate a population total by
$$
\widehat{Y}(\mathcal{Q}_2) = \sum_{k\in S} \frac{y_k}{\pi_k(\mathcal{Q}_2)}.
$$
This estimator is conditionally unbiased under $p(s;\mathcal{Q}_1\mid
\mathcal{Q}_2)$ if and only if $\pi_k(\mathcal{Q}_2) > 0$ for all $k\in U$.

\subsection{Conditional inference on a partition}

Let $U_1,\dots,U_H$ be a partition of $U$ into \emph{folds}, i.e.
$$
\bigcup_{h=1}^H U_h = U \qquad\text{and}\qquad U_h \cap U_i = \emptyset  \text{ for } h\neq i.
$$
For a sample $s$, write $n_h = |s\cap U_h|$, $h=1,\dots,H$, and
$n=\sum_{h=1}^H n_h$.

\begin{theorem}
If $p_{exp}(s,\mathcal{Q},\lambdag)$ is an exponential design on a symmetric
support $\mathcal{Q}$, then the design conditional on $n_1,\dots,n_H$
factorizes as a product of independent exponential designs on
$U_1,\dots,U_H$.
\end{theorem}

\begin{proof}
Write $\omega_k=\exp(\lambda_k)$, $k\in U$, and $s_h = s\cap U_h$. By
definition, the conditional design is
$$
p_{exp}(s,\mathcal{Q},\lambdag\mid n_1,\dots,n_H) = \frac{\prod_{k\in s} \omega_k}{\sum_{\substack{r \in
\mathcal{Q}\\ |r\cap U_h|=n_h, h=1,\dots,H}} \prod_{k\in r} \omega_k}.
$$
Because $\mathcal{Q}$ is symmetric, the set of samples in $\mathcal{Q}$
satisfying the $H$ size constraints coincides with the Cartesian product of
all subsets of each fold with the prescribed size:
$$
\{r \in \mathcal{Q}: |r\cap U_h|=n_h, h=1,\dots,H\} = \prod_{h=1}^H \{s_h \subset U_h: |s_h| = n_h\};
$$
symmetry of $\mathcal{Q}$ guarantees that every such combination of
per-fold samples indeed belongs to $\mathcal{Q}$. Consequently the
denominator factorizes:
$$
\sum_{\substack{r \in \mathcal{Q}\\ |r\cap U_h|=n_h}} \prod_{k\in r} \omega_k = \sum_{\substack{s_1\subset
U_1,\dots,s_H\subset U_H\\ |s_h|=n_h}} \prod_{h=1}^H \prod_{k\in s_h} \omega_k = \prod_{h=1}^H
\sum_{\substack{s_h \subset U_h\\ |s_h|=n_h}} \prod_{k\in s_h} \omega_k.
$$
Since the numerator also factorizes as $\prod_{k\in s}\omega_k =
\prod_{h=1}^H\prod_{k\in s_h}\omega_k$, we obtain
$$
p_{exp}(s,\mathcal{Q},\lambdag\mid n_1,\dots,n_H) = \prod_{h=1}^H \frac{\prod_{k\in s_h}
\omega_k}{\sum_{\substack{s_h \subset U_h\\ |s_h|=n_h}} \prod_{k\in s_h} \omega_k},
$$
which is a product of $H$ independent exponential designs, one on each
fold $U_h$.
\end{proof}

This result applies in particular to Bernoulli sampling, simple random
sampling, Poisson sampling, and conditional Poisson sampling.

\subsection{Examples of conditional designs}

\subsubsection{Example 1: conditioning on the $n_h$ in a Bernoulli design}

Under Bernoulli sampling, let $s_h = s\cap U_h$ and $n_h = |s_h|$,
$h=1,\dots,H$. Each $n_h$ has a binomial distribution,
$n_h \sim \mathrm{Bin}(N_h,\pi)$.

Since $|s| = \sum_h n_h$ is fixed once $n_1,\dots,n_H$ are fixed, the
factor $\pi^{|s|}(1-\pi)^{N-|s|}$ is the same for every $s$ satisfying the
constraints, so it cancels in the conditional design:
$$
\begin{aligned}
p(s\mid n_1,\dots, n_H) &= \frac{\pi^{|s|}(1-\pi)^{N-|s|}}{\sum_{\substack{r\\ |r_h|=n_h, h=1,\dots H
}}\pi^{|r|}(1-\pi)^{N-|r|}} \\
&= \frac{1}{\big|\{r : |r_h|=n_h, h=1,\dots H\}\big|} = \prod_{h=1}^H
\binom{N_h}{n_h}^{-1},
\end{aligned}
$$
for every $s$ such that $|s_h|=n_h$, $h=1,\dots,H$.

Conditionally on $n_1,\dots,n_H$, the design is therefore simple random
sampling without replacement within each post-stratum, independently across
strata, with conditional inclusion probabilities $n_h/N_h$ for $k\in U_h$.

\subsubsection{Example 2: conditioning on the $n_h$ in simple random sampling}

Under simple random sampling of fixed size $n$, let again $s_h=s\cap U_h$
and $n_h=|s_h|$, $h=1,\dots,H$. Here $(n_1,\dots,n_H)$ has a multivariate
hypergeometric distribution.

Since $p(s) = \binom{N}{n}^{-1}$ does not depend on $s$ (only on its size
$n$, which is fixed), the conditional design is again uniform over samples
satisfying the constraints:
$$
p(s\mid n_1,\dots, n_H) = \frac{\binom{N}{n}^{-1}}{\sum_{\substack{r\\ |r_h|=n_h, h=1,\dots H
}}\binom{N}{n}^{-1}} = \frac{1}{\big|\{r : |r_h|=n_h, h=1,\dots H\}\big|} = \prod_{h=1}^H
\binom{N_h}{n_h}^{-1},
$$
for every $s$ such that $|s_h|=n_h$, $h=1,\dots,H$.

As in Example 1, conditionally on $n_1,\dots,n_H$ the design is simple
random sampling without replacement within each post-stratum, independent
across strata, with conditional inclusion probabilities $n_h/N_h$ for
$k\in U_h$.

\subsubsection{Example 3: conditioning on the sample size in a Poisson design}

Suppose a sample is selected by Poisson sampling, design
\eqref{plandepoisson}. The design conditional on a fixed sample size is
$$
p_{ME}(s; \mathcal{S}\mid \mathcal{S}_n) = \frac{p_{ME}(s; \mathcal{S})} {\sum_{r\in \mathcal{S}_n} p_{ME}(r;
\mathcal{S})} = \frac{ \prod_{k\in s} \omega_k }{\sum_{r\in\mathcal{S}_n} \prod_{k\in r} \omega_k },
$$
which is exactly the CPS design of Section 1.2.5. Formally:
\begin{enumerate}
\item Select a random sample $S$ by Poisson sampling with inclusion probabilities $\pi_k$.
\item Compute $n = |S|$.
\item Compute $\omega_k = \pi_k/(1-\pi_k)$, $k\in U$.
\item Compute the inclusion probabilities $\pi_k(\mathcal{S}_n)$ of the CPS design of size $n$ with parameters $\omega_k$.
\item Use $\pi_k(\mathcal{S}_n)$ in the estimator.
\end{enumerate}

The following R code implements this construction, using the
\textsf{sampling} package of \citet{til:mat:25} and the
\textsf{StratifiedSampling} package of \citet{StratifiedSampling2025}.
\begingroup
\small
\begin{verbatim}
# install.packages("sampling")
# install.packages("StratifiedSampling")
###########################################
rm(list = ls())

cond_incl_prob <- function(pik, n) {
  N <- length(pik)
  if (n == 0) {
    pikn <- rep(0, N)
  } else if (n == N) {
    pikn <- rep(1, N)
  } else {
    TEST <- pik < 1
    n1 <- n - (length(TEST) - sum(TEST))
    w <- pik[TEST] / (1 - pik[TEST])
    pikn <- pik
    pikn[TEST] <- StratifiedSampling::pikfromq(StratifiedSampling::qfromw(w, n1))
  }
  pikn
}

###########################################
strat_from_Poisson <- function(pik, s, STRAT) {
  STRAT <- sampling::cleanstrata(STRAT)
  H <- max(STRAT)
  TEST <- pik < 1
  pikn <- pik
  STRATR <- STRAT[TEST]
  nh <- as.vector(t(s[TEST]) %*% sampling::disjunctive(STRATR))
  for (h in 1:H) {
    pikn[TEST][STRATR == h] <- cond_incl_prob(pik[TEST][STRATR == h], nh[h])
  }
  pikn
}

###########################################
data(MU284, package = "sampling")
n <- 100
pik <- sampling::inclusionprobabilities(MU284$P75, n)

#############################################
SIM <- 1000
EHT <- EC <- rep(0, SIM)
y <- MU284$REV84 / 10000

for (i in 1:SIM) {
  s <- sampling::UPpoisson(pik)
  n <- sum(s)
  pikn <- cond_incl_prob(pik, n)
  EHT[i] <- sum(s * y / pik)
  EC[i] <- sum(s * y / pikn)
}

boxplot(EHT, EC, names = c("Horvitz-Thompson", "Conditional design"))
var(EC) / var(EHT)
\end{verbatim}
\endgroup

\subsubsection{Example 4: conditioning on the $n_h$ in a Poisson design}

Suppose again a sample is selected by Poisson sampling, design
\eqref{plandepoisson}. The design conditional on fixed sample sizes
$n_1,\dots,n_H$ in each stratum is
$$
p_{ME}(s; \mathcal{S}\mid \mathcal{S}_{\mathrm{strat}(n_1,\dots,n_H)}) = \prod_{h=1}^H \frac{ \prod_{k\in
(s\cap U_h)} \omega_k }      {  \sum_{\substack{s_h\subset U_h\\ |s_h|=n_h }} \prod_{k\in s_h} \omega_k },
$$
i.e.\ the maximum entropy stratified design. Formally:
\begin{enumerate}
\item Select a random sample $S$ by Poisson sampling with inclusion probabilities $\pi_k$.
\item Compute $\omega_k = \pi_k/(1-\pi_k)$, $k\in U$.
\item Compute $n_h=|U_h \cap S|$, $h=1,\dots,H$.
\item In each stratum, compute the inclusion probabilities $\pi_k(\mathcal{S}\mid \mathcal{S}_{\mathrm{strat}(n_1,\dots,n_H)})$ of the maximum entropy stratified design of sizes $n_1,\dots,n_H$, using the parameters $\omega_k$.
\item Use $\pi_k(\mathcal{S}\mid \mathcal{S}_{\mathrm{strat}(n_1,\dots,n_H)})$ in the estimator.
\end{enumerate}

Again using the \textsf{sampling} package of \citet{til:mat:25} and the
\textsf{StratifiedSampling} package of \citet{StratifiedSampling2025}:
\begingroup
\small
\begin{verbatim}
############################################
CAT <- cut(y, breaks = quantile(y, probs = seq(0, 1, by = 1/3), na.rm = TRUE),
           include.lowest = TRUE, labels = 1:3)

SIM <- 1000
EHT <- ESTRP <- rep(0, SIM)
y <- MU284$REV84 / 10000
STRAT <- MU284$REG
STRAT <- CAT

for (i in 1:SIM) {
  s <- sampling::UPpoisson(pik)
  pikn <- strat_from_Poisson(pik, s, STRAT)
  EHT[i] <- sum((y / pik)[s == 1])
  ESTRP[i] <- sum((y / pikn)[s == 1])
}

boxplot(EHT, ESTRP, names = c("Horvitz-Thompson", "Conditional design"))
var(ESTRP) / var(EHT)
\end{verbatim}
\endgroup

\subsubsection{Example 5: conditioning on the $n_h$ in a CPS design}

Suppose a sample of fixed size $n$ is selected by CPS, design
\eqref{cpsdesign}. The design conditional on fixed sample sizes
$n_1,\dots,n_H$ in each stratum (necessarily with $n_1+\cdots+n_H=n$) is
$$
p_{ME}(s; \mathcal{S}_n\mid \mathcal{S}_{\mathrm{strat}(n_1,\dots,n_H)}) = \prod_{h=1}^H \frac{ \prod_{k\in
(s\cap U_h)} \omega_k }      {  \sum_{\substack{s_h\subset U_h\\ |s_h|=n_h }} \prod_{k\in s_h} \omega_k },
$$
again the maximum entropy stratified design. Formally:
\begin{enumerate}
\item Select a random sample $S$ by CPS with inclusion probabilities $\pi_k$.
\item Compute $\omega_k$ from $\pi_k$ using one of the methods in \citet{tille2006sampling}.
\item Compute $n_h=|U_h \cap S|$, $h=1,\dots,H$.
\item In each stratum, compute the inclusion probabilities $\pi_k(\mathcal{S}\mid \mathcal{S}_{\mathrm{strat}(n_1,\dots,n_H)})$ of the maximum entropy stratified design of sizes $n_1,\dots,n_H$, using the parameters $\omega_k$.
\item Use $\pi_k(\mathcal{S}\mid \mathcal{S}_{\mathrm{strat}(n_1,\dots,n_H)})$ in the estimator.
\end{enumerate}

Again using the \textsf{sampling} package of \citet{til:mat:25} and the
\textsf{StratifiedSampling} package of \citet{StratifiedSampling2025}:
\begingroup
\small
\begin{verbatim}
##################################################################
pikt <- pik
TEST <- pik < 1
pikt[TEST] <- StratifiedSampling::piktfrompik(pik[TEST])

SIM <- 1000
EHT <- ESTRP <- rep(0, SIM)
y <- MU284$REV84 / 10000
STRAT <- MU284$REG
STRAT <- CAT

for (i in 1:SIM) {
  print(i)
  s <- sampling::UPmaxentropy(pik)
  pikn <- strat_from_Poisson(pikt, s, STRAT)
  EHT[i] <- sum((y / pik)[s == 1])
  ESTRP[i] <- sum((y / pikn)[s == 1])
}

boxplot(EHT, ESTRP, names = c("Horvitz-Thompson", "Conditional stratified"))
var(ESTRP) / var(EHT)
\end{verbatim}
\endgroup

\Black

\section{Proofs of Section \ref{sec:crossfit}}

\subsection{Proof of Theorem \ref{theo:condIndExp}} \label{Prooftheo:condIndExp}

We prove the two implications separately. First, we show that exponential designs factorize after conditioning on fold sample sizes. Second, we show that conditional independence forces the design to be exponential.

\subsubsection*{Step 1: Exponential designs factorize}

Assume (ii), and write $\omega_k(n)=\exp(\lambda_k(n))$ for the weights of the exponential design at each possible sample size $n$, with 
$$
\omega_k(n) =
\frac{ \omega_k }{\left(\sum_{r\in\mathcal{S}_n} \prod_{k\in r} \omega_k\right)^{1/n} }.
$$
Then, exponential designs factorize as
$$
p(s;\mathcal{S}_n) = \prod_{k\in s} \omega_k(n) .
$$
\Black Let $\boldsymbol{d}$ be a realization of the partition $\boldsymbol{\delta}$, and let $\boldsymbol{m}$ satisfy $\P(\boldsymbol{n}=\boldsymbol{m},\boldsymbol{\delta}=\boldsymbol{d})>0$. Set  $E_{\boldsymbol{m}, \boldsymbol{d}} := \left\{ r \in \cP(U) : |r \cap U_j(\boldsymbol{d})| = m_j,  j=1,...,M \right\}$ be the set of compatible samples. For $s \in E_{\boldsymbol{m}, \boldsymbol{d}}$, we first compute the conditional sample probability:
\begin{align*}
&\P\left(S=s\mid\boldsymbol{n}=\boldsymbol{m},\boldsymbol{\delta}=\boldsymbol{d}\right)\\
&\qquad=\frac{\P(S=s,\boldsymbol{n}=\boldsymbol{m},\boldsymbol{\delta}=\boldsymbol{d})}
{\P(\boldsymbol{n}=\boldsymbol{m},\boldsymbol{\delta}=\boldsymbol{d})}\\
&\qquad=\frac{\P(S=s,\boldsymbol{\delta}=\boldsymbol{d})}
{\P(\boldsymbol{n}=\boldsymbol{m},\boldsymbol{\delta}=\boldsymbol{d})}.
\end{align*}
The last equality holds because $S=s$ and $\boldsymbol{\delta}=\boldsymbol{d}$ give $n_j=|s\cap U_j(\boldsymbol{d})|=m_j$ for every $j$. For the denominator, summing over all possible samples gives
\begin{align*}
\P(\boldsymbol{n}=\boldsymbol{m},\boldsymbol{\delta}=\boldsymbol{d})
&=\sum_{r\in\cP(U)}
\P(S=r,\boldsymbol{n}=\boldsymbol{m},\boldsymbol{\delta}=\boldsymbol{d})\\
&=\sum_{r\in E_{\boldsymbol{m},\boldsymbol{d}}}
\P(S=r,\boldsymbol{\delta}=\boldsymbol{d})\\
&=\P(\boldsymbol{\delta}=\boldsymbol{d})
\sum_{r\in E_{\boldsymbol{m},\boldsymbol{d}}}p(r).
\end{align*}
Here the samples outside $E_{\boldsymbol{m},\boldsymbol{d}}$ contribute zero, and the last equality uses $S\indep\boldsymbol{\delta}$. The same independence gives
$$\P(S=s,\boldsymbol{\delta}=\boldsymbol{d})
=p(s)\P(\boldsymbol{\delta}=\boldsymbol{d}).$$
Substituting and cancelling $\P(\boldsymbol{\delta}=\boldsymbol{d})$, we obtain
\begin{align*}
&\P\left(S=s\mid\boldsymbol{n}=\boldsymbol{m},\boldsymbol{\delta}=\boldsymbol{d}\right)\\
&\qquad=\frac{p(s)}{\sum_{r\in E_{\boldsymbol{m},\boldsymbol{d}}}p(r)}\\
&\qquad=\frac{\P(S=s\mid |S|=n)}
{\sum_{r\in E_{\boldsymbol{m},\boldsymbol{d}}}\P(S=r\mid |S|=n)}\\
&\qquad=\frac{\prod_{k\in s}\omega_k(n)}
{\sum_{r\in E_{\boldsymbol{m},\boldsymbol{d}}}\prod_{k\in r}\omega_k(n)}.
\end{align*}
Every compatible sample has size $n$, so the common factor $\P(|S|=n)$ cancels in the second equality. The common normalizing constant of the exponential design cancels in the third.

Since the folds form a partition, each compatible sample is uniquely determined by its intersections with the folds:
$$E_{\boldsymbol{m},\boldsymbol{d}}
=\left\{\bigcup_{j=1}^M r_j:r_j\subseteq U_j(\boldsymbol{d}),
|r_j|=m_j, j=1,\ldots,M\right\}.$$
Thus, the denominator factors as
\begin{align*}
\sum_{r\in E_{\boldsymbol{m},\boldsymbol{d}}}\prod_{k\in r}\omega_k(n)
&=\sum_{\substack{r_1\subseteq U_1(\boldsymbol{d})\\|r_1|=m_1}}
\cdots
\sum_{\substack{r_M\subseteq U_M(\boldsymbol{d})\\|r_M|=m_M}}
\prod_{j=1}^M\prod_{k\in r_j}\omega_k(n)\\
&=\prod_{j=1}^M
\sum_{\substack{r_j\subseteq U_j(\boldsymbol{d})\\|r_j|=m_j}}
\prod_{k\in r_j}\omega_k(n).
\end{align*}
The numerator also factors over the folds, giving
$$\P\left(S=s\mid\boldsymbol{n}=\boldsymbol{m},\boldsymbol{\delta}=\boldsymbol{d}\right)
=\prod_{j=1}^M
\frac{\prod_{k\in s\cap U_j(\boldsymbol{d})}\omega_k(n)}
{\sum_{\substack{r_j\subseteq U_j(\boldsymbol{d})\\|r_j|=m_j}}
\prod_{k\in r_j}\omega_k(n)}.$$
This proves (i).

\subsubsection*{Step 2: Conditional independence given fold sizes holds only for exponential designs}

The following lemma characterizes exponential designs through ratios of sample probabilities.

\begin{lemma}\label{lemma:proof1}
   Suppose that $p$ has symmetric support, let $2\leq n\leq N-2$ be a possible sample size, and set $p_n(s) = \P[S = s  |  |S| = n]$, with $\cS_n = \{ s \in \cP(U) : |s| = n\}$. Then, the following statements are equivalent.
   \begin{enumerate}
       \item[(i)] For every pair of distinct elements $k,l\in U$, the ratio $$ \rho_{kl} = \dfrac{p_n(r \cup \{k\})}{p_n(r \cup \{l\})}$$ does not depend on $r\in \{ s \subset U\backslash \{k,l\} : |s| = n-1\}$.
       \item[(ii)] The sampling design $p_n$ is exponential.
   \end{enumerate}
\end{lemma}
\begin{proof}
    \emph{Proof of $(ii) \Rightarrow (i)$.} Assume that the design $p_n$ is exponential. Then, for some positive weights $\omega_k(n)$, $$p_n(s)\propto \prod_{k\in s} \omega_k(n).$$ Hence, $$\rho_{kl} = \dfrac{ \omega_k(n)\prod_{i\in r} \omega_i(n)}{\omega_l(n)\prod_{i\in r} \omega_i(n) }= \dfrac{\omega_k(n)}{\omega_l(n)},$$ from which the implication follows.

        \emph{Proof of $(i) \Rightarrow (ii)$.} Assume (i). For distinct $k,l,m\in U$, choose $r \subseteq U \backslash \{k,l,m\}$ with $|r|=n-1$; this is possible since $n\leq N-2$. Then, $$ \rho_{km} = \dfrac{p_n(r \cup \{k\})}{p_n(r \cup \{m\})} = \dfrac{p_n(r \cup \{k\})}{p_n(r \cup \{l\})} \times \dfrac{p_n(r \cup \{l\})}{p_n(r \cup \{m\})} = \rho_{kl} \rho_{lm}. $$ We now construct the weights. For an arbitrary $k_0 \in U$, let $\omega_{k_0}(n) = 1$, and $\omega_k(n) = \rho_{kk_0}$ for $k \in U \backslash \{k_0\}$. For distinct $k,l\in U\backslash\{k_0\}$, the identity above gives $\omega_k(n) = \rho_{kl}\rho_{lk_0} =\rho_{kl} \omega_l(n)$, and hence $\rho_{kl} = \omega_k(n)/\omega_l(n)$. The same equality is immediate if $k$ or $l$ equals $k_0$.

        Now, let $s, t \in \cS_n$. Because both are of the same size, the sets $s \backslash t$ and $t \backslash s$ contain the same number of elements $L := n - |s\cap t|$; write $s \backslash t = \{a_1, ..., a_L\}$ and $t \backslash s = \{b_1, ..., b_L\}$. Starting from $t_0=t$, set $t_q=(t_{q-1}\backslash\{b_q\})\cup\{a_q\}$, so that $t_L=s$. At each exchange, the common set $t_{q-1}\backslash\{b_q\}$ has size $n-1$ and contains neither $a_q$ nor $b_q$. Hence,
        $$\frac{p_n(t_q)}{p_n(t_{q-1})}
        =\frac{p_n((t_{q-1}\backslash\{b_q\})\cup\{a_q\})}
        {p_n((t_{q-1}\backslash\{b_q\})\cup\{b_q\})}
        =\rho_{a_qb_q}=\frac{\omega_{a_q}(n)}{\omega_{b_q}(n)}.$$
        Multiplying these ratios gives
        $$\dfrac{p_n(s)}{p_n(t)} = \dfrac{p_n(t_1)}{p_n(t_0)}\dfrac{p_n(t_2)}{p_n(t_1)} ...  \dfrac{p_n(t_L)}{p_n(t_{L-1})} = \prod_{k=1}^L \dfrac{\omega_{a_k}(n)}{\omega_{b_k}(n)}= \dfrac{\prod_{k\in s} \omega_k(n)}{\prod_{k\in t} \omega_k(n)}.$$
        Thus, $p_n(s)/\prod_{k\in s} \omega_k(n)$ is constant over $\cS_n$. Since $\sum_{s\in\cS_n}p_n(s)=1$, we get
        $$p_n(s) = \dfrac{\prod_{k\in s} \omega_k(n)}{\sum_{t \in \cS_n} \prod_{k\in t} \omega_k(n)}.$$
        Since $\omega_k(n)>0$, setting $\lambda_k(n)=\log \omega_k(n)$ shows that the design is exponential.
\end{proof}

\paragraph*{Proof of Step 2.} Assume (i). Fix a possible sample size $n$ and write $p_n(s)=\P(S=s\mid |S|=n)$. If $n=0$ or $n=N$, then $p_n$ is a point mass and is therefore exponential. If $n=1$, take $\omega_k(1)=p_1(\{k\})$. If $n=N-1$, take $\omega_k(N-1)=1/p_{N-1}(U\backslash\{k\})$. These choices give the exponential form directly. Thus, suppose that $2\leq n\leq N-2$.

Let $k,l \in U$ be distinct and let $r \in \{ s \subset U\backslash \{k,l\} : |s| = n-1\}$. We show that the ratio $$ \rho_{kl}(r)  = \dfrac{p_n(r \cup \{k\})}{p_n(r \cup \{l\})}$$ does not depend on $r$. We first compare samples that differ by one element, then use a sequence of such exchanges.

\emph{Step 2. a.} Let $r' \in \{ s \subset U\backslash \{k,l\} : |s| = n-1\}$ differ from $r$ by one unit, so that $r'=(r\backslash\{a\})\cup\{b\}$ for some elements $a,b$. The elements $k,l,a,b$ are distinct. Since $N_1,N_2\geq 2$ and the random partition has full support, we can choose a realization $\boldsymbol{d}$ such that $k,l \in U_1(\boldsymbol{d})$ and $a,b \in U_2(\boldsymbol{d})$. Consider the four samples $s_k = r \cup \{k\}$, $s_l = r \cup \{l\}$, $s_k' = r' \cup \{k\}$, and $s_l' = r' \cup \{l\}$. They have the same fold-count vector, which we denote by $\boldsymbol{m}$. Moreover, $\P(\boldsymbol{n}=\boldsymbol{m},\boldsymbol{\delta}=\boldsymbol{d})>0$ because symmetry gives $p_n$ full support on $\cS_n$, the distribution of $\boldsymbol{\delta}$ has full support, and $\boldsymbol{\delta}\indep S$.

For $s\in E_{\boldsymbol{m},\boldsymbol{d}}$, set $$Q(s) := \P \left( S = s  \big\rvert  \boldsymbol{n} = \boldsymbol{m}, \boldsymbol{\delta} = \boldsymbol{d}\right).$$ We compare $Q(s_k)/Q(s_l)$ and $Q(s_k')/Q(s_l')$. Conditional independence between fold one and its complement gives
\begin{align*}
Q(s)
&=\P\left(S\cap U_1(\boldsymbol{d})=s\cap U_1(\boldsymbol{d})
\mid\boldsymbol{n}=\boldsymbol{m},\boldsymbol{\delta}=\boldsymbol{d}\right)\times\P\left(S\backslash U_1(\boldsymbol{d})=s\backslash U_1(\boldsymbol{d})
\mid\boldsymbol{n}=\boldsymbol{m},\boldsymbol{\delta}=\boldsymbol{d}\right).
\end{align*}
In each ratio, the two samples agree outside $U_1(\boldsymbol{d})$, so the second factor cancels. Since $r'\cap U_1(\boldsymbol{d})=r\cap U_1(\boldsymbol{d})$, the remaining factors give
\begin{align*}
\frac{Q(s_k)}{Q(s_l)}
&=\frac{\P\left(S\cap U_1(\boldsymbol{d})=(r\cap U_1(\boldsymbol{d}))\cup\{k\}
\mid\boldsymbol{n}=\boldsymbol{m},\boldsymbol{\delta}=\boldsymbol{d}\right)}
{\P\left(S\cap U_1(\boldsymbol{d})=(r\cap U_1(\boldsymbol{d}))\cup\{l\}
\mid\boldsymbol{n}=\boldsymbol{m},\boldsymbol{\delta}=\boldsymbol{d}\right)}\\
&=\frac{Q(s_k')}{Q(s_l')}.
\end{align*}
Now, on $E_{\boldsymbol{m}, \boldsymbol{d}}$,  
\begin{align*}
    Q(s) &= \dfrac{p(s) \P \left[ \boldsymbol{\delta} = \boldsymbol{d}  \big\rvert S = s\right]}{\sum_{r \in E_{\boldsymbol{m}, \boldsymbol{d}}} p(r) \P \left[ \boldsymbol{\delta} = \boldsymbol{d}  \big\rvert S = r\right]}\\
    &= \dfrac{p(s) }{\sum_{r \in E_{\boldsymbol{m}, \boldsymbol{d}}} p(r) } \qquad \text{(Using that $\boldsymbol{\delta} \indep S$)}\\
    &= \dfrac{p_n(s) }{\sum_{r \in E_{\boldsymbol{m}, \boldsymbol{d}}} p_n(r) }. \qquad \text{(Since $p(s) = p_n(s) \P(|S|=n)$ for $s \in  E_{\boldsymbol{m}, \boldsymbol{d}}$)}
\end{align*}
Therefore, $\rho_{kl}(r) = \rho_{kl}(r')$ when $r$ and $r'$ differ by one element.

\emph{Step 2. b.} Finally, let $r, r'\subset U\backslash \{k,l\}$ have size $n-1$, and set $L=n-1-|r\cap r'|$. We can construct $r'$ from $r$ by a sequence of $L$ one-element swaps, which defines a sequence $r_0 = r$, $r_1$, ..., $r_L = r'$ where $$\rho_{kl}(r)=\rho_{kl}(r_0)=\rho_{kl}(r_1)=...=\rho_{kl}(r_L)=\rho_{kl}(r').$$ This follows by applying Step 2.a. to each one-element swap. Thus, $\rho_{kl}(r)$ does not depend on $r$. Lemma \ref{lemma:proof1} shows that $p_n$ is exponential. Since $n$ was arbitrary, this proves (ii).

\subsection{Proof of Theorem \ref{Theo:eqAwareAgno}} \label{ProofTheo:eqAwareAgno}

Fix a partition $\boldsymbol{d}$ such that $\P(\boldsymbol{n}(S,\boldsymbol{d})=\boldsymbol{n})>0$ and a sample $s\in\cP(U)$. We show that the two schemes give the same conditional sample probabilities.

For the design-aware scheme, the numbers of ways to partition the sampled and nonsampled units are, respectively, $$\dfrac{n!}{n_1!n_2! \ldots n_M!} \qquad \text{and} \qquad \dfrac{(N-n)!}{(N_1-n_1)!(N_2-n_2)! \ldots (N_M-n_M)!}.$$ Hence, the total number of possible pairs of partitions in Algorithm \ref{alg:design-aware} is
$$K_{\rm aware} = \dfrac{n!}{n_1!n_2! \ldots n_M!}\dfrac{(N-n)!}{(N_1-n_1)!(N_2-n_2)! \ldots (N_M-n_M)!}.$$
Given $S=s$, Algorithm \ref{alg:design-aware} can produce $\boldsymbol{d}$ if and only if $\boldsymbol{n}(s,\boldsymbol{d})=\boldsymbol{n}$. In that case, the sampled and nonsampled parts are uniquely determined by $S_j=s\cap U_j(\boldsymbol{d})$ and $R_j=U_j(\boldsymbol{d})\backslash s$, for $j=1,\ldots,M$. Since all pairs of partitions are equally likely,
$$\P(S=s,\boldsymbol{\delta}_{\rm aware}=\boldsymbol{d})=\dfrac{p(s)\mathds{1}_{\boldsymbol{n}(s,\boldsymbol{d})=\boldsymbol{n}}}{K_{\rm aware}}.$$
Summing over all samples gives
$$\P(\boldsymbol{\delta}_{\rm aware}=\boldsymbol{d})=\dfrac{1}{K_{\rm aware}}\sum_{r\in\cP(U)}p(r)\mathds{1}_{\boldsymbol{n}(r,\boldsymbol{d})=\boldsymbol{n}}.$$
Therefore,
$$\P \left( S = s  \big\rvert  \boldsymbol{\delta}_{\rm aware} = \boldsymbol{d}\right) = \dfrac{p(s)\mathds{1}_{\boldsymbol{n}(s,\boldsymbol{d})=\boldsymbol{n}}}{\sum_{r \in \cP(U)}p(r)\mathds{1}_{\boldsymbol{n}(r,\boldsymbol{d})=\boldsymbol{n}}}.$$

For the uniform scheme, we use the independence of $\boldsymbol{\delta}_{\rm unif}$ and $S$ to obtain
\begin{align*}
&\P \left( S=s  \big\rvert  \boldsymbol{\delta}_{\rm unif} = \boldsymbol{d}, \boldsymbol{n}(S, \boldsymbol{d}) = \boldsymbol{n}\right)\\
&\qquad=\dfrac{\P \left( S=s, \boldsymbol{n}(S, \boldsymbol{d}) = \boldsymbol{n}  \big\rvert  \boldsymbol{\delta}_{\rm unif} = \boldsymbol{d}\right)}{\P \left( \boldsymbol{n}(S, \boldsymbol{d}) = \boldsymbol{n}  \big\rvert  \boldsymbol{\delta}_{\rm unif} = \boldsymbol{d} \right)}\\
&\qquad=\dfrac{\P(S=s\mid\boldsymbol{\delta}_{\rm unif}=\boldsymbol{d})\mathds{1}_{\boldsymbol{n}(s,\boldsymbol{d})=\boldsymbol{n}}}
{\sum_{r\in\cP(U)}\P(S=r\mid\boldsymbol{\delta}_{\rm unif}=\boldsymbol{d})\mathds{1}_{\boldsymbol{n}(r,\boldsymbol{d})=\boldsymbol{n}}}\\
&\qquad=\dfrac{p(s)\mathds{1}_{\boldsymbol{n}(s,\boldsymbol{d})=\boldsymbol{n}}}{\sum_{r\in\cP(U)}p(r)\mathds{1}_{\boldsymbol{n}(r,\boldsymbol{d})=\boldsymbol{n}}}\\
&\qquad=\P \left( S = s  \big\rvert  \boldsymbol{\delta}_{\rm aware} = \boldsymbol{d}\right).
\end{align*}
Thus, the two conditional laws are equal. It follows that one law factorizes across folds if and only if the other does. Applying this result to every target vector gives the stated equivalence between the two conditional independence properties. The equality of the first- and second-order conditional inclusion probabilities follows by summing the common probabilities over samples containing $k$ and samples containing both $k$ and $l$, respectively.

\subsection{Proof of Proposition \ref{propAssumptions}} \label{proofpropAssumptions}

By Theorem~\ref{Theo:eqAwareAgno}, under SRSWOR, all samples with the prescribed fold counts have the same conditional probability given $\cA$. There are $\prod_{j=1}^M\binom{N_{j,v}}{n_{j,v}}$ such samples, since the units in each fold can be chosen separately. Thus, for every compatible $s$,
$$
\P(S_v=s\mid\cA) =\frac{1}{\prod_{j=1}^M\binom{N_{j,v}}{n_{j,v}}}
=\prod_{j=1}^M\frac{1}{\binom{N_{j,v}}{n_{j,v}}}.
$$
Each factor is the probability of an SRSWOR sample in its fold. Hence, the fold samples are conditionally independent, the scheme is honest, and
$$
\pi_{k|\cA}=\frac{n_{j,v}}{N_{j,v}}, \qquad k\in U_{j,v}.
$$
For distinct $k,l\in U_{j,v}$,
$$
\pi_{kl|\cA}=\frac{n_{j,v}(n_{j,v}-1)}{N_{j,v}(N_{j,v}-1)}, \qquad
\Delta_{kl|\cA}=-\frac{(n_{j,v}/N_{j,v})(1-n_{j,v}/N_{j,v})}{N_{j,v}-1}.
$$
For units in different folds, $\pi_{kl|\cA}=\pi_{k|\cA}\pi_{l|\cA}$ and $\Delta_{kl|\cA}=0$.

The fold-size condition gives \ref{CF1}. By \ref{D1} and the allocation condition, $n_{j,v}/N_{j,v}\xrightarrow[v \to \infty]{}\pi_\star>0$, uniformly over the fixed number of folds. Since $n_{j,v}\geq1$ for every $j,v$, these ratios are bounded below by a positive constant over all $j,v$. This proves \ref{CF2b}, \ref{CF2c}, and \ref{CF6}. The second-order probabilities for distinct units also converge uniformly to $\pi_\star^2$, which gives \ref{CF2a}.

For \ref{CF3}, use $\pi_k=n_v/N_v$ and the allocation condition to obtain
$$
\E_p\left[\max_{k\in U_v}(\pi_{k|\cA}-\pi_k)^2\right] =\max_{1\leq j\leq
M}\left(\frac{n_{j,v}}{N_{j,v}}-\frac{n_v}{N_v}\right)^2 =\mathcal{O}(n_v^{-1}).
$$
For \ref{CF4}, the covariance formula gives, for all sufficiently large $v$,
$$n_v\max_{k\neq l\in U_v}|\Delta_{kl|\cA}|
\leq\frac{n_v}{\min_j N_{j,v}-1}=\mathcal{O}(1),$$
since each $N_{j,v}$ is of order $N_v$ and $n_v\leq N_v$. The constants in both bounds can be enlarged to cover the finitely many initial populations.

It remains to verify \ref{CF5}. For $1\leq r\leq4$ distinct units $k_1,\ldots,k_r$ in the same fold, and all sufficiently large $v$,
$$
\E_p\left[\prod_{t=1}^r I_{k_t}\Bigm|\cA\right] =\prod_{t=0}^{r-1}\frac{n_{j,v}-t}{N_{j,v}-t}
=\left(\frac{n_{j,v}}{N_{j,v}}\right)^r+\mathcal{O}(N_{j,v}^{-1}).
$$
For the last bound, all factors are in $[0,1]$, so replacing them one at a time gives
$$
\left|\prod_{t=0}^{r-1}\frac{n_{j,v}-t}{N_{j,v}-t} -\left(\frac{n_{j,v}}{N_{j,v}}\right)^r\right|
\leq\sum_{t=0}^{r-1}\left|\frac{n_{j,v}-t}{N_{j,v}-t}-\frac{n_{j,v}}{N_{j,v}}\right|
\leq\sum_{t=0}^{r-1}\frac{t}{N_{j,v}-t} =\mathcal{O}(N_{j,v}^{-1}).
$$
By independence across folds, both $\E_p[I_iI_jI_kI_l\mid\cA]$ and $\pi_{ij|\cA}\pi_{kl|\cA}$ differ from $\pi_{i|\cA}\pi_{j|\cA}\pi_{k|\cA}\pi_{l|\cA}$ by $\mathcal{O}(N_v^{-1})$, since each $N_{j,v}$ is of order $N_v$. Subtracting gives, uniformly over four distinct units,
$$
\E_p\left[(I_iI_j-\pi_{ij|\cA})(I_kI_l-\pi_{kl|\cA})\Bigm|\cA\right] =\mathcal{O}(N_v^{-1})=o(1).
$$
This proves \ref{CF5} and completes the proof.

\section{Proofs of Section \ref{sec:cdtMA}}

We first give a second-moment bound for centered sums within the folds.

\begin{lemma}\label{lemma:secondMoment}
Assume \textnormal{(Ind-$\cA$)} and \ref{CF4}. Let $(d_k)_{k\in U_v}$ be square-integrable random variables such that, for each $j$, $(d_k)_{k\in U_{j,v}}$ is measurable with respect to $\cA$ and $\{I_q\}_{q\in U_v\backslash U_{j,v}}$. Then,
$$
n_v\E_p\left[\left\{\frac{1}{N_v}\sum_{k\in U_v}(I_k-\pi_{k|\cA})d_k\right\}^2\right]
\leq\frac{C}{N_v}\sum_{k\in U_v}\E_p(d_k^2),
$$
where $C$ depends only on $M$ and $C_{\Delta 2}$.
\end{lemma}
\begin{proof}
Conditional on $\cA$ and the sampling indicators outside $U_{j,v}$, the coefficients $d_k$, $k\in U_{j,v}$, are fixed. By \textnormal{(Ind-$\cA$)},
\begin{align*}
&\E_p\left[\left\{\sum_{k\in U_{j,v}}(I_k-\pi_{k|\cA})d_k\right\}^2
\Bigm|\cA,\{I_q\}_{q\in U_v\backslash U_{j,v}}\right]\\
&\quad=\sum_{k,l\in U_{j,v}}d_kd_l
\E_p\left[(I_k-\pi_{k|\cA})(I_l-\pi_{l|\cA})
\Bigm|\cA,\{I_q\}_{q\in U_v\backslash U_{j,v}}\right]\\
&\quad=\sum_{k,l\in U_{j,v}}\Delta_{kl|\cA}d_kd_l\\
&\quad\leq\frac{1}{4}\sum_{k\in U_{j,v}}d_k^2
+\frac{C_{\Delta 2}}{n_v}\left(\sum_{k\in U_{j,v}}|d_k|\right)^2\\
&\quad\leq\left(\frac{1}{4}+\frac{C_{\Delta 2}N_{j,v}}{n_v}\right)\sum_{k\in U_{j,v}}d_k^2,
\end{align*}
where we used $\Delta_{kk|\cA}\leq1/4$, \ref{CF4}, and the Cauchy--Schwarz inequality. For the sum over all folds, Cauchy--Schwarz and the preceding bound give
\begin{align*}
&n_v\E_p\left[\left\{\frac{1}{N_v}\sum_{k\in U_v}(I_k-\pi_{k|\cA})d_k\right\}^2\right]\\
&\quad\leq\frac{n_vM}{N_v^2}\sum_{j=1}^M
\E_p\left[\left\{\sum_{k\in U_{j,v}}(I_k-\pi_{k|\cA})d_k\right\}^2\right]\\
&\quad\leq\frac{M}{N_v^2}\sum_{j=1}^M
\E_p\left[\left(\frac{n_v}{4}+C_{\Delta 2}N_{j,v}\right)
\sum_{k\in U_{j,v}}d_k^2\right]\\
&\quad\leq\frac{C}{N_v}\sum_{k\in U_v}\E_p(d_k^2),
\end{align*}
where we used $n_v\leq N_v$ and $N_{j,v}\leq N_v$.
\end{proof}

\subsection{Proof of Proposition \ref{prop:unb}} \label{Proofprop:unb}

By the law of total expectation, it is enough to show that $\E_p[\widehat{\mu}_{cf}(\widehat{m},\pi_{|\cA})\mid\cA]=\mu$. Since $\widehat{m}_j$ is fitted using only the observations outside $U_j$, it is measurable with respect to $\cA$ and $\{I_l\}_{l\in U\backslash U_j}$. Moreover, by \textnormal{(Ind-$\cA$)},
$$
\E_p\left[I_k\mid\cA,\{I_l\}_{l\in U\backslash U_j}\right]
=\E_p[I_k\mid\cA]
=\pi_{k|\cA},
\qquad k\in U_j.
$$
Consequently, for $k\in U_j$,
\begin{align*}
&\E_p\left[\widehat{m}_j(\bx_k)
+I_k\frac{y_k-\widehat{m}_j(\bx_k)}{\pi_{k|\cA}}
\Bigm|\cA,\{I_l\}_{l\in U\backslash U_j}\right]\\
&\quad=\widehat{m}_j(\bx_k)
+\pi_{k|\cA}\frac{y_k-\widehat{m}_j(\bx_k)}{\pi_{k|\cA}}
=y_k.
\end{align*}
Therefore, by the tower property,
\begin{align*}
\E_p\left[\widehat{\mu}_{cf}(\widehat{m},\pi_{|\cA})\mid\cA\right]
&=\frac{1}{N}\sum_{j=1}^M\sum_{k\in U_j}
\E_p\left[\widehat{m}_j(\bx_k)+I_k\frac{y_k-\widehat{m}_j(\bx_k)}{\pi_{k|\cA}}\Bigm|\cA\right]\\
&=\frac{1}{N}\sum_{j=1}^M\sum_{k\in U_j}y_k
=\mu.
\end{align*}

\subsection{Proof of Theorem \ref{th:equivalent1}} \label{Proofth:equivalent1}

Let $\cE_v:=\{\min_{k\in U_v}\pi_{k|\cA}\geq\lambda_{cf}\}$ and let $\epsilon>0$. By Markov's inequality,
\begin{align*}
&\P\left(\sqrt{n_v}\left\rvert\widehat{\mu}_{cf}(\widehat{m},\pi_{|\cA})-\widehat{\mu}_{\mathrm{ma}}(\widetilde{m},\pi_{|\cA})\right\rvert>\epsilon\right)\\
&\quad\leq \P(\cE_v^c)+\frac{n_v}{\epsilon^2}\E_p\left[\left\{\widehat{\mu}_{cf}(\widehat{m},\pi_{|\cA})-\widehat{\mu}_{\mathrm{ma}}(\widetilde{m},\pi_{|\cA})\right\}^2\mathds{1}_{\cE_v}\right].
\end{align*}
Assumption~\ref{CF2a} gives $\P(\cE_v^c)\xrightarrow[v \to \infty]{} 0$. It remains to control the expectation.

For $k\in U_v$, set $\rho_k:=\widehat{m}^{(-k)}(\bx_k)-\widetilde{m}(\bx_k)$. On $\cE_v$,
$$
\widehat{\mu}_{cf}(\widehat{m},\pi_{|\cA})-\widehat{\mu}_{\mathrm{ma}}(\widetilde{m},\pi_{|\cA})
=-\frac{1}{N_v}\sum_{k\in U_v}(I_k-\pi_{k|\cA})\frac{\rho_k}{\pi_{k|\cA}}.
$$
Apply Lemma~\ref{lemma:secondMoment} with $d_k=\rho_k/\pi_{k|\cA}$ on $\cE_v$ and $d_k=0$ otherwise. The required measurability follows because $\rho_k$ uses only observations outside its fold, while $\cE_v$ is determined by $\cA$. Since $d_k^2\leq\rho_k^2/\lambda_{cf}^2$, the lemma gives
$$
n_v\E_p\left[\left\{\widehat{\mu}_{cf}(\widehat{m},\pi_{|\cA})-\widehat{\mu}_{\mathrm{ma}}(\widetilde{m},\pi_{|\cA})\right\}^2\mathds{1}_{\cE_v}\right]
\leq\frac{C}{N_v}\sum_{k\in U_v}\E_p(\rho_k^2),
$$
where $C$ is independent of $v$. The right-hand side converges to zero by \ref{M1}. This proves the result.

\subsection{Proof of Theorem \ref{theo:consCond}} \label{Prooftheo:consCond}

Let $\cE_v:=\{\min_{k\in U_v}\pi_{k|\cA}\geq\lambda_{cf}\}$. On $\cE_v$,
$$\widehat{\mu}_{\mathrm{ma}}(\widetilde{m},\pi_{|\cA})-\mu
=\frac{1}{N_v}\sum_{k\in U_v}\widetilde e_k
\left(\frac{I_k}{\pi_{k|\cA}}-1\right).$$
Its conditional mean is
$$\E_p\left[\widehat{\mu}_{\mathrm{ma}}(\widetilde{m},\pi_{|\cA})-\mu\mid\cA\right]
=\frac{1}{N_v}\sum_{k\in U_v}\widetilde e_k
\left(\frac{\pi_{k|\cA}}{\pi_{k|\cA}}-1\right)=0.$$
Therefore,
\begin{align*}
n_v \E_p \left[ \left(  \widehat{\mu}_{\mathrm{ma}}(\widetilde{m},\pi_{|\cA}) - \mu\right)^2\big \rvert \cA \right] &= n_v \V_p  \left(  \widehat{\mu}_{\mathrm{ma}}(\widetilde{m},\pi_{|\cA}) \big \rvert \cA \right) \\
&= n_v \left\{\dfrac{1}{N_v^2}\sum_{k\in U_v}\Delta_{kk|\cA}   \dfrac{\widetilde{e}_k^{2}}{\pi_{k\rvert \cA}^2} + \dfrac{1}{N_v^2}\sum_{k\in U_v} \sum_{\substack{l\in U_v \\ l \neq k}}\Delta_{kl|\cA} \dfrac{\widetilde{e}_k}{\pi_{k\rvert \cA}}\dfrac{\widetilde{e}_l}{\pi_{l\rvert \cA}} \right\}\\
&\leq 
\dfrac{1}{N_v \lambda_{cf}^2}\sum_{k\in U_v}  \widetilde{e}_k^{2} \left\{ \dfrac{n_v}{N_v} + C_{\Delta 2}\right\}.
\end{align*}
Here we used \ref{CF4} and
$$
\sum_{k\in U_v}\sum_{\substack{l\in U_v\\l\neq k}}|\widetilde{e}_k\widetilde{e}_l|
\leq\left(\sum_{k\in U_v}|\widetilde{e}_k|\right)^2
\leq N_v\sum_{k\in U_v}\widetilde{e}_k^2.
$$
Multiplying by $\mathds{1}_{\cE_v}$, taking expectations, and using \ref{M2a}, we obtain, for all sufficiently large $v$,
$$
n_v\E_p\left[\left\{\widehat{\mu}_{\mathrm{ma}}(\widetilde{m},\pi_{|\cA})-\mu\right\}^2\mathds{1}_{\cE_v}\right]\leq C,
$$
for some constant $C$.
Now, Markov's inequality gives, for every $K>0$,
\begin{align*}
&\P\left(\left\rvert\sqrt{n_v}\left\{\widehat{\mu}_{\mathrm{ma}}(\widetilde{m},\pi_{|\cA})-\mu\right\}\right\rvert>K\right)\\
&\quad\leq \P(\cE_v^c)+\frac{n_v}{K^2}\E_p\left[\left\{\widehat{\mu}_{\mathrm{ma}}(\widetilde{m},\pi_{|\cA})-\mu\right\}^2\mathds{1}_{\cE_v}\right]
\leq \P(\cE_v^c)+\frac{C}{K^2}.
\end{align*}
Assumption~\ref{CF2a} gives $\P(\cE_v^c)\xrightarrow[v \to \infty]{}0$. The result follows by choosing $K$ large enough.

\subsection{Proof of Theorem \ref{theo:consistencyCondVar}} \label{Prooftheo:consistencyCondVar}

Write
\begin{align*}
&\left\rvert \widehat{V}_{cf}(\widehat{m},\pi_{|\cA})-\V_p\left\{\widehat{\mu}_{\mathrm{ma}}(\widetilde{m},\pi_{|\cA})\right\}\right\rvert\\
&\quad\leq \left\rvert \widehat{V}_{cf}(\widehat{m},\pi_{|\cA})-\V_p\left\{\widehat{\mu}_{\mathrm{ma}}(\widetilde{m},\pi_{|\cA})\mid\cA\right\}\right\rvert + \left\rvert \V_p\left\{\widehat{\mu}_{\mathrm{ma}}(\widetilde{m},\pi_{|\cA})\mid\cA\right\}-\V_p\left\{\widehat{\mu}_{\mathrm{ma}}(\widetilde{m},\pi_{|\cA})\right\}\right\rvert\\
&\quad=:T_{1v}+T_{2v}.
\end{align*}
We first consider $T_{1v}$. Let
$$
\cE_v:=\left\{\min_{k\in U_v}\pi_{k|\cA}\geq\lambda_{cf}, 
\min_{\substack{k,l\in U_v\\k\neq l}}\pi_{kl|\cA}\geq\lambda_{cf}^\star\right\}.
$$
By \ref{CF2a}, $\P(\cE_v^c)\xrightarrow[v \to \infty]{}0$. On $\cE_v$, set $c_{kl|\cA}:=\Delta_{kl|\cA}/(\pi_{k|\cA}\pi_{l|\cA})$. We have
\begin{align*}
T_{1v}&\leq \left\rvert\dfrac{1}{N_v^2}\sum_{k\in U_v}\sum_{l\in U_v}c_{kl|\cA}\widetilde{e}_k\widetilde{e}_l\left(\dfrac{I_kI_l}{\pi_{kl|\cA}}-1\right)\right\rvert\\
&\quad+\left\rvert\dfrac{1}{N_v^2}\sum_{k\in S_v}\sum_{l\in S_v}\dfrac{c_{kl|\cA}}{\pi_{kl|\cA}}\left(\widehat{e}_{(-k)}\widehat{e}_{(-l)}-\widetilde{e}_k\widetilde{e}_l\right)\right\rvert\\
&=:A_v+B_v.
\end{align*}

For $A_v$, we condition on $\cA$ and use the argument in the proof of Theorem~3 of \citet{breidt2000local}, with the conditional inclusion probabilities in place of the unconditional ones. On $\cE_v$, the required lower bounds follow from \ref{CF2a}, while \ref{CF4}, \ref{CF5}, and \ref{M2b} give the remaining bounds. Hence, $n_vA_v=o_\P(1)$ on $\cE_v$.

We next control $B_v$. Set $d_k:=\widehat{e}_{(-k)}-\widetilde{e}_k=\widetilde m(\bx_k)-\widehat m^{(-k)}(\bx_k)$. For $k,l\in S_v$,
$$
\widehat{e}_{(-k)}\widehat{e}_{(-l)}-\widetilde{e}_k\widetilde{e}_l
=d_k\widetilde{e}_l+d_l\widetilde{e}_k+d_kd_l.
$$
On $\cE_v$, the diagonal and off-diagonal coefficients satisfy
$$
\frac{|c_{kk|\cA}|}{\pi_{kk|\cA}}
=\frac{1-\pi_{k|\cA}}{\pi_{k|\cA}^2}
\leq\frac{1}{\lambda_{cf}^2},
\qquad
\max_{k\neq l\in U_v}\frac{|c_{kl|\cA}|}{\pi_{kl|\cA}}
\leq\frac{C_{\Delta 2}}{n_v\lambda_{cf}^2\lambda_{cf}^\star},
$$
by \ref{CF4}. Taking absolute values and bounding the sample sums by population sums gives
\begin{align*}
n_vB_v
&\leq\frac{n_v}{N_v^2\lambda_{cf}^2}
\sum_{k\in U_v}\left(2|d_k\widetilde e_k|+d_k^2\right)\\
&\quad+\frac{C_{\Delta 2}}{N_v^2\lambda_{cf}^2\lambda_{cf}^\star}
\left[
2\left(\sum_{k\in U_v}|d_k|\right)
\left(\sum_{k\in U_v}|\widetilde e_k|\right)
+\left(\sum_{k\in U_v}|d_k|\right)^2
\right].
\end{align*}
Cauchy--Schwarz gives
\begin{align*}
\frac{1}{N_v}\sum_{k\in U_v}|d_k\widetilde e_k|
&\leq\left\{\left(\frac{1}{N_v}\sum_{k\in U_v}d_k^2\right)
\left(\frac{1}{N_v}\sum_{k\in U_v}\widetilde e_k^2\right)\right\}^{1/2},\\
\frac{1}{N_v^2}\left(\sum_{k\in U_v}|d_k|\right)^2
&\leq\frac{1}{N_v}\sum_{k\in U_v}d_k^2,\\
\frac{1}{N_v^2}\left(\sum_{k\in U_v}|d_k|\right)
\left(\sum_{k\in U_v}|\widetilde e_k|\right)
&\leq\left\{\left(\frac{1}{N_v}\sum_{k\in U_v}d_k^2\right)
\left(\frac{1}{N_v}\sum_{k\in U_v}\widetilde e_k^2\right)\right\}^{1/2}.
\end{align*}
Together with $n_v\leq N_v$, these bounds give
$$
n_vB_v\leq C_B\left[2\left\{\left(\frac{1}{N_v}\sum_{k\in U_v}d_k^2\right)\left(\frac{1}{N_v}\sum_{k\in U_v}\widetilde{e}_k^2\right)\right\}^{1/2}+\frac{1}{N_v}\sum_{k\in U_v}d_k^2\right],
$$
where $C_B$ is a constant independent of $v$.
By \ref{M1} and Markov's inequality, $N_v^{-1}\sum_k d_k^2=o_\P(1)$. Moreover, by the Cauchy--Schwarz inequality and \ref{M2b},
$$
\left(\frac{1}{N_v}\sum_{k\in U_v}\widetilde{e}_k^2\right)^2
\leq\frac{1}{N_v}\sum_{k\in U_v}\widetilde{e}_k^4=\mathcal{O}(1).
$$
Hence, $n_vB_v=o_\P(1)$ on $\cE_v$. Since $\P(\cE_v^c)\xrightarrow[v \to \infty]{}0$, we obtain $n_vT_{1v}=o_\P(1)$.

For $T_{2v}$, set
$$
Z_v:=n_v\V_p\left\{\widehat{\mu}_{\mathrm{ma}}(\widetilde{m},\pi_{|\cA})\mid\cA\right\}.
$$
The law of total variance gives
$$
n_vT_{2v}
\leq\left\rvert Z_v-\E_p(Z_v)\right\rvert
+n_v\V_p\left\{\E_p\left[\widehat{\mu}_{\mathrm{ma}}(\widetilde{m},\pi_{|\cA})\mid\cA\right]\right\}.
$$
We first show that $Z_v-\E_p(Z_v)=o_\P(1)$. Since $Z_v\xrightarrow[v \to \infty]{\P}c$ by assumption, it is enough to verify uniform integrability. In the sums below, terms for which a first-order conditional inclusion probability is zero are omitted. By \ref{CF4},
\begin{align*}
Z_v
&\leq \frac{n_v}{N_v^2}\sum_{k\in U_v}\frac{\widetilde{e}_k^2}{\pi_+}
+\frac{C_{\Delta 2}}{N_v^2\pi_+^2}\sum_{k\in U_v}\sum_{\substack{l\in U_v\\l\neq k}}|\widetilde{e}_k\widetilde{e}_l|\\
&\leq \left\{\frac{n_v}{N_v\pi_+}+\frac{C_{\Delta 2}}{\pi_+^2}\right\}\frac{1}{N_v}\sum_{k\in U_v}\widetilde{e}_k^2.
\end{align*}
By \ref{M2b},
$$
C_e:=\sup_{v\in\N}\frac{1}{N_v}\sum_{k\in U_v}\widetilde{e}_k^2<\infty.
$$
Since $n_v/N_v\leq1$ and $\pi_+\leq1$,
$$
Z_v\leq\frac{C_e(1+C_{\Delta 2})}{\pi_+^2}.
$$
It follows from \ref{CF6} that
$$
\sup_{v\in\N}\E_p\left(Z_v^{1+\delta/2}\right)
\leq\{C_e(1+C_{\Delta 2})\}^{1+\delta/2}
\sup_{v\in\N}\E_p\left(\frac{1}{\pi_+^{2+\delta}}\right)<\infty.
$$
Thus, $(Z_v)_{v\in\N}$ is uniformly integrable, so $\E_p(Z_v)\xrightarrow[v \to \infty]{} c$, and therefore
$$
Z_v-\E_p(Z_v)=o_\P(1).
$$
It remains to bound the variance of the conditional mean. Let $\cF_v:=\{\min_{k\in U_v}\pi_{k|\cA}>0\}$. On $\cF_v$, the oracle estimator is conditionally unbiased. In general, however,
\begin{align*}
&\E_p\left[\widehat{\mu}_{\mathrm{ma}}(\widetilde{m},\pi_{|\cA})\mid\cA\right]-\mu\\
&\quad=\frac{1}{N_v}\left\{
\sum_{k\in U_v}\widetilde m(\bx_k)
+\sum_{\substack{k\in U_v\\\pi_{k|\cA}>0}}\widetilde e_k
-\sum_{k\in U_v}y_k\right\}\\
&\quad=-\frac{1}{N_v}\sum_{\substack{k\in U_v\\\pi_{k|\cA}=0}}\widetilde{e}_k.
\end{align*}
The Cauchy--Schwarz inequality gives
\begin{align*}
&\left\rvert\E_p\left[\widehat{\mu}_{\mathrm{ma}}(\widetilde{m},\pi_{|\cA})\mid\cA\right]-\mu\right\rvert\\
&\quad\leq\left\{\left(\frac{1}{N_v}\sum_{k\in U_v}\mathds{1}_{\{\pi_{k|\cA}=0\}}\right)
\left(\frac{1}{N_v}\sum_{k\in U_v}\widetilde e_k^2\right)\right\}^{1/2}\\
&\quad\leq C_e^{1/2}\mathds{1}_{\cF_v^c}.
\end{align*}
The variance is at most the mean squared deviation from $\mu$. Hence, by \ref{CF2b},
\begin{align*}
&\V_p\left\{\E_p\left[\widehat{\mu}_{\mathrm{ma}}(\widetilde{m},\pi_{|\cA})\mid\cA\right]\right\}\\
&\quad\leq\E_p\left[\left\{\E_p\left[
\widehat{\mu}_{\mathrm{ma}}(\widetilde{m},\pi_{|\cA})\mid\cA\right]-\mu\right\}^2\right]\\
&\quad\leq C_e\P(\cF_v^c)=o(n_v^{-1}).
\end{align*}
These two bounds give $n_vT_{2v}=o_\P(1)$.
Combining the bounds for $T_{1v}$ and $T_{2v}$ proves the result.

\subsection{Proof of Corollary \ref{coro1}} \label{proofcoro1}

The proof of Theorem~\ref{theo:consistencyCondVar} gives
$$
n_v\V_p\left\{\widehat{\mu}_{\mathrm{ma}}(\widetilde{m},\pi_{|\cA})\right\}\xrightarrow[v \to \infty]{} c>0.
$$
Together with Theorem~\ref{theo:consistencyCondVar}, this implies
$$
\frac{\widehat{V}_{cf}(\widehat{m},\pi_{|\cA})}{\V_p\left\{\widehat{\mu}_{\mathrm{ma}}(\widetilde{m},\pi_{|\cA})\right\}}
\xrightarrow[v \to \infty]{\P}1.
$$
In particular, $\P\{\widehat{V}_{cf}(\widehat{m},\pi_{|\cA})>0\}\xrightarrow[v \to \infty]{}1$.
Moreover, the oracle standard deviation is of order $n_v^{-1/2}$, so Theorem~\ref{th:equivalent1} gives
$$
\frac{\widehat{\mu}_{cf}(\widehat{m},\pi_{|\cA})-\widehat{\mu}_{\mathrm{ma}}(\widetilde{m},\pi_{|\cA})}{\sqrt{\V_p\left\{\widehat{\mu}_{\mathrm{ma}}(\widetilde{m},\pi_{|\cA})\right\}}}
=o_\P(1).
$$
On the event that the estimated variance is positive,
\begin{align*}
&\frac{\widehat{\mu}_{cf}(\widehat m,\pi_{|\cA})-\mu}
{\sqrt{\widehat V_{cf}(\widehat m,\pi_{|\cA})}}\\
&\quad=\left[
\frac{\widehat{\mu}_{\mathrm{ma}}(\widetilde m,\pi_{|\cA})-\mu}
{\sqrt{\V_p\{\widehat{\mu}_{\mathrm{ma}}(\widetilde m,\pi_{|\cA})\}}}
+o_\P(1)\right]
\left[
\frac{\V_p\{\widehat{\mu}_{\mathrm{ma}}(\widetilde m,\pi_{|\cA})\}}
{\widehat V_{cf}(\widehat m,\pi_{|\cA})}
\right]^{1/2}.
\end{align*}
The result follows from the assumed central limit theorem and Slutsky's theorem.

\section{Proofs of Section \ref{Sec:uncond}}

\subsection{Proof of Theorem \ref{th:equivalent2}} \label{Proofth:equivalent2}

We first split the difference between the two estimators into two terms:
\begin{align*}
&\sqrt{n_v}\left\{\widehat{\mu}_{cf}(\widehat{m},\pi)-\widehat{\mu}_{\mathrm{ma}}(\widetilde{m},\pi)\right\}\\
&\quad=\frac{\sqrt{n_v}}{N_v}\sum_{j=1}^M\sum_{k\in U_{j,v}}\frac{I_k-\pi_{k|\cA}}{\pi_k}\left\{\widetilde{m}(\bx_k)-\widehat{m}^{(-k)}(\bx_k)\right\}\\
&\qquad+\frac{\sqrt{n_v}}{N_v}\sum_{j=1}^M\sum_{k\in U_{j,v}}\left(\frac{\pi_{k|\cA}}{\pi_k}-1\right)\left\{\widetilde{m}(\bx_k)-\widehat{m}^{(-k)}(\bx_k)\right\}\\
&\quad:=\sum_{j=1}^M\left\{\sqrt{n_v}T_{j,v}^{(1)}+\sqrt{n_v}T_{j,v}^{(2)}\right\}.
\end{align*}
We first control the centered term $\sqrt{n_v}\sum_{j=1}^M T_{j,v}^{(1)}$. For $k\in U_v$, set $d_k:=\widetilde{m}(\bx_k)-\widehat{m}^{(-k)}(\bx_k)$. Apply Lemma~\ref{lemma:secondMoment} to the coefficients $d_k/\pi_k$, which depend only on the observations outside the fold containing $k$ and on $\cA$. By \ref{D2} and \ref{M1},
$$
n_v\E_p\left[\left\{\sum_{j=1}^M T_{j,v}^{(1)}\right\}^2\right] \leq\frac{C}{N_v}\sum_{k\in
U_v}\E_p(d_k^2)=o(1),
$$
where $C$ is independent of $v$. Thus, Markov's inequality gives $\sqrt{n_v}\sum_{j=1}^M T_{j,v}^{(1)}=o_\P(1)$.

For the second term, the Cauchy--Schwarz inequality gives
\begin{align*}
n_v\left|T_{j,v}^{(2)}\right|^2
&\leq\frac{n_v}{N_v^2\lambda^2}\left(\sum_{k\in U_{j,v}}|\pi_{k|\cA}-\pi_k||d_k|\right)^2\\
&\leq\frac{n_vN_{j,v}}{N_v^2\lambda^2}
\max_{k\in U_v}(\pi_{k|\cA}-\pi_k)^2\sum_{k\in U_{j,v}}d_k^2\\
&\leq\frac{n_v}{\lambda^2}\max_{k\in U_v}(\pi_{k|\cA}-\pi_k)^2
\frac{1}{N_v}\sum_{k\in U_v}d_k^2.
\end{align*}
For every $\epsilon,K>0$, the last bound and Markov's inequality give
\begin{align*}
\P\left(\sqrt{n_v}|T_{j,v}^{(2)}|>\epsilon\right)
&\leq\P\left(n_v\max_{k\in U_v}(\pi_{k|\cA}-\pi_k)^2>K\right)\\
&\quad+\P\left(\frac{1}{N_v}\sum_{k\in U_v}d_k^2>\frac{\epsilon^2\lambda^2}{K}\right)\\
&\leq\frac{C_1}{K}
+\frac{K}{\epsilon^2\lambda^2N_v}\sum_{k\in U_v}\E_p(d_k^2),
\end{align*}
where the last line uses \ref{CF3}. By \ref{M1}, letting $v\to\infty$ and then $K\to\infty$ proves that $\sqrt{n_v}T_{j,v}^{(2)}=o_\P(1)$. Since $M$ is fixed, the result follows.

\subsection{Proof of Theorem \ref{theo:consUnCond}} \label{Prooftheo:consUnCond}

Since $\widetilde{m}$ is deterministic,
$$
\widehat{\mu}_{\mathrm{ma}}(\widetilde{m},\pi)-\mu =\frac{1}{N_v}\sum_{k\in
U_v}\widetilde{e}_k\left(\frac{I_k}{\pi_k}-1\right),
$$
and this quantity has mean zero. Therefore, by \ref{D2} and \ref{D3},
\begin{align*}
&\E_p\left[\left\{\widehat{\mu}_{\mathrm{ma}}(\widetilde{m},\pi)-\mu\right\}^2\right]\\
&\quad=\frac{1}{N_v^2}\sum_{k\in U_v}\frac{1-\pi_k}{\pi_k}\widetilde{e}_k^2
+\frac{1}{N_v^2}\sum_{k\in U_v}\sum_{\substack{l\in U_v\\l\neq k}}\frac{\Delta_{kl}}{\pi_k\pi_l}\widetilde{e}_k\widetilde{e}_l\\
&\quad\leq\left\{\frac{1-\lambda}{N_v\lambda}+\frac{C}{n_v\lambda^2}\right\}
\frac{1}{N_v}\sum_{k\in U_v}\widetilde{e}_k^2,
\end{align*}
where we used
$$
\sum_{k\in U_v}\sum_{\substack{l\in U_v\\l\neq k}}|\widetilde{e}_k\widetilde{e}_l| \leq\left(\sum_{k\in
U_v}|\widetilde{e}_k|\right)^2 \leq N_v\sum_{k\in U_v}\widetilde{e}_k^2.
$$
By \ref{D1}, $n_v$ and $N_v$ have the same order, while \ref{M2a} bounds the average squared residuals. Thus, the second moment is $\mathcal{O}(N_v^{-1})$, and the result follows from Chebyshev's inequality.

\subsection{Proof of Theorem \ref{theoVar}} \label{ProoftheoVar}

The variance estimator based on the fixed residuals is
$$
\widetilde{V}(\widetilde{m},\pi) =\frac{1}{N_v^2}\sum_{k\in S_v}\sum_{l\in S_v} \frac{\Delta_{kl}}{\pi_{kl}}
\frac{\widetilde{e}_k}{\pi_k} \frac{\widetilde{e}_l}{\pi_l}.
$$
We split the variance estimation error as follows:
\begin{align*}
&n_v\left\rvert\widehat{V}_{cf}(\widehat{m},\pi)-\V_p \left\{\widehat{\mu}_{\mathrm{ma}}(\widetilde{m},\pi)\right\}\right\rvert\\
&\quad\leq n_v\left\rvert\widehat{V}_{cf}(\widehat{m},\pi)-\widetilde{V}(\widetilde{m},\pi)\right\rvert
+n_v\left\rvert\widetilde{V}(\widetilde{m},\pi)-\V_p \left\{\widehat{\mu}_{\mathrm{ma}}(\widetilde{m},\pi)\right\}\right\rvert\\
&\quad=:A_v+B_v.
\end{align*}

We first control $A_v$. Set $d_k:=\widehat{e}_{(-k)}-\widetilde{e}_k=\widetilde m(\bx_k)-\widehat m^{(-k)}(\bx_k)$. By \ref{D2} and \ref{D3},
$$
\frac{|\Delta_{kk}|}{\pi_k^2\pi_{kk}} =\frac{1-\pi_k}{\pi_k^2} \leq\frac{1}{\lambda^2}, \qquad \max_{k\neq
l\in U_v}\frac{|\Delta_{kl}|}{\pi_k\pi_l\pi_{kl}} \leq\frac{C}{n_v\lambda^2\lambda^\star}.
$$
We have
$$
\widehat{e}_{(-k)}\widehat{e}_{(-l)}-\widetilde{e}_k\widetilde{e}_l
=d_kd_l+d_k\widetilde{e}_l+d_l\widetilde{e}_k.
$$
Taking absolute values and bounding the sample sums by population sums gives
\begin{align*}
A_v
&\leq\frac{n_v}{N_v^2\lambda^2}
\sum_{k\in U_v}\left(d_k^2+2|d_k\widetilde e_k|\right)\\
&\quad+\frac{C}{N_v^2\lambda^2\lambda^\star}
\left[
\left(\sum_{k\in U_v}|d_k|\right)^2
+2\left(\sum_{k\in U_v}|d_k|\right)
\left(\sum_{k\in U_v}|\widetilde e_k|\right)
\right].
\end{align*}
Cauchy--Schwarz gives
\begin{align*}
\frac{1}{N_v}\sum_{k\in U_v}|d_k\widetilde e_k|
&\leq\left\{\left(\frac{1}{N_v}\sum_{k\in U_v}d_k^2\right)
\left(\frac{1}{N_v}\sum_{k\in U_v}\widetilde e_k^2\right)\right\}^{1/2},\\
\frac{1}{N_v^2}\left(\sum_{k\in U_v}|d_k|\right)^2
&\leq\frac{1}{N_v}\sum_{k\in U_v}d_k^2,\\
\frac{1}{N_v^2}\left(\sum_{k\in U_v}|d_k|\right)
\left(\sum_{k\in U_v}|\widetilde e_k|\right)
&\leq\left\{\left(\frac{1}{N_v}\sum_{k\in U_v}d_k^2\right)
\left(\frac{1}{N_v}\sum_{k\in U_v}\widetilde e_k^2\right)\right\}^{1/2}.
\end{align*}
Together with $n_v\leq N_v$, these bounds give
$$
A_v\leq C_A\left[ \frac{1}{N_v}\sum_{k\in U_v}d_k^2 +2\left\{\left(\frac{1}{N_v}\sum_{k\in U_v}d_k^2\right)
\left(\frac{1}{N_v}\sum_{k\in U_v}\widetilde{e}_k^2\right)\right\}^{1/2} \right],
$$
where $C_A$ is a constant independent of $v$.
By \ref{M1} and Markov's inequality, $N_v^{-1}\sum_{k\in U_v}d_k^2=o_\P(1)$. Moreover, \ref{M2b} and the Cauchy--Schwarz inequality imply that $N_v^{-1}\sum_{k\in U_v}\widetilde{e}_k^2$ is bounded. Hence, $A_v=o_\P(1)$.

It remains to control $B_v$. By \ref{D2},
$$
\begin{aligned}
\E_p\left\{\widetilde{V}(\widetilde{m},\pi)\right\}
&=\frac{1}{N_v^2}\sum_{k,l\in U_v}\frac{\E_p(I_kI_l)}{\pi_{kl}}
\frac{\Delta_{kl}}{\pi_k\pi_l}\widetilde e_k\widetilde e_l\\
&=\frac{1}{N_v^2}\sum_{k,l\in U_v}
\frac{\Delta_{kl}}{\pi_k\pi_l}\widetilde e_k\widetilde e_l\\
&=\V_p \left\{\widehat{\mu}_{\mathrm{ma}}(\widetilde{m},\pi)\right\}.
\end{aligned}
$$
The argument used in the proof of Theorem~3 of \citet{breidt2000local} applies to the fixed residuals $(\widetilde{e}_k)_{k\in U_v}$. Under \ref{D1}--\ref{D4} and \ref{M2b}, it gives $B_v=o_\P(1)$.

\subsection{Proof of Corollary \ref{coro2}} \label{proofcoro2}

We first compare the estimated and oracle variances. By \eqref{eq:varOrder}, the oracle variance is bounded below by a positive constant times $n_v^{-1}$. Therefore, Theorem~\ref{theoVar} gives
$$
\frac{\widehat{V}_{cf}(\widehat{m},\pi)}{\V_p \left\{\widehat{\mu}_{\mathrm{ma}}(\widetilde{m},\pi)\right\}}
\xrightarrow[v \to \infty]{\P}1.
$$
In particular, $\P\{\widehat{V}_{cf}(\widehat{m},\pi)>0\}\xrightarrow[v \to \infty]{}1$. For the difference between the two numerators, Theorem~\ref{th:equivalent2} and \eqref{eq:varOrder} give
$$
\frac{\widehat{\mu}_{cf}(\widehat{m},\pi)-\widehat{\mu}_{\mathrm{ma}}(\widetilde{m},\pi)}{\sqrt{\V_p \left\{\widehat{\mu}_{\mathrm{ma}}(\widetilde{m},\pi)\right\}}}
=o_\P(1).
$$
On the event that the estimated variance is positive,
\begin{align*}
&\frac{\widehat{\mu}_{cf}(\widehat m,\pi)-\mu}
{\sqrt{\widehat V_{cf}(\widehat m,\pi)}}\\
&\quad=\left[
\frac{\widehat{\mu}_{\mathrm{ma}}(\widetilde m,\pi)-\mu}
{\sqrt{\V_p\{\widehat{\mu}_{\mathrm{ma}}(\widetilde m,\pi)\}}}
+o_\P(1)\right]
\left[
\frac{\V_p\{\widehat{\mu}_{\mathrm{ma}}(\widetilde m,\pi)\}}
{\widehat V_{cf}(\widehat m,\pi)}
\right]^{1/2}.
\end{align*}
The result follows from the assumed central limit theorem and Slutsky's theorem.

\section{Proofs of Section \ref{sec:6}}

\subsection{Proof of Theorem \ref{theo:eqq}} \label{Prooftheo:eqq}

We first compare the oracle estimators; the first-order equivalences will then give the result. By definition,
$$
B_\cA=\dfrac{1}{N_v}\sum_{k\in U_v}\left(\dfrac{\pi_{k|\cA}}{\pi_k}-1\right)\widetilde{e}_k.
$$
Let $\cE_v:=\{\min_{k\in U_v}\pi_{k|\cA}\geq\lambda_{cf}\}$. On $\cE_v$, centering the sampling indicators at their conditional means gives
\begin{align*}
&\widehat{\mu}_{\mathrm{ma}}(\widetilde{m},\pi)-\widehat{\mu}_{\mathrm{ma}}(\widetilde{m},\pi_{|\cA})-B_\cA\\
&\quad=\dfrac{1}{N_v}\sum_{k\in U_v}\widetilde{e}_k(I_k-\pi_{k|\cA})
\left(\dfrac{1}{\pi_k}-\dfrac{1}{\pi_{k|\cA}}\right).
\end{align*}
Its conditional mean is zero, since $\E_p(I_k-\pi_{k|\cA}\mid\cA)=0$. To bound its conditional second moment, first note that, by \ref{D2},
$$
\left|\dfrac{1}{\pi_k}-\dfrac{1}{\pi_{k|\cA}}\right|
\leq\dfrac{|\pi_{k|\cA}-\pi_k|}{\lambda\lambda_{cf}}.
$$
Separating the diagonal and off-diagonal terms, and using \ref{CF4} and the Cauchy--Schwarz inequality, we obtain on $\cE_v$
\begin{align*}
&n_v\E_p\left[\left\{\widehat{\mu}_{\mathrm{ma}}(\widetilde{m},\pi)-\widehat{\mu}_{\mathrm{ma}}(\widetilde{m},\pi_{|\cA})-B_\cA\right\}^2\Bigm|\cA\right]\\
&\quad=\frac{n_v}{N_v^2}\sum_{k,l\in U_v}\Delta_{kl|\cA}\widetilde e_k\widetilde e_l
\left(\frac{1}{\pi_k}-\frac{1}{\pi_{k|\cA}}\right)
\left(\frac{1}{\pi_l}-\frac{1}{\pi_{l|\cA}}\right)\\
&\quad\leq
\frac{\max_{k\in U_v}(\pi_{k|\cA}-\pi_k)^2}{\lambda^2\lambda_{cf}^2}
\left[
\frac{n_v}{4N_v^2}\sum_{k\in U_v}\widetilde e_k^2
+\frac{C_{\Delta 2}}{N_v^2}\sum_{k\in U_v}\sum_{\substack{l\in U_v\\l\neq k}}|\widetilde e_k\widetilde e_l|
\right]\\
&\quad\leq C\left(\dfrac{1}{N_v}\sum_{k\in U_v}\widetilde{e}_k^2\right)
\max_{k\in U_v}(\pi_{k|\cA}-\pi_k)^2,
\end{align*}
where $C$ is a constant independent of $v$. Here we used $n_v\leq N_v$ and
$$
\sum_{k\in U_v}\sum_{\substack{l\in U_v\\l\neq k}}|\widetilde{e}_k\widetilde{e}_l|
\leq\left(\sum_{k\in U_v}|\widetilde{e}_k|\right)^2
\leq N_v\sum_{k\in U_v}\widetilde{e}_k^2.
$$
Multiplying by $\mathds{1}_{\cE_v}$ and taking expectations, \ref{CF3} and \ref{M2a} give
\begin{align*}
&n_v\E_p\left[\left\{\widehat{\mu}_{\mathrm{ma}}(\widetilde{m},\pi)-\widehat{\mu}_{\mathrm{ma}}(\widetilde{m},\pi_{|\cA})-B_\cA\right\}^2\mathds{1}_{\cE_v}\right]\\
&\quad\leq C\left(\frac{1}{N_v}\sum_{k\in U_v}\widetilde e_k^2\right)
\E_p\left[\max_{k\in U_v}(\pi_{k|\cA}-\pi_k)^2\right]
=\mathcal{O}(n_v^{-1})=o(1).
\end{align*}
Since $\P(\cE_v^c)\xrightarrow[v \to \infty]{}0$ by \ref{CF2a}, Markov's inequality gives
$$
\sqrt{n_v}\left\{\widehat{\mu}_{\mathrm{ma}}(\widetilde{m},\pi)-\widehat{\mu}_{\mathrm{ma}}(\widetilde{m},\pi_{|\cA})-B_\cA\right\}=o_\P(1).
$$
The result follows from Theorems~\ref{th:equivalent1} and \ref{th:equivalent2}.

\subsection{Proof of Theorem \ref{Thm:reminder}} \label{ProofThm:reminder}

We first compare the two oracle variances by conditioning on $\cA$. By \ref{D2} and \ref{CF2c},
$$
\E_p\left[\widehat{\mu}_{\mathrm{ma}}(\widetilde{m},\pi)\mid\cA\right]=\mu+B_\cA,
\qquad
\E_p\left[\widehat{\mu}_{\mathrm{ma}}(\widetilde{m},\pi_{|\cA})\mid\cA\right]=\mu.
$$
The tower property gives
$$\E_p(B_\cA)=\frac{1}{N_v}\sum_{k\in U_v}\widetilde e_k
\left\{\frac{\E_p[\E_p(I_k\mid\cA)]}{\pi_k}-1\right\}
=\frac{1}{N_v}\sum_{k\in U_v}\widetilde e_k
\left(\frac{\pi_k}{\pi_k}-1\right)=0.$$
The law of total variance therefore gives
\begin{equation}\label{eq:condVar1}
\V_p\left\{\widehat{\mu}_{\mathrm{ma}}(\widetilde{m},\pi)\right\}
=\E_p\left[\dfrac{1}{N_v^2}\sum_{k,l\in U_v}\Delta_{kl|\cA}
\dfrac{\widetilde{e}_k\widetilde{e}_l}{\pi_k\pi_l}\right]+\E_p(B_\cA^2)
\end{equation}
and
\begin{equation}\label{eq:condVar2}
\V_p\left\{\widehat{\mu}_{\mathrm{ma}}(\widetilde{m},\pi_{|\cA})\right\}
=\E_p\left[\dfrac{1}{N_v^2}\sum_{k,l\in U_v}\Delta_{kl|\cA}
\dfrac{\widetilde{e}_k\widetilde{e}_l}{\pi_{k|\cA}\pi_{l|\cA}}\right].
\end{equation}
Subtracting yields
\begin{equation}\label{dec1}
\V_p\left\{\widehat{\mu}_{\mathrm{ma}}(\widetilde{m},\pi)\right\}
-\V_p\left\{\widehat{\mu}_{\mathrm{ma}}(\widetilde{m},\pi_{|\cA})\right\}
=\E_p(B_\cA^2)+R_v,
\end{equation}
where
$$
R_v:=\E_p\left[\dfrac{1}{N_v^2}\sum_{k,l\in U_v}\Delta_{kl|\cA}\widetilde{e}_k\widetilde{e}_l
\left(\dfrac{1}{\pi_k\pi_l}-\dfrac{1}{\pi_{k|\cA}\pi_{l|\cA}}\right)\right].
$$

We next bound $R_v$. Write
$$
\pi_{k|\cA}\pi_{l|\cA}-\pi_k\pi_l
=(\pi_{k|\cA}-\pi_k)\pi_{l|\cA}
+\pi_k(\pi_{l|\cA}-\pi_l).
$$
Since all inclusion probabilities are at most one, \ref{D2} and \ref{CF2c} give
$$
\left|\dfrac{1}{\pi_k\pi_l}-\dfrac{1}{\pi_{k|\cA}\pi_{l|\cA}}\right|
\leq\dfrac{|\pi_{k|\cA}-\pi_k|+|\pi_{l|\cA}-\pi_l|}{\lambda^2\lambda_c^2}.
$$
Using $\Delta_{kk|\cA}\leq1/4$, \ref{CF4}, and the Cauchy--Schwarz inequality as in the preceding proof, we get
\begin{align*}
n_v|R_v|
&\leq\frac{2}{\lambda^2\lambda_c^2}
\E_p\left[\max_{k\in U_v}|\pi_{k|\cA}-\pi_k|\right]\times\left[
\frac{n_v}{4N_v^2}\sum_{k\in U_v}\widetilde e_k^2
+\frac{C_{\Delta 2}}{N_v^2}\sum_{k\in U_v}\sum_{\substack{l\in U_v\\l\neq k}}|\widetilde e_k\widetilde e_l|
\right]\\
&\leq C\left(\dfrac{1}{N_v}\sum_{k\in U_v}\widetilde{e}_k^2\right)
\E_p\left[\max_{k\in U_v}|\pi_{k|\cA}-\pi_k|\right],
\end{align*}
for a constant $C$ independent of $v$. By \ref{CF3},
$$
\E_p\left[\max_{k\in U_v}|\pi_{k|\cA}-\pi_k|\right]
\leq\left\{\E_p\left[\max_{k\in U_v}(\pi_{k|\cA}-\pi_k)^2\right]\right\}^{1/2}
=\mathcal{O}(n_v^{-1/2}).
$$
Thus, \ref{M2a} gives $n_vR_v=o(1)$. The leading term $n_v\E_p(B_\cA^2)$ is bounded, since
$$
n_v\E_p(B_\cA^2)
\leq\dfrac{n_v}{\lambda^2}\left(\dfrac{1}{N_v}\sum_{k\in U_v}\widetilde{e}_k^2\right)
\E_p\left[\max_{k\in U_v}(\pi_{k|\cA}-\pi_k)^2\right]
=\mathcal{O}(1).
$$

It remains to relate $B_\cA$ to the difference between the estimators. Under \ref{CF2c}, the second-moment bound in the preceding proof holds on the whole sample space. Hence, by \ref{CF3} and \ref{M2a},
\begin{align*}
&n_v\E_p\left[\left\{\widehat{\mu}_{\mathrm{ma}}(\widetilde{m},\pi)-\widehat{\mu}_{\mathrm{ma}}(\widetilde{m},\pi_{|\cA})-B_\cA\right\}^2\right]\\
&\quad\leq C\left(\frac{1}{N_v}\sum_{k\in U_v}\widetilde e_k^2\right)
\E_p\left[\max_{k\in U_v}(\pi_{k|\cA}-\pi_k)^2\right]
=\mathcal{O}(n_v^{-1})=o(1).
\end{align*}
The quantity inside braces has conditional mean zero given $\cA$. Since $B_\cA$ is determined by $\cA$, the tower property gives
\begin{align*}
&\E_p\left[B_\cA\left\{\widehat{\mu}_{\mathrm{ma}}(\widetilde m,\pi)
-\widehat{\mu}_{\mathrm{ma}}(\widetilde m,\pi_{|\cA})-B_\cA\right\}\right]\\
&\quad=\E_p\left[B_\cA\E_p\left[
\widehat{\mu}_{\mathrm{ma}}(\widetilde m,\pi)
-\widehat{\mu}_{\mathrm{ma}}(\widetilde m,\pi_{|\cA})-B_\cA\mid\cA\right]\right]=0.
\end{align*}
Together with \eqref{dec1}, the following equality proves the result:
\begin{align*}
&n_v\E_p\left[\left\{\widehat{\mu}_{\mathrm{ma}}(\widetilde{m},\pi)-\widehat{\mu}_{\mathrm{ma}}(\widetilde{m},\pi_{|\cA})\right\}^2\right]\\
&\quad=n_v\E_p(B_\cA^2)
+n_v\E_p\left[\left\{\widehat{\mu}_{\mathrm{ma}}(\widetilde{m},\pi)-\widehat{\mu}_{\mathrm{ma}}(\widetilde{m},\pi_{|\cA})-B_\cA\right\}^2\right]\\
&\quad=n_v\E_p(B_\cA^2)+o(1).
\end{align*}

\subsection{Failure of first-order equivalence}
\label{supp:nonequivalence}
Let $p$ be SRSWOR of size $n_v$, with $n_v/N_v\xrightarrow[v\to\infty]{}\pi_\star\in(0,1)$. Assume that $N_v$ is even and fix $M=2$ population folds of sizes $N_1=N_2=N_v/2$ before sampling. With $\cA=(\boldsymbol{n},\boldsymbol{\delta})$, we get
$$\pi_{k|\cA}=\dfrac{2n_j}{N_v},\qquad k\in U_j,$$
while $\pi_k=n_v/N_v$. Suppose that $\widehat{m}^{(-k)}=\widetilde m=0$ for every $k$, and that $y_k=1$ on $U_1$ and $y_k=-1$ on $U_2$, so that $\mu=0$. Then,
$$\widehat{\mu}_{cf}(\widehat m,\pi)=\dfrac{2n_1-n_v}{n_v},\qquad
\widehat{\mu}_{cf}(\widehat m,\pi_{|\cA})=0\quad\text{if }n_1n_2>0.$$
The previous event has probability tending to one. Indeed, for $n_v\leq N_v/2$,
$$\P(n_1n_2=0)
=2\frac{\binom{N_v/2}{n_v}}{\binom{N_v}{n_v}}
=2\prod_{r=0}^{n_v-1}\frac{N_v/2-r}{N_v-r}
\leq2^{1-n_v},$$
and this probability is zero if $n_v>N_v/2$.

Conditional on the fold counts, the samples in the two folds are independent SRSWOR samples. The count $n_1$ is hypergeometric, with
$$\E_p(n_1)=\frac{n_v}{2},\qquad
\V_p(n_1)=\frac{n_v(N_v-n_v)}{4(N_v-1)}.$$
For every $\epsilon>0$, Chebyshev's inequality gives
$$\P\left(\left|\frac{n_1}{n_v}-\frac12\right|>\epsilon\right)
\leq\frac{N_v-n_v}{4n_v(N_v-1)\epsilon^2}
\xrightarrow[v\to\infty]{}0.$$
Since $n_2=n_v-n_1$, it follows that $2n_j/N_v\xrightarrow[v\to\infty]{\P}\pi_\star$ for $j=1,2$. The second-order conditional inclusion probabilities are
$$\frac{n_j(n_j-1)}{(N_v/2)(N_v/2-1)}
\quad\text{within fold }j,\qquad
\frac{4n_1n_2}{N_v^2}
\quad\text{across folds}.$$
Both converge in probability to $\pi_\star^2$, giving \ref{CF2a}. Moreover,
\begin{align*}
\E_p\left[\max_{k\in U_v}(\pi_{k|\cA}-\pi_k)^2\right]
&=\frac{4}{N_v^2}\E_p\left[\left(n_1-\frac{n_v}{2}\right)^2\right]\\
&=\frac{4}{N_v^2}\V_p(n_1)
=\frac{n_v(N_v-n_v)}{N_v^2(N_v-1)}
=\mathcal{O}(n_v^{-1}),
\end{align*}
which verifies \ref{CF3}. For distinct $k,l\in U_j$,
$$\Delta_{kl|\cA}
=\frac{n_j(n_j-1)}{(N_v/2)(N_v/2-1)}-\frac{4n_j^2}{N_v^2}
=-\frac{\pi_{k|\cA}(1-\pi_{k|\cA})}{N_v/2-1}.$$
Across folds, the conditional covariance is zero. Therefore,
$$n_v\max_{k\neq l\in U_v}|\Delta_{kl|\cA}|
\leq\frac{n_v}{4(N_v/2-1)}=\mathcal{O}(1),$$
which verifies \ref{CF4}.

Here $B_\cA=(2n_1-n_v)/n_v$, so
$$\sqrt{n_v}B_\cA
=\sqrt{\frac{N_v-n_v}{N_v-1}} 
\frac{n_1-n_v/2}{\sqrt{\V_p(n_1)}}.$$
The hypergeometric central limit theorem \citep{lahiri2006sub} and $(N_v-n_v)/(N_v-1)\xrightarrow[v\to\infty]{}1-\pi_\star$ give
$$\sqrt{n_v}B_\cA\xrightarrow[v\to\infty]{\cL}\cN(0,1-\pi_\star).$$
Thus, first-order equivalence fails. The fold residual means differ.

\subsection{First-order equivalence under balanced fold residuals}
\label{supp:balanced-residuals}
Suppose that the assumptions of Theorem \ref{Thm:reminder} hold. Consider SRSWOR with $n_v/N_v\xrightarrow[v\to\infty]{}\pi_\star\in(0,1)$, fix the population folds before sampling, and let $\cA=(\boldsymbol{n},\boldsymbol{\delta})$. Define
$$\overline e=\dfrac{1}{N_v}\sum_{k\in U_v}\widetilde e_k,\qquad
\overline e_j=\dfrac{1}{N_{j,v}}\sum_{k\in U_{j,v}}\widetilde e_k,\quad j=1,\ldots,M.$$
Then,
$$\lim_{v\to\infty}n_v\left\{\V_p\left(\widehat{\mu}_{\mathrm{ma}}(\widetilde m,\pi)\right)
-\V_p\left(\widehat{\mu}_{\mathrm{ma}}(\widetilde m,\pi_{|\cA})\right)\right\}=0$$
if and only if
$$\lim_{v\to\infty}|\overline e_j-\overline e|=0,\qquad j=1,\ldots,M.$$

\begin{proof}
By \eqref{dec1} and $n_vR_v=o(1)$, it is enough to determine when $n_v\E_p(B_\cA^2)$ tends to zero.
Since $\pi_{k|\cA}=n_j/N_{j,v}$ for $k\in U_{j,v}$, the conditional bias is
$$B_\cA=\sum_{j=1}^M\left(\dfrac{n_j}{n_v}-\dfrac{N_{j,v}}{N_v}\right)\overline e_j
=\sum_{j=1}^M\left(\dfrac{n_j}{n_v}-\dfrac{N_{j,v}}{N_v}\right)(\overline e_j-\overline e).$$
To compute its second moment, write $n_j=\sum_{k\in U_{j,v}}I_k$. Under SRSWOR,
$$\Delta_{kk}=\frac{n_v(N_v-n_v)}{N_v^2},\qquad
\Delta_{kq}=-\frac{n_v(N_v-n_v)}{N_v^2(N_v-1)},\quad k\neq q.$$
Summing these covariances over $U_{j,v}$ and $U_{l,v}$ gives
\begin{align*}
\operatorname{Cov}_p(n_j,n_l)
&=\sum_{k\in U_{j,v}}\sum_{q\in U_{l,v}}\Delta_{kq}\\
&=n_v\frac{N_v-n_v}{N_v-1}
\left\{\frac{N_{j,v}}{N_v}\mathds{1}_{\{j=l\}}
-\frac{N_{j,v}N_{l,v}}{N_v^2}\right\}.
\end{align*}
Also, $\E_p(n_j)=n_vN_{j,v}/N_v$, so $\E_p(B_\cA)=0$. Hence,
\begin{align*}
n_v\E_p(B_\cA^2)
&=\frac{1}{n_v}\sum_{j=1}^M\sum_{l=1}^M
\operatorname{Cov}_p(n_j,n_l)
(\overline e_j-\overline e)(\overline e_l-\overline e)\\
&=\frac{N_v-n_v}{N_v-1}\left[
\sum_{j=1}^M\frac{N_{j,v}}{N_v}(\overline e_j-\overline e)^2
-\left\{\sum_{j=1}^M\frac{N_{j,v}}{N_v}(\overline e_j-\overline e)\right\}^2
\right]\\
&=\frac{N_v-n_v}{N_v-1}\sum_{j=1}^M\frac{N_{j,v}}{N_v}(\overline e_j-\overline e)^2,
\end{align*}
where the last equality uses $\sum_{j=1}^M(N_{j,v}/N_v)\overline e_j=\overline e$.
Combining this identity with \eqref{dec1} and $n_vR_v=o(1)$ yields
\begin{align*}
&n_v\left\{\V_p\left(\widehat{\mu}_{\mathrm{ma}}(\widetilde m,\pi)\right)
-\V_p\left(\widehat{\mu}_{\mathrm{ma}}(\widetilde m,\pi_{|\cA})\right)\right\}\\
&\qquad=\dfrac{N_v-n_v}{N_v-1}\sum_{j=1}^M\dfrac{N_{j,v}}{N_v}(\overline e_j-\overline e)^2+o(1).
\end{align*}
The factoring term tends to $1-\pi_\star>0$. Moreover, \ref{CF2c} requires $N_{j,v}>N_v-n_v$ for every fold, since otherwise that fold could contain no sampled unit. Thus, the fold proportions are bounded away from zero. As $M$ is fixed, the result follows.
\end{proof}

{
\begingroup
\setstretch{1.1}
\setlength{\bibsep}{5pt}
\setlength{\parskip}{0pt}
\filterreferences
\repeatbibliography
\endgroup
}
\endgroup
\end{document}